\documentclass[12pt]{article}
\usepackage[margin=.88in]{geometry}

\usepackage{multirow}
\usepackage{amsmath, amssymb,amsthm} %amsfonts
\usepackage{bm}
\usepackage[pdftex]{graphicx}
\usepackage{algorithm, algorithmic}
\usepackage{comment}
\usepackage{caption}
\usepackage{subcaption}
\usepackage{enumitem}
\usepackage[sort]{natbib}
\usepackage{diagbox}

\usepackage{hyperref}
\hypersetup{
	colorlinks=true,
	linkcolor=blue,
	filecolor=blue,      
	urlcolor=red,
	citecolor=blue,
}

\DeclareMathAlphabet\mathbfcal{OMS}{cmsy}{b}{n}

\newcommand{\argmax}{\operatornamewithlimits{argmax}}
\newcommand{\argmin}{\operatornamewithlimits{argmin}}

\newcommand{\be}{\begin{equation}}
	\newcommand{\ee}{\end{equation}}
 \newcommand{\bs}{\begin{split}}
	\newcommand{\es}{\end{split}}
\newcommand{\bea}{\begin{eqnarray}}
	\newcommand{\eea}{\end{eqnarray}}
\newcommand{\beas}{\begin{eqnarray*}}
	\newcommand{\eeas}{\end{eqnarray*}}

\usepackage{color}

\newcommand{\E}{\mathbb{E}}
\newcommand{\PP}{\mathbb{P}}
\newcommand{\R}{\mathbb{R}}

\newcommand{\bI}{\mathbf{I}}
\newcommand{\bi}{\mathbf{i}}

\newcommand{\bcA}{{\mathbfcal{A}}}
\newcommand{\bcM}{{\mathbfcal{M}}}

\newcommand{\bcE}{{\mathbfcal{E}}}

\newcommand{\bcK}{{\mathbfcal{K}}}

\newtheorem*{theorem*}{Theorem}
\newtheorem{theorem}{Theorem}

\newtheorem{lemma}{Lemma}

\newtheorem{definition}{Definition}
\newtheorem{remark}{Remark}

\usepackage{setspace}
\title{Bandable Cumulant Tensors: Optimal Estimation and Applications in Non-Gaussian Data Modeling}

\author{Runshi Tang\footnote{Department of Statistics, Columbia University}, ~ 
Anru R. Zhang\footnote{Department of Biostatistics \& Bioinformatics and Department of Computer Science, Duke University}, ~
Yuefeng Han\footnote{Department of Applied and Computational Mathematics and Statistics, University of Notre Dame}, ~
and ~ 
Wei Biao Wu\footnote{Department of Statistics, University of Chicago}}

\date{}

\begin{document}

\maketitle

\begin{abstract}
Higher-order cumulants capture the non-Gaussian dependence that covariance misses, but they are hard to use in high dimensions. An order-$d$ cumulant tensor has $p^d$ entries, and the plug-in sample cumulant is generally not even rate-optimal under the tensor spectral norm. For ordered data, both difficulties admit one remedy: assuming that higher-order interactions decay away from the main tensor diagonal, we introduce a bandable cumulant class and a tapered sample cumulant estimator that computes only $O(pk^{d-1})$
local entries at bandwidth $k$ and never forms the full tensor. Under exponential-type tail conditions, we prove nonasymptotic spectral-norm bounds that separate tapering bias from stochastic error, and a minimax lower bound over the same class that matches the leading bias and stochastic terms of the upper bound; for sub-Gaussian observations, the tapered estimator attains the minimax rate whenever $n\gg(k+\log p)^{d-1}$ at the oracle bandwidth $k$, with the ambient dimension entering only through $\log p$. Localization also suppresses the higher-order fluctuations behind this suboptimality, so tapering plays a stronger role here than in bandable covariance estimation. The spectral-norm guarantee transfers directly to downstream tasks, yielding plug-in error bounds for cumulant Yule--Walker estimation in autoregressive models, minimum-distance estimation in moving-average models, and matched-filter source localization in sensor arrays. Simulations corroborate the theory, and real-data analyses of RR-interval, air-quality, and Neuropixels recordings illustrate the
resulting stability gains.
\end{abstract}

\section{Introduction}

% Covariance matrices play a central role in multivariate statistics. They
% summarize second-order dependence and provide the foundation for principal
% component analysis, linear prediction, spectral analysis, and many other
% procedures
% \citep{anderson2003multivariate,muirhead1982aspects,jolliffe2002principal}.
% For Gaussian data, this second-order description is complete, since the
% multivariate normal distribution is determined by its mean vector and
% covariance matrix
% \citep{anderson2003multivariate,muirhead1982aspects}. Outside the Gaussian
% setting, however, covariance can miss important structure. Features such as
% asymmetry, phase dependence, nonlinear interaction, and latent non-Gaussian
% dependence are naturally higher-order and have long motivated the use of
% cumulants and higher-order spectra in statistics and signal processing
% \citep{brillinger1981time,mccullagh1987tensor,mendel1991tutorial,nikias1993higher}.
For Gaussian data, second-order analysis is complete: the mean vector and
covariance matrix determine the distribution, so covariance-based methods
such as principal component analysis, linear prediction, and spectral
analysis lose nothing
\citep{anderson2003multivariate,muirhead1982aspects,jolliffe2002principal}.
Real data rarely are, and covariance then misses the structure of interest.
Asymmetry, phase dependence, nonlinear interaction, and latent non-Gaussian
dependence are intrinsically higher-order, and have long motivated cumulants
and higher-order spectra in statistics and signal processing
\citep{brillinger1981time,mccullagh1987tensor,mendel1991tutorial,nikias1993higher}.
These ideas are also central in independent component analysis and blind source
separation, where non-Gaussian higher-order information can identify latent
structure not visible from covariance alone
\citep{comon1994independent,hyvarinen2001independent,nordhausen2018independent}.

Cumulants are particularly useful because they remove lower-order contributions
from moments and isolate joint dependence at a specified order. A key feature is
that cumulants of order $d\ge3$ vanish for Gaussian random variables.
Consequently, higher-order cumulants are sensitive to non-Gaussian signal
structure but insensitive to independent Gaussian contamination. This built-in immunity to Gaussian noise drives higher-order spectral
analysis, nonlinear dependence diagnostics, and non-Gaussian time-series
modeling
\citep{nikias1987bispectrum,mendel1991tutorial,nikias1993signal,zhang2018asymptotic,giannakis1990identifiability,swami1990linear,gospodinov2015minimum}.

The difficulty is that higher-order cumulants are not only much larger than
covariance matrices, but also statistically more delicate to estimate. For a
random vector in $\mathbb R^p$, the order-$d$ cumulant is a tensor with $p^d$
entries. More importantly, unlike the empirical covariance, the sample cumulant
is a nonlinear polynomial in empirical moments of multiple orders. For
$d\ge3$, the resulting plug-in estimator is generally not rate-optimal under
tensor spectral norm: higher-order empirical fluctuations can dominate the
optimal stochastic scale in the unstructured problem \citep{tang2026detection}.
Thus, the statistical question is not only how to estimate a large tensor, but
how to regularize it in a way that suppresses unstable global fluctuations while
respecting the structure of the underlying dependence.

% This paper develops a bandability framework for higher-order cumulant tensors.
% The idea is motivated by the literature on bandable covariance matrices, where
% entries far from the main diagonal are assumed to decay and tapering or banding
% can improve estimation under spectral norm
% \citep{cai2012minimax,cai2016structured}. We extend this principle from
% second-order dependence to higher-order cumulants. The guiding assumption is
% that, for ordered data, the strongest higher-order interactions are local. Such
% locality is natural for temporal lags, spatial sensor arrays,
% frequency-indexed measurements, genomic positions, and other sequence-ordered
% features, where nearby coordinates often represent related scientific units and
% long-range higher-order interactions are expected to be weaker.
This paper develops a bandability framework for higher-order cumulant
tensors, built on a simple premise: when coordinates are ordered by time,
space, frequency, or genomic position, nearby coordinates represent related
scientific units and the strongest higher-order interactions are local. The
second-order version of this premise is the classical bandable covariance
model, in which off-diagonal decay allows tapering to achieve the optimal
spectral-norm rate \citep{cai2012minimax,cai2016structured}; we extend the
principle to higher-order cumulants.

To formalize this locality, for a multi-index $(i_1,\ldots,i_d)$ we measure
its distance from the main tensor diagonal by
\[
    \operatorname{diam}(i_1,\ldots,i_d)
    =
    \max_{a,b\in[d]} |i_a-i_b|.
\]
A cumulant tensor is called bandable if its entries decay as this diameter
grows. The tensor setting, however, has a different geometry. A diagonal tube of
width $k$ contains $O(pk^{d-1})$ entries, rather than $O(pk)$ as in the matrix
case, leading to a different approximation bias and bias--variance tradeoff.

To estimate such tensors, we introduce a tapered sample cumulant estimator that
preserves entries near the main tensor diagonal and downweights or removes
entries farther away. This localization reduces the variance from noisy
long-range cumulants and lowers the computational cost: for fixed order $d$ and
bandwidth $k$, only $O(pk^{d-1})$ entries need to be computed and stored. The statistical payoff, however, goes beyond this reduction in size and variance. In the unstructured higher-order problem, the full
plug-in sample cumulant contains higher-order stochastic terms that can make it
rate-suboptimal. Localization replaces these global fluctuations, governed by
the ambient dimension $p$, with local fluctuations governed by the bandwidth
$k$. Tapering therefore does more here than in bandable covariance estimation,
where it is purely a bias--variance device: at the right bandwidth and in
the regimes identified below, it turns the simple plug-in estimator into a
minimax rate-optimal one.

The framework also leads naturally to regularized plug-in procedures for
downstream statistical tasks. In autoregressive models, tapered cumulants
regularize higher-order Yule--Walker equations. In moving-average models, the
cumulant tensor is exactly banded, leading to a regularized minimum-distance
estimator based on non-Gaussian identifying equations. In spatial sensor
arrays, higher-order cumulants remove additive Gaussian noise and yield a
tapered matched-filter score for localizing non-Gaussian sources. These
examples illustrate the broader principle that existing cumulant-based
procedures can often be stabilized by replacing raw empirical cumulants with
their tapered counterparts.

\paragraph{Main Contributions.}
The paper makes three main contributions.

First, we introduce the bandable cumulant tensor class and the tapered
sample cumulant estimator, prove nonasymptotic spectral-norm upper bounds
with an explicit decomposition into tapering bias and stochastic error, and
establish a minimax lower bound that matches the leading bias and stochastic
terms of the upper bound. The analysis pinpoints where the higher-order
problem departs from covariance estimation: localization reduces the
higher-order empirical fluctuations that make the unstructured plug-in
sample cumulant rate-suboptimal. For sub-Gaussian observations and a
bandwidth $k$ of the oracle order, these fluctuations are negligible once
$ n\gg (k+\log p)^{d-1},$ and the simple tapered plug-in estimator attains the structured minimax
rate.
To the best of our knowledge, this is the first minimax theory for
estimating higher-order cumulant tensors under structural decay
assumptions.

Second, we show that the tapered estimator serves as a regularized input to
downstream cumulant-based procedures: cumulant Yule--Walker estimation for
autoregressive models, minimum-distance estimation for moving-average
models, and matched-filter source localization in spatial sensor arrays. In
each case, under explicit stability conditions, the parameter or
localization error is controlled by the tensor spectral-norm error, so the
bias--variance guarantees for the tensor transfer directly to the
downstream task. Because cumulants of order $d\ge3$ vanish for Gaussian
variables, the estimating equations and the localization score are
unaffected by independent additive Gaussian measurement noise that biases
their second-order counterparts.

Third, we make the estimator practical. It runs in $O(npk^{d-1})$ time and
$O(pk^{d-1})$ memory without ever forming the full order-$d$ tensor, and a
Lepski-type stability rule selects the bandwidth from data. Simulations
corroborate the bias--variance theory, and three real-data analyses of RR
intervals, air-quality sensors, and paired patch-clamp/Neuropixels
waveforms show that tapered cumulants stabilize downstream estimating
equations, localize non-Gaussian sources more accurately than a
covariance-based score, and match the accuracy of the raw third-order score
at a small fraction of its computational and storage cost.

\paragraph{Organization.}
The rest of the paper is organized as follows.
Section~\ref{sec:bandable-cumulants-estimator} introduces cumulant tensors, the
bandable cumulant class, and the tapered sample cumulant estimator.
Section~\ref{sec:stat-guarantee} gives nonasymptotic upper bounds, minimax
lower bounds, and the resulting rate comparisons.
Section~\ref{sec:application} develops plug-in applications of tapered
cumulants. Section~\ref{sec:computation-tuning} discusses computation and
tuning. Section~\ref{sec:simulation} presents simulation studies, and
Section~\ref{sec:real-data} gives real-data illustrations. Proofs and
additional technical details are collected in the Supplement.

\subsection{Related Literature}

Moments and cumulants are classical objects in probability and statistics.
Moment-based estimation goes back at least to the method-of-moments tradition
\citep{pearson1894contributions}, while cumulants provide an algebraic way to
separate genuine higher-order dependence from lower-order moment
contributions. Their combinatorial structure is naturally described through
partitions and partition lattices
\citep{speed1983cumulants,mccullagh1987tensor}. In the univariate setting,
moments and cumulants have also been studied from the viewpoint of the moment
problem and distributional identifiability \citep{schmudgen2017moment}. Our
work uses these classical algebraic objects, but the main focus is on
high-dimensional random vectors, where the order-$d$ moment or cumulant becomes
a tensor.

The order-$2$ cumulant is the covariance matrix, and our work is closely
connected to the large literature on structured covariance estimation. For
high-dimensional random vectors, structure is essential for consistent
estimation under strong norms, and common examples include bandable, sparse,
low-rank, factor, spiked, Toeplitz, and other structured covariance models
\citep{bickel2008regularized,cai2010optimal,cai2012minimax,cai2016structured,koltchinskii2017concentration}.
For ordered variables and stationary processes, banding and tapering are
natural regularization tools because nearby coordinates or lags are expected to
be more strongly dependent than distant ones
\citep{wu2009banding,cai2016structured}. For matrix- and tensor-valued
observations, another major line of work exploits product or additive multiway
structure. Examples include separable matrix-normal covariance models,
Kronecker-product covariance models, regularized transposable covariance
models, sums of Kronecker products, Kronecker-sum decompositions, and sparse
graphical models for matrix-variate or tensor-valued data
\citep{dutilleul1999mle,werner2008estimation,allen2010transposable,tsiligkaridis2013covariance,greenewald2013kronecker,zhou2014gemini,greenewald2015robust,greenewald2019tensor,wang2022kronecker}.
More recent work has also considered approximately separable matrix-variate
covariance through core shrinkage, tensor-train or dimension-free structured
covariance estimation, and mode-wise spiked covariance models
\citep{hoff2023core,puchkin2024dimension,patarasau2025structured,tang2025mode}.
The present paper is closest in spirit to bandable covariance estimation, but
the higher-order setting changes both the geometry and the stochastic behavior:
a diagonal tube in an order-$d$ tensor has width-dependent size
$O(pk^{d-1})$, and the empirical cumulant contains nonlinear lower-order moment
correction terms absent from covariance estimation.

Higher-order moment and cumulant tensors of random vectors have appeared in
several areas. In signal processing and time-series analysis, higher-order
cumulants and spectra are used to describe nonlinear and non-Gaussian
dependence beyond covariance
\citep{brillinger1981time,mendel1991tutorial,nikias1987bispectrum,nikias1993higher,zhang2018asymptotic}.
Independent component analysis provides another important setting where
non-Gaussian higher-order information is used for identifiability and
estimation
\citep{comon1994independent,cardoso1999high,hyvarinen2001independent,nordhausen2018independent,auddy2025large}.
Moment tensors are also a key tool in learning latent variable models, where
low-order observable moments are decomposed to recover model parameters
\citep{anandkumar2012method,anandkumar2014tensor}. More directly related to
the present paper, \cite{tang2026detection} studies unstructured estimation and
detection of high-order moment and cumulant tensors under tensor spectral norm.

This work is also related to structured tensor decomposition. Tensor
decomposition methods, including Tucker and CP decompositions, provide
low-dimensional representations of multiway arrays and have become standard
tools for high-dimensional tensor data analysis
\citep{delathauwer2000multilinear,kolda2009tensor}. Because many tensor
optimization problems are computationally hard in general
\citep{hillar2013most}, a substantial body of work studies special structures
or statistical models that make estimation and computation tractable. Examples
include tensor methods for latent variable models \citep{anandkumar2014tensor},
tensor SVD and low-rank tensor estimation \citep{zhang2018tensor},
tensor-on-tensor regression and related structured tensor models
\citep{luo2024tensor}, structured CP/PARAFAC decomposition with sparse or
incoherent factors \citep{rambhatla2020provable}, and recent nonasymptotic
theory for CP decomposition and alternating least squares under
signal-plus-noise models \citep{tang2025revisit}.

\subsection{Notation}

For a positive integer $m$, write $[m]=\{1,\ldots,m\}$. For vectors,
$\|\cdot\|_2$ denotes the Euclidean norm, and
$\mathbb S^{p-1}=\{u\in\mathbb R^p:\|u\|_2=1\}$ denotes the unit sphere in
$\mathbb R^p$. For matrices, $\|\cdot\|$ and $\|\cdot\|_F$ denote the operator
norm and Frobenius norm, respectively. For a tensor $\mathcal A$,
$\mathcal A_{\mathbf i}$ denotes the entry indexed by the multi-index
$\mathbf i$. We use $\langle\cdot,\cdot\rangle$ for the usual Euclidean inner
product, with the natural entrywise extension to matrices and tensors.

For vectors $u_1,\ldots,u_d\in\mathbb R^p$, let
$u_1\otimes\cdots\otimes u_d$ denote their tensor product. In particular, for
$x\in\mathbb R^p$, $x^{\otimes d}$ denotes the order-$d$ tensor with entries $(x^{\otimes d})_{i_1,\ldots,i_d}    =    x_{i_1}\cdots x_{i_d}$. For a tensor $\mathcal A\in(\mathbb R^p)^{\otimes d}$, its spectral norm is
defined as
\[
    \|\mathcal A\|
    =
    \sup_{u_1,\ldots,u_d\in\mathbb S^{p-1}}
    \left|
    \left\langle
    \mathcal A,
    u_1\otimes\cdots\otimes u_d
    \right\rangle
    \right|.
\]
When $d=2$, this definition coincides with the usual matrix operator norm.

For two tensors of the same size, $\mathcal A*\mathcal B$ denotes their
entrywise product. We write $\mathbf 1$ for the all-one tensor of the
appropriate dimension and $\mathbf I\{\cdot\}$ for the indicator function.
Constants $C,c,C_1,c_1,\ldots$ may change from line to line. Their dependence
on fixed quantities is indicated by subscripts when needed; for example,
$C_{\alpha,d}$ denotes a constant depending only on $\alpha$ and $d$. We write
$a\lesssim b$ if $a\le Cb$ for a universal constant $C$, and $a\asymp b$ if
both $a\lesssim b$ and $b\lesssim a$.

For a real-valued random variable $Z$ and $\gamma>0$, define the Orlicz norm
\[
    \|Z\|_{\Phi_\gamma}
    =
    \inf\left\{
    t>0:
    \mathbb E\exp\left(\frac{|Z|^\gamma}{t^\gamma}\right)\le 2
    \right\}.
\]
For a random vector $X\in\mathbb R^p$, the quantity $\sup_{u\in\mathbb S^{p-1}}\|u^\top X\|_{\Phi_\gamma}$ measures the uniform $\Phi_\gamma$ exponential-type tail size of one-dimensional projections of
$X$. All other symbols and model-specific notation are introduced where they first appear.

\section{Bandable Higher-Order Cumulant Tensors and the Tapered Estimator}
\label{sec:bandable-cumulants-estimator}

\subsection{Cumulant Tensors}
\label{subsec:cumulant-tensors}

Let $X=(X_1,\ldots,X_p)^\top\in\R^p$ be a random vector satisfying $\mathbb E\|X\|_2^d<\infty$.
For $1\le r\le d$, define the $r$-th moment tensor $\bcM_r=\E[X^{\otimes r}]\in(\R^p)^{\otimes r}$ by
\[
    (\bcM_r)_{i_1,\ldots,i_r}
    =
    \E[X_{i_1}\cdots X_{i_r}].
\]
We define cumulant tensors through the moment--cumulant formula.
Let $\mathcal P([r])$ be the set of all partitions of $[r]=\{1,\ldots,r\}$.
For a partition $\pi\in\mathcal P([r])$, write $|\pi|$ for the number of blocks in $\pi$.
For a block $B=\{j_1,\ldots,j_s\}\subset[r]$ with $j_1<\cdots<j_s$, set $i_B=(i_{j_1},\ldots,i_{j_s})$.
The order-$r$ cumulant tensor $\bcK_r\in(\R^p)^{\otimes r}$ is defined entrywise by
\begin{equation}\label{eq_mk_formula}
    (\bcK_r)_{i_1,\ldots,i_r}
    =
    \sum_{\pi\in\mathcal P([r])}
    \mu(\pi)
    \prod_{B\in\pi}
    (\bcM_{|B|})_{i_B},
    \qquad
    \mu(\pi)=(-1)^{|\pi|-1}(|\pi|-1)!.
\end{equation}
Equivalently, the moment tensors and cumulant tensors satisfy
\[
    (\bcM_r)_{i_1,\ldots,i_r}
    =
    \sum_{\pi\in\mathcal P([r])}
    \prod_{B\in\pi}
    (\bcK_{|B|})_{i_B}.
\]
These algebraic definitions agree with the derivatives of the cumulant generating function whenever the cumulant generating function exists in a neighborhood of the origin.
In particular, $\bcK_1=\E[X]$.
If $X$ is centered, then $\bcK_2=\bcM_2$ and $\bcK_3=\bcM_3$.

For fixed order $d$, separating the one-block partition from all remaining partitions gives tensor-valued polynomial identities
\begin{equation}
\label{eq_relation_cumulant_moment}
    \bcK_d
    =
    \bcM_d
    +
    F_d(\bcM_1,\ldots,\bcM_{d-1}),
    \qquad
    \bcM_d
    =
    \bcK_d
    +
    G_d(\bcK_1,\ldots,\bcK_{d-1}),
\end{equation}
where $F_d$ and $G_d$ depend only on $d$.
Thus the difference between the order-$d$ moment tensor and the order-$d$ cumulant tensor is determined entirely by lower-order moment or cumulant tensors.

Given i.i.d. observations $X_1,\ldots,X_n$, define the empirical moment tensors by
\[
    \widehat{\bcM}_r
    =
    \frac1n\sum_{\ell=1}^n X_\ell^{\otimes r},
    \qquad
    1\le r\le d.
\]
The order-$d$ sample cumulant tensor is the plug-in estimator obtained by replacing each population moment tensor in \eqref{eq_relation_cumulant_moment} by its empirical version:
\begin{equation}
\label{eq_sample_cumulant_tensor}
    \widehat{\bcK}_d
    =
    \widehat{\bcM}_d
    +
    F_d(\widehat{\bcM}_1,\ldots,\widehat{\bcM}_{d-1}).
\end{equation}
This plug-in definition is used throughout the paper.

\subsection{Bandable Cumulant Structure}
\label{subsec:bandable-cumulant-structure}

The bandability assumption requires an ordering of the coordinates.
Such an ordering may come from time, space, frequency, genomic location, a sequence index, or another scientifically meaningful arrangement of the variables.
For a multi-index $\bi=(i_1,\ldots,i_d)\in[p]^d$, define its diameter by
\[
    \operatorname{diam}(\bi)
    =
    \max_{a,b\in[d]} |i_a-i_b|.
\]
This quantity measures how far the tensor entry indexed by $\bi$ is from the main tensor diagonal $i_1=\cdots=i_d$.
Small diameter corresponds to a local higher-order interaction among nearby coordinates, while large diameter corresponds to a long-range interaction.

\begin{definition}[Bandable cumulant tensor]
For $\alpha>0$ and $\beta>0$, an order-$d$ cumulant tensor $\bcK_d$ is called $(\alpha,\beta)$-bandable if
\[
    |(\bcK_d)_{\bi}|
    \le
    \beta\left(1+\operatorname{diam}(\bi)\right)^{-(\alpha+d-1)}
\]
for every $\bi\in[p]^d$.
We denote the corresponding class by
\[
    \mathcal K_{\alpha,\beta}^d
    =
    \left\{
    \bcK_d\in(\R^p)^{\otimes d}:
    |(\bcK_d)_{\bi}|
    \le
    \beta\left(1+\operatorname{diam}(\bi)\right)^{-(\alpha+d-1)}
    \text{ for all } \bi\in[p]^d
    \right\}.
\]
\end{definition}

\begin{remark}[Normalization of the bandability exponent]
The exponent $\alpha+d-1$ accounts for the geometry of a diagonal tube in an order-$d$ tensor.
For fixed $d$, the number of indices satisfying $\operatorname{diam}(\bi)\le k$ is of order $pk^{d-1}$, since after one coordinate is fixed, the remaining $d-1$ coordinates may vary within a window of length proportional to $k$.
Thus the term $d-1$ compensates for the growth in the number of entries away from the main diagonal.
With this normalization, tapering at bandwidth $k$ leads to the spectral-norm approximation bias $\beta k^{-\alpha-d/2+1}$; see Section~\ref{sec:stat-guarantee}.
When $d=2$, the condition reduces to $|(\bcK_2)_{ij}|  \le \beta(1+|i-j|)^{-(\alpha+1)}$,
which is the usual bandable covariance condition \citep{cai2016structured,cai2012minimax}.
\end{remark}

\subsection{Banding and Tapering}
\label{subsec:tapered-estimator}

For a bandwidth $k\ge1$, define the hard-banding operator $\mathcal B_k$ by
\begin{equation}
\label{eq:banded-op}
    \mathcal B_k(\bcA)_{\bi}
    =
    \bcA_{\bi} \cdot \bI\{\operatorname{diam}(\bi)\le k\}
\end{equation}
for any order-$d$ tensor $\bcA$.
The corresponding hard-banded sample cumulant estimator is
\[
    \widehat{\bcK}_{d,B,k}
    =
    \mathcal B_k(\widehat{\bcK}_d).
\]
Hard banding keeps all empirical cumulants inside the diagonal tube of radius $k$ and sets all other entries to zero.

We now define the tapering estimator. For an integer $k\ge2$, let $h_k=\lfloor k/2\rfloor$ and define
\begin{equation}
\label{eq:taper-estimator-w}
    w_k(m)
    =
    \begin{cases}
    1, & 0\le m\le h_k,\\
    \dfrac{k-m}{k-h_k}, & h_k<m<k,\\
    0, & m\ge k.
    \end{cases}
\end{equation}
Thus $w_k(m)$ equals one near the main diagonal, decreases linearly in the transition region, and vanishes outside the diagonal tube of radius $k$.
Let $\mathcal W_k$ be the deterministic weight tensor
\[
    (\mathcal W_k)_{\bi}
    =
    w_k\!\left(\operatorname{diam}(\bi)\right).
\]
The tapered sample cumulant estimator is
\begin{equation}
\label{eq:tapered-cumulant-estimator}
    \widehat{\bcK}_{d,T,k}
    =
    \mathcal W_k*\widehat{\bcK}_d,
    \qquad
    (\widehat{\bcK}_{d,T,k})_{\bi}
    =
    w_k\!\left(\operatorname{diam}(\bi)\right)
    (\widehat{\bcK}_d)_{\bi}.
\end{equation}
Here $*$ denotes entrywise multiplication.
We visualize the raw cumulant tensor, the tapering function, and the tapered cumulants as in Figure \ref{fig:banding-tapering}. 
\begin{figure}[ht]
    \centering
    \includegraphics[width=\textwidth]{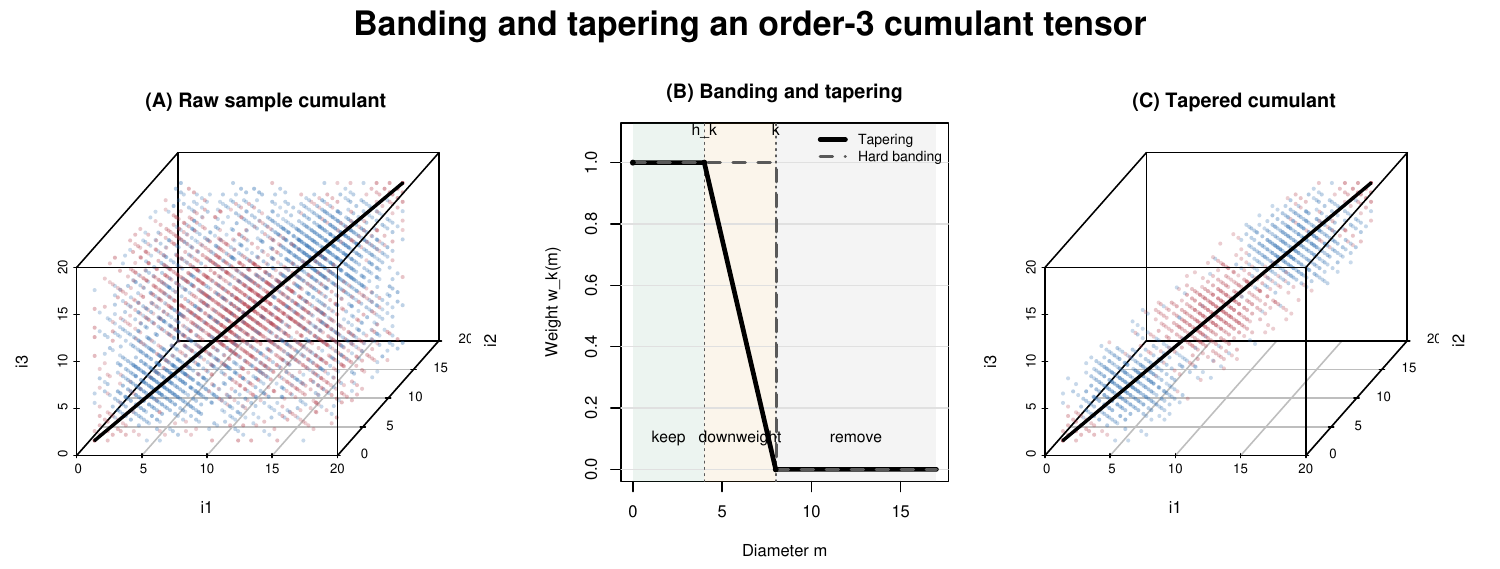}
    \caption{Illustration of banding and tapering for an order-$3$ cumulant tensor. }
    \label{fig:banding-tapering}
\end{figure}
The associated population tapered tensor is
\[
    \bcK_{d,T,k}
    =
    \mathcal W_k*\bcK_d.
\]
Therefore,
\[
    \widehat{\bcK}_{d,T,k}-\bcK_d
    =
    \mathcal W_k*(\widehat{\bcK}_d-\bcK_d)
    +
    (\mathcal W_k-\mathbf 1)*\bcK_d.
\]
The first term is the stochastic error after localization, while the second term is the deterministic approximation bias caused by tapering.
The upper bound in Section~\ref{sec:stat-guarantee} controls these two terms separately.

\section{Statistical Guarantees}
\label{sec:stat-guarantee}

This section gives nonasymptotic spectral-norm guarantees for the tapered cumulant estimator and establishes a minimax lower bound over the bandable cumulant class that matches the leading bias and stochastic terms of the upper bound.

\subsection{Upper Bound}
\label{subsec:upper-bound}

We first state an upper bound for a fixed bandwidth $k$.

\begin{theorem}[Upper bound for the tapered sample cumulant estimator]
\label{thm:tapered-cumulant-upper-bound}
Let $d\ge2$ be fixed, and let $X_1,\ldots,X_n$ be i.i.d. copies of a centered random vector $X\in\R^p$.
Let $\bcK_d=\bcK_d(X)$ be the order-$d$ cumulant tensor.
Suppose that $\bcK_d\in\mathcal K_{\alpha,\beta}^d$ and that, for some $0<\gamma\le d$, 
$
    \sup_{u\in\mathbb S^{p-1}}
    \|u^\top X\|_{\Phi_\gamma}
    \le
    K.
$
For $x\ge0$, define
\[
    \Delta_{k,x}
    =
    \sqrt{\frac{k+\log p+x}{n}}
    +
    \frac{(k+\log p+x)^{d/\gamma}}{n}.
\]
Then for every integer $2\le k\le p$ and every $x\ge0$, with probability at least $1-Ce^{-x}$,
\[
    \left\|
    \widehat{\bcK}_{d,T,k}
    -
    \bcK_d
    \right\|
    \le
    C_{\alpha,d}\beta k^{-\alpha-d/2+1}
    +
    C_{\gamma,d,K}
    \Delta_{k,x}
    (1+\Delta_{k,x})^{d-1}.
\]
Consequently,
\[
    \mathbb E
    \left\|
    \widehat{\bcK}_{d,T,k}
    -
    \bcK_d
    \right\|
    \le
    C_{\alpha,d}\beta k^{-\alpha-d/2+1}
    +
    C_{\gamma,d,K}
    \Delta_{k,0}
    (1+\Delta_{k,0})^{d-1}.
\]
\end{theorem}

The bound has two parts.
The first is the approximation bias term $C_{\alpha,d}\beta k^{-\alpha-d/2+1}$ that comes from the tail of the bandable cumulant tensor outside the central region where the taper equals one.
The second is the stochastic error from estimating the retained local cumulant entries, which depends on $k$ rather than $p$, apart from the logarithmic factor.

\subsection{Minimax Lower Bound}
\label{subsec:minimax-lower-bound}

We now state a minimax lower bound, which further shows that the bias--variance tradeoff in Theorem~\ref{thm:tapered-cumulant-upper-bound} is unavoidable.

\begin{theorem}[Lower bound for bandable cumulant tensors]
\label{thm_bandable_cumulant_lower_bound}
Let $d\ge3$ be fixed, and let 
\[
    \mathcal P_{\alpha,\beta}^d
    =
    \left\{
    X\in\R^p:
    \E X=0,\ 
    \|v^\top X\|_{\Phi_2}\le 2,\ \forall v\in\mathbb S^{p-1},
    \text{ and } \bcK_d(X)\in\mathcal K_{\alpha,\beta}^d
    \right\}
\]
be the sub-Gaussian bandable distributional class. There exist constants $k_0=k_0(d)\ge2$ and $c_d>0$, depending only on $d$, such that, for every $p\ge k_0$ and $n\ge1$,
\[
    \inf_{\widehat{\bcK}_d}
    \sup_{X\in\mathcal P_{\alpha,\beta}^d}
    \E
    \left\|
    \widehat{\bcK}_d-\bcK_d(X)
    \right\|
    \ge
    c_d
    \sup_{1\le k\le p}
    \min
    \left\{
    \beta k^{-\alpha-d/2+1},
    \sqrt{\frac{k+\log p}{n}},
    1
    \right\}.
\]
\end{theorem}

\begin{remark}[Lower-bound construction]
The main technical issue in the lower bound is to construct least favorable distributions whose order-$d$ cumulant perturbations are compatible with the bandable geometry. We use perturbations localized to diagonal tubes of width about $k$, allow these perturbations to occur at approximately $p/k$ disjoint locations, and calibrate their amplitudes according to the decay constraint. This gives an effective testing complexity of order
$
    k+\log(ep/k)
    \asymp
    k+\log p
$
and hence the stochastic scale $\sqrt{(k+\log p)/n}$, while the bandability constraint limits the signal size to $\beta k^{-\alpha-d/2+1}$. Optimizing over the tube width $k$ then gives the stated minimax lower bound. The construction is related to rank-one perturbation methods for unstructured high-order cumulant estimation \citep{tang2026detection}, but the localization, location packing, and bandwidth optimization are specific to the present bandable setting.
\end{remark}

The lower bound is stated over a sub-Gaussian distributional class whose order-$d$ cumulant tensor belongs to the same bandable class as in the upper bound. The next subsection compares the two bounds and derives the resulting minimax rate.

% Item 6: replace the complete Consequences and Comparisons subsection with the following.

\subsection{Consequences and Comparisons}
\label{subsec:rate-comparison}

\paragraph{Oracle bandwidth and minimax rate.}
Let $a=\alpha+\frac d2-1$ denote the exponent in the approximation term. Theorem~\ref{thm_bandable_cumulant_lower_bound} shows that the relevant bias--variance tradeoff at bandwidth $k$ is between the approximation scale $\beta k^{-a}$ and the leading stochastic scale $\sqrt{(k+\log p)/n}$. Accordingly, define the leading-order oracle bandwidth by
\[
    k_{\mathrm{or}}
    \in
    \argmin_{2\le k\le p}
    \left\{
    \beta k^{-a}
    +
    \sqrt{\frac{k+\log p}{n}}
    \right\}.
\]

Let $k_\star=(\beta^2n)^{1/(2\alpha+d-1)}$. If $k_0\le k_\star\le p$ and $\log p\lesssim k_\star$, then the logarithmic location term is dominated by the local block width, so $k_{\mathrm{or}}\asymp k_\star$. In this case,
\[
    \beta k_\star^{-a}
    \asymp
    \sqrt{\frac{k_\star}{n}}
    \asymp
    \beta^{1/(2\alpha+d-1)}
    n^{-(\alpha+d/2-1)/(2\alpha+d-1)}.
\]
Therefore, Theorem~\ref{thm_bandable_cumulant_lower_bound} gives
\[
    \inf_{\widehat{\bcK}_d}
    \sup_{X\in\mathcal P_{\alpha,\beta}^d}
    \mathbb E
    \left\|
    \widehat{\bcK}_d-\bcK_d(X)
    \right\|
    \gtrsim
    \min
    \left\{
    \sqrt{\frac{k_\star}{n}},
    1
    \right\}.
\]

Conversely, choosing $k\asymp k_\star$ in Theorem~\ref{thm:tapered-cumulant-upper-bound} yields, in the regime $\Delta_{k,0}\le1$,
\[
    \mathbb E
    \left\|
    \widehat{\bcK}_{d,T,k}
    -
    \bcK_d
    \right\|
    \lesssim
    \sqrt{\frac{k_\star}{n}}
    +
    \frac{(k_\star+\log p)^{d/\gamma}}{n}.
\]
Thus, if
\[
    \frac{(k_\star+\log p)^{d/\gamma}}{n}
    \lesssim
    \sqrt{\frac{k_\star}{n}},
\]
then the tapered plug-in estimator attains the structured minimax rate $\sqrt{k_\star/n}$. Since $\log p\lesssim k_\star$, a sufficient condition is $n\gg(k_\star+\log p)^{2d/\gamma-1}$. In particular, for sub-Gaussian observations with $\gamma=2$, this condition becomes $n\gg(k_\star+\log p)^{d-1}$.

\paragraph{Comparison with bandable covariance estimation.}
When $d=2$, the order-$2$ cumulant tensor is the covariance matrix, and the leading bias--variance tradeoff becomes
\[
    \beta k^{-\alpha}
    +
    \sqrt{\frac{k+\log p}{n}}.
\]
This agrees with the familiar tradeoff in bandable covariance estimation under spectral norm, where banding and tapering are known to be rate-optimal \citep{cai2010optimal,cai2012minimax,cai2016structured}. The term $k$ represents the complexity of estimation within a local matrix block, while $\log p$ accounts for the multiplicity of possible block locations. The higher-order case, however, is not a direct matrix analogue. A diagonal tube of width $k$ in an order-$d$ tensor contains $O(pk^{d-1})$ entries rather than $O(pk)$, which changes the approximation geometry and leads to the bias term $\beta k^{-\alpha-d/2+1}$. In addition, for $d\ge3$, the plug-in cumulant contains the higher-order fluctuation $(k+\log p)^{d/\gamma}/n$, which has no analogue in ordinary covariance estimation. Thus tapering plays a stronger role in the higher-order problem: it regularizes both the tensor geometry and the instability of empirical cumulants.

\paragraph{Comparison with unstructured high-order cumulant estimation.}
The bandable result should also be contrasted with the unstructured high-order cumulant problem. For fixed $d\ge3$, \cite{tang2026detection} show that, under sub-Gaussianity and without structural assumptions on $\bcK_d$, the minimax spectral-norm rate is $\sqrt{p/n}\wedge1$, while the raw plug-in sample cumulant is generally not rate-optimal because of higher-order empirical fluctuations. In the present bandable setting, the leading stochastic scale is reduced from $\sqrt{p/n}$ to $\sqrt{(k+\log p)/n}$. The bandwidth $k$ represents the effective local dimension, while $\log p$ is the unavoidable cost of identifying the location of the local perturbation. Consequently, bandability can yield a substantial improvement whenever $k+\log p\ll p$.

The comparison also clarifies the role of the higher-order plug-in term. In the unstructured setting, the raw sample cumulant contains a fluctuation of order $p^{d/\gamma}/n$, whereas tapering replaces this global term by $(k+\log p)^{d/\gamma}/n$. When this term is negligible relative to the leading bias and square-root stochastic terms, the tapered plug-in estimator attains the corresponding structured minimax rate. Outside that regime, the lower bound does not determine whether the higher-order term is information-theoretically unavoidable or specific to the plug-in estimator. This is the higher-order cumulant analogue of the classical principle that structural decay can permit estimation at a rate faster than the unrestricted high-dimensional benchmark \citep{koltchinskii2017concentration,cai2010optimal,cai2016structured}.

\section{Plug-in Applications of Tapered Cumulants}\label{sec:application}

\subsection{Cumulant Yule--Walker Estimation for Autoregressive Models}
\label{subsec:cumulant-yule-walker}

We first illustrate how the tapered cumulant estimator can regularize higher-order estimating equations for autoregressive models.
Ordinary Yule--Walker estimation identifies autoregressive coefficients through second-order autocovariance equations \citep{brockwell1991time,hamilton1994time}.
For non-Gaussian time series, higher-order cumulants provide additional estimating equations and are insensitive to independent additive Gaussian noise at orders $d\geq3$ \citep{brillinger1981time,mendel1991tutorial,nikias1993higher}.
The construction below replaces the covariance Yule--Walker equations by cumulant equations and estimates their entries using the tapered cumulant tensor, in the same spirit that banding regularizes covariance-based linear prediction for ordered data \citep{wu2009banding}.

Let $\{Y_t\}_{t\in\mathbb Z}$ be a centered stationary autoregressive process of order $r_0\in\mathbb N\cup\{\infty\}$ satisfying
\begin{equation}
\label{eq:ar-model}
    Y_t
    =
    \sum_{j=1}^{r_0}
    \phi_jY_{t-j}
    +
    \varepsilon_t.
\end{equation}
When $r_0=\infty$, assume that $\sum_{j\geq1}|\phi_j|<\infty$ and that the series in \eqref{eq:ar-model} converges in $L^d$.
When $r_0<\infty$, set $\phi_j=0$ for $j>r_0$.
Assume that the innovations $\{\varepsilon_t\}_{t\in\mathbb Z}$ are i.i.d., centered, satisfy $\mathbb E|\varepsilon_t|^d<\infty$, and that $\varepsilon_t$ is independent of the past $(Y_{t-1},Y_{t-2},\ldots)$.
Let $\tau_d$ denote the scalar order-$d$ cumulant of $\varepsilon_t$.

Assume further that the process admits the causal linear representation
\[
    Y_t
    =
    \sum_{\ell=0}^{\infty}
    \psi_\ell\varepsilon_{t-\ell},
    \qquad
    \psi_0=1,
\]
where, for some constants $C_\psi<\infty$ and $\rho\in(0,1)$,
\[
    |\psi_\ell|
    \leq
    C_\psi\rho^\ell,
    \qquad
    \ell\geq0.
\]
For finite-order autoregressive processes, this condition follows from the stability assumption that the autoregressive polynomial has no roots in the closed unit disk.
For an infinite-order process, it holds, for example, if the autoregressive power series extends to a zero-free disk of radius strictly larger than one \citep{brockwell1991time,chen2016stability}.
For notational convenience, set $\psi_\ell=0$ for $\ell<0$.

Fix $p\geq2$ and define the lag vector
\[
    Z_t
    =
    (Y_t,Y_{t-1},\ldots,Y_{t-p+1})^\top
    \in\R^p.
\]
Let $\bcK_d=\bcK_d(Z_t)\in(\R^p)^{\otimes d}$ be the order-$d$ cumulant tensor of $Z_t$.
For arbitrary nonnegative lags $h_0,\ldots,h_{d-1}$, define
\[
    \kappa_d(h_0,\ldots,h_{d-1})
    =
    \operatorname{cum}
    \left(
    Y_{t-h_0},
    \ldots,
    Y_{t-h_{d-1}}
    \right).
\]
Whenever $h_0,\ldots,h_{d-1}\in\{0,\ldots,p-1\}$, this cumulant is an entry of $\bcK_d$:
\[
    \kappa_d(h_0,\ldots,h_{d-1})
    =
    (\bcK_d)_{h_0+1,\ldots,h_{d-1}+1}.
\]
The causal representation implies that these cumulants are localized around the main tensor diagonal.

\begin{lemma}[Bandability of autoregressive cumulants]
\label{lem:ar-cumulant-bandability}
Under the assumptions above, for every $\alpha>0$, the cumulant tensor $\bcK_d$ belongs to $\mathcal K_{\alpha,\beta_{\mathrm{AR}}}^d$, where
\[
    \beta_{\mathrm{AR}}
    =
    \frac{
    |\tau_d|C_\psi^d
    }{
    1-\rho^d
    }
    \sup_{D\geq0}
    (1+D)^{\alpha+d-1}\rho^D
    <
    \infty.
\]
\end{lemma}

We next derive the cumulant Yule--Walker equations.
Let $h_1,\ldots,h_{d-1}\geq1$.
Since $Y_{t-h_1},\ldots,Y_{t-h_{d-1}}$ belong to the past relative to time $t$, the joint cumulant involving $\varepsilon_t$ and these variables vanishes.
Applying \eqref{eq:ar-model} inside the first argument of the cumulant and using multilinearity gives
\begin{equation}
\label{eq:ar-cumulant-yw-equation}
    \kappa_d(0,h_1,\ldots,h_{d-1})
    =
    \sum_{j=1}^{r_0}
    \phi_j
    \kappa_d(j,h_1,\ldots,h_{d-1}).
\end{equation}

Let $r\leq p-1$ be a finite working order, which need not equal the true order $r_0$ and may increase with the sample size, and define
\[
    \phi^{(r)}
    =
    (\phi_1,\ldots,\phi_r)^\top
    \in\R^r.
\]
Let
\[
    \mathcal H
    =
    \{h^{(1)},\ldots,h^{(m)}\}
    \subset
    \{1,\ldots,p-1\}^{d-1},
    \qquad
    h^{(a)}
    =
    (h_1^{(a)},\ldots,h_{d-1}^{(a)}).
\]
Define $b\in\R^m$, $A=A_r\in\R^{m\times r}$, and $q_r\in\R^m$ by
\[
    b_a
    =
    \kappa_d(0,h^{(a)}),
    \qquad
    A_{aj}
    =
    \kappa_d(j,h^{(a)}),
\]
and
\[
    (q_r)_a
    =
    \sum_{j>r}
    \phi_j
    \kappa_d(j,h^{(a)}),
    \qquad
    a=1,\ldots,m,
    \quad
    j=1,\ldots,r.
\]
Then \eqref{eq:ar-cumulant-yw-equation} gives
\begin{equation}
\label{eq:ar-truncated-linear-system}
    b
    =
    A\phi^{(r)}
    +
    q_r.
\end{equation}

We estimate $A$ and $b$ using selected entries of the tapered sample cumulant tensor.
Let $\widehat{\bcK}_{d,T,k}$ be the tapered estimator with bandwidth $k$, and define
\[
    \widehat\kappa_{d,k}(h_0,\ldots,h_{d-1})
    =
    (\widehat{\bcK}_{d,T,k})_{h_0+1,\ldots,h_{d-1}+1}.
\]
Set
\[
    \widehat b_a
    =
    \widehat\kappa_{d,k}(0,h^{(a)}),
    \qquad
    \widehat A_{aj}
    =
    \widehat\kappa_{d,k}(j,h^{(a)}),
    \qquad
    a=1,\ldots,m,
    \quad
    j=1,\ldots,r.
\]
The tapered cumulant Yule--Walker estimator is
\[
    \widehat\phi_{d,k}^{(r)}
    =
    (\widehat A^\top\widehat A)^{-1}
    \widehat A^\top\widehat b,
\]
whenever $\widehat A$ has full column rank.

\begin{theorem}[Growing-order tapered cumulant Yule--Walker estimation]
\label{thm:cumulant-yule-walker}
Let $d\geq3$ be fixed and suppose the autoregressive assumptions above hold.
Assume that $m\ge r$ and $s_{A,r}:=\sigma_{\min}(A)>0$, and write
\[
    \delta_k
    =
    \left\|
    \widehat{\bcK}_{d,T,k}
    -
    \bcK_d
    \right\|.
\]
Whenever
\begin{equation}
\label{eq:ar-yw-spectral-small-error}
    \sqrt m\,\delta_k
    \leq
    \frac{s_{A,r}}{2},
\end{equation}
the matrix $\widehat A$ has full column rank and
\begin{equation}
\label{eq:ar-yw-spectral-coefficient-bound}
    \left\|
    \widehat\phi_{d,k}^{(r)}
    -
    \phi^{(r)}
    \right\|_2
    \leq
    \frac{2}{s_{A,r}}
    \left[
    \sqrt m
    \left(
    1+\|\phi^{(r)}\|_2
    \right)
    \delta_k
    +
    \|q_r\|_2
    \right].
\end{equation}

Moreover, fix any $\alpha>0$, and let $\beta_{\mathrm{AR}}=\beta_{\mathrm{AR}}(\alpha)$ be the constant in Lemma~\ref{lem:ar-cumulant-bandability}. Suppose that the tapered estimator is computed from $n$ i.i.d. copies of the marginal distribution of the lag vector $Z_t$, that 
    $\sup_{u\in\mathbb S^{p-1}}
    \|u^\top Z_t\|_{\Phi_\gamma}
    \leq
    K$
for some $0<\gamma\leq d$, and $2\le k\le p$.
For $x\geq0$, define
\[
    \zeta
    =
    C_{\alpha,d}
    \beta_{\mathrm{AR}}
    k^{-\alpha-d/2+1}
    +
    C_{\gamma,d,K}
    \Delta_{k,x}
    (1+\Delta_{k,x})^{d-1},
\]
with $\Delta_{k,x}$ as in Theorem~\ref{thm:tapered-cumulant-upper-bound}. Then, with probability at least $1-Ce^{-x}$,
\[
    \left\|
    \widehat\phi_{d,k}^{(r)}
    -
    \phi^{(r)}
    \right\|_2
    \leq
    \frac{2}{s_{A,r}}
    \left[
    \sqrt m
    \left(
    1+\|\phi^{(r)}\|_2
    \right)
    \zeta
    +
    \|q_r\|_2
    \right],
\]
provided
\[
    \sqrt m\,
    \zeta
    \leq
    \frac{s_{A,r}}{2}.
\]
\end{theorem}

\begin{remark}[True and working orders]
The true order $r_0$ and the working order $r$ play different roles.
If $r_0<\infty$ and $r\geq r_0$, then $q_r=0$ and the theorem directly controls estimation of the full autoregressive coefficient vector.
If $r<r_0$, including the case $r_0=\infty$, then $q_r$ is the truncation remainder and must vanish sufficiently quickly relative to $s_{A,r}$ for consistent growing-order estimation.
\end{remark}

\begin{remark}[Spectral and entrywise control]
The deterministic result only requires a spectral-norm bound for the cumulant estimator.
For independent copies of $Z_t$, this bound follows from Theorem~\ref{thm:tapered-cumulant-upper-bound}, while a single stationary trajectory requires a dependent-data analogue.
Section~\ref{sec:supp-yw-spectral-entrywise} compares this approach with entrywise control of the raw sample cumulant.
Tensor spectral-norm control avoids the worst-case $\sqrt r$ conversion from entrywise error to operator-norm error in establishing rank stability.
\end{remark}

\subsection{Cumulant Minimum-Distance Estimation for Moving-Average Models}
\label{subsec:cumulant-ma}

We next consider moving-average models, where the bandable cumulant structure
is especially transparent. Moving-average processes are classical models for
short-memory dependence and are commonly estimated by likelihood,
quasi-likelihood, or second-order moment methods
\citep{brockwell1991time,hamilton1994time}. In non-Gaussian models, however,
higher-order cumulants provide additional identifying information. In
particular, second-order methods can suffer from the usual invertibility
ambiguity, whereas higher-order cumulants can distinguish moving-average
coefficient vectors that generate the same autocovariance function when the
innovations are non-Gaussian
\citep{giannakis1990identifiability,swami1990linear,gospodinov2015minimum}.
The tapered version below regularizes these empirical cumulants before they
enter the minimum-distance criterion.

Let $\{Y_t\}_{t\in\mathbb Z}$ be a centered stationary moving-average process
satisfying
\begin{equation}
\label{eq:ma-model}
    Y_t
    =
    \sum_{\ell=0}^q
    \theta_\ell \varepsilon_{t-\ell},
    \qquad
    \theta_0=1,
\end{equation}
where $\{\varepsilon_t\}_{t\in\mathbb Z}$ are i.i.d. centered innovations satisfying $\mathbb E|\varepsilon_t|^d<\infty$. Let $\tau_d$ denote the scalar order-$d$
cumulant of $\varepsilon_t$, and assume that $\tau_d\ne0$. For notational
convenience, set $\theta_j=0$ for $j\notin\{0,\ldots,q\}$.

Fix $p\ge q+1$ and define the lag vector
\[
    Z_t
    =
    (Y_t,Y_{t-1},\ldots,Y_{t-p+1})^\top
    \in\R^p.
\]
Let $\bcK_d=\bcK_d(Z_t)\in(\R^p)^{\otimes d}$ be the order-$d$
cumulant tensor of $Z_t$. By stationarity, this tensor does not depend on
$t$. For lags $h_0,\ldots,h_{d-1}\in\{0,\ldots,p-1\}$, define
\[
    \kappa_d(h_0,\ldots,h_{d-1})
    =
    (\bcK_d)_{h_0+1,\ldots,h_{d-1}+1}.
\]
Thus $\kappa_d(h_0,\ldots,h_{d-1})$ is the joint order-$d$ cumulant of
$Y_{t-h_0},\ldots,Y_{t-h_{d-1}}$.

The model-implied cumulants have a simple convolution form. Expanding each
lagged variable $Y_{t-h_\ell}$ in \eqref{eq:ma-model} and using multilinearity
of cumulants gives a sum over products of innovation coefficients and joint
cumulants of innovations. Since the innovations are independent, all mixed
innovation cumulants vanish unless the innovation time indices coincide. If the
common innovation is $\varepsilon_{t-a}$, then the coefficient in the
$\ell$-th factor is $\theta_{a-h_\ell}$. Therefore,
\begin{equation}
\label{eq:ma-cumulant-identity}
    \kappa_d(h_0,\ldots,h_{d-1})
    =
    \tau_d
    \sum_{a\in\mathbb Z}
    \prod_{\ell=0}^{d-1}
    \theta_{a-h_\ell}.
\end{equation}
Because $\theta_j=0$ outside $\{0,\ldots,q\}$, this identity also shows that
the cumulant tensor is exactly supported on a diagonal tube of width $q$.

\begin{lemma}[Exact banding of moving-average cumulants]
\label{lem:ma-cumulant-bandability}
Under the assumptions above, the cumulant tensor $\bcK_d$ is exactly banded
with bandwidth $q$. That is, if
$\operatorname{diam}(h_0,\ldots,h_{d-1})>q$, then
\[
    \kappa_d(h_0,\ldots,h_{d-1})=0.
\]
Moreover, for every $\alpha>0$,
$\bcK_d\in\mathcal K_{\alpha,\beta_{\mathrm{MA}}}^d$, where
\[
    \beta_{\mathrm{MA}}
    =
    |\tau_d|
    (q+1)
    (1+q)^{\alpha+d-1}
    (1\vee\|\theta\|_\infty)^d .
\]
\end{lemma}

We now form a minimum-distance estimator by matching selected empirical
cumulants to their model-implied values. Let
\[
    \mathcal H
    =
    \{h^{(1)},\ldots,h^{(m)}\}
    \subset
    \{0,\ldots,q\}^{d-1},
    \qquad
    h^{(a)}
    =
    (h_1^{(a)},\ldots,h_{d-1}^{(a)}).
\]
For each $h^{(a)}\in\mathcal H$, write
$\kappa_d(0,h^{(a)})=\kappa_d(0,h_1^{(a)},\ldots,h_{d-1}^{(a)})$ and define
$b\in\R^m$ by
\[
    b_a
    =
    \kappa_d(0,h^{(a)}),
    \qquad
    a=1,\ldots,m.
\]

For $\theta=(\theta_1,\ldots,\theta_q)^\top$, with $\theta_0=1$ and
$\theta_j=0$ for $j\notin\{0,\ldots,q\}$, define
\[
    \Psi_a(\theta)
    =
    \sum_{s\in\mathbb Z}
    \theta_s
    \theta_{s-h_1^{(a)}}
    \cdots
    \theta_{s-h_{d-1}^{(a)}},
    \qquad
    a=1,\ldots,m.
\]
Then \eqref{eq:ma-cumulant-identity} gives the population equations
\[
    b_a
    =
    \tau_d\Psi_a(\theta),
    \qquad
    a=1,\ldots,m.
\]
Equivalently, with
\[
    \eta
    =
    (\theta_1,\ldots,\theta_q,\tau_d)^\top
    \in\R^{q+1},
\]
define $M:\R^{q+1}\to\R^m$ by
\[
    (M(\eta))_a
    =
    \tau_d\Psi_a(\theta),
    \qquad
    a=1,\ldots,m.
\]
The population moment equation is
\[
    b=M(\eta).
\]

We estimate $b$ using selected entries of the tapered sample cumulant tensor.
Let $\widehat{\bcK}_{d,T,k}$ be the tapered estimator with bandwidth $k$, and
set
\[
    \widehat\kappa_{d,k}(h_0,\ldots,h_{d-1})
    =
    (\widehat{\bcK}_{d,T,k})_{h_0+1,\ldots,h_{d-1}+1}.
\]
For $h^{(a)}\in\mathcal H$, define
\[
    \widehat b_a
    =
    \widehat\kappa_{d,k}(0,h^{(a)}),
    \qquad
    a=1,\ldots,m.
\]
A tapered cumulant minimum-distance estimator minimizes
$\|\widehat b-M(\eta)\|_2^2$ over a chosen parameter space. We first analyze
a local estimator over a neighborhood of the true parameter; the corresponding
result for a global minimum-distance estimator is discussed in
Remark~\ref{rem:ma-local-global}.

The following theorem separates the deterministic stability of the
minimum-distance criterion from the probabilistic accuracy of the tapered
cumulant estimator. The first part bounds the parameter error directly by the
realized tensor estimation error, while the second part inserts the
high-probability bound from
Theorem~\ref{thm:tapered-cumulant-upper-bound}.

\begin{theorem}[Tapered cumulant minimum-distance estimation for moving-average models]
\label{thm:cumulant-ma}
Let $d\ge3$ be fixed. Let
$\eta_\star=(\theta_\star,\tau_{d,\star})^\top\in\R^{q+1}$ be the true
parameter, and suppose that $b=M(\eta_\star)$. Let $m\ge q+1$ and 
$J_\star=\nabla M(\eta_\star)\in\R^{m\times(q+1)}$ satisfy
\[
    \sigma_{\min}(J_\star)\ge s>0.
\]
Assume that there exist constants $\rho_0>0$ and $L>0$ such that $M$ is
continuously differentiable on
\[
    \mathcal N_{\rho_0}
    =
    \{\eta:\|\eta-\eta_\star\|_2\le\rho_0\}
\]
and
\[
    \|\nabla M(\eta)-J_\star\|
    \le
    L\|\eta-\eta_\star\|_2,
    \qquad
    \eta\in\mathcal N_{\rho_0}.
\]
Suppose that
\[
    \rho_0
    \le
    \frac{s}{2L},
\]
and let
\[
    \widehat\eta
    \in
    \argmin_{\eta\in\mathcal N_{\rho_0}}
    \|\widehat b-M(\eta)\|_2^2.
\]
Then the set of minimizers is nonempty, and every such minimizer satisfies
\[
    \|\widehat\eta-\eta_\star\|_2
    \le
    \frac{4\sqrt m}{s}
    \|\widehat{\bcK}_{d,T,k}-\bcK_d\|.
\]

Moreover, fix any $\alpha>0$, and let
$\beta_{\mathrm{MA}}=\beta_{\mathrm{MA}}(\alpha)$ be the constant in
Lemma~\ref{lem:ma-cumulant-bandability}. In the independent-replicate setting
where the observations are i.i.d. copies of the marginal distribution of the
lag vector $Z_t$ and $\sup_{u\in\mathbb S^{p-1}}  \|u^\top Z_t\|_{\Phi_\gamma}    \leq  K$
for some $0<\gamma\leq d$, for every integer
$2\le k\le p$ and every $x\ge0$, with probability at least $1-Ce^{-x}$,
\[
    \|\widehat\eta-\eta_\star\|_2
    \le
    \frac{4\sqrt m}{s}
    \left[
    C_{\alpha,d}
    \beta_{\mathrm{MA}}
    k^{-\alpha-d/2+1}
    +
    C_{\gamma,d,K}
    \Delta_{k,x}(1+\Delta_{k,x})^{d-1}
    \right].
\]
\end{theorem}

\begin{remark}[Smoothness of the moment map]
For fixed $q$, $m$, and $d$, each component of $M(\eta)$ is a polynomial in $(\theta,\tau_d)$. Hence $M$ is infinitely differentiable and its Jacobian is locally Lipschitz. In particular, on any compact neighborhood $\mathcal N_{\rho_0}$ of $\eta_\star$, there exists a constant $L<\infty$ such that
\[
    \|\nabla M(\eta)-\nabla M(\eta_\star)\|
    \le
    L\|\eta-\eta_\star\|_2,
    \qquad
    \eta\in\mathcal N_{\rho_0}.
\]
Thus, the smoothness condition in Theorem~\ref{thm:cumulant-ma} holds automatically. If $q$ increases with the sample size, however, the constant $L$ may depend on $q$.
\end{remark}

\begin{remark}[Exact banding versus polynomial bandability]
\label{rem:ma-exact-banding}
The moving-average model highlights the distinction between polynomial
bandability and exact banding. Since $\bcK_d$ is supported on the set
$\{\bi:\operatorname{diam}(\bi)\le q\}$, there is no approximation bias if the
bandwidth is chosen to retain this support. In particular, if the order $q$ is
known, one may use the hard-banded estimator
\[
    \widehat{\bcK}_{d,B,q}
    =
    \mathcal B_q(\widehat{\bcK}_d).
\]
Because $\mathcal B_q(\bcK_d)=\bcK_d$, the error is purely stochastic:
\[
    \widehat{\bcK}_{d,B,q}-\bcK_d
    =
    \mathcal B_q(\widehat{\bcK}_d-\bcK_d).
\]
By the same local-support argument used in the proof of
Theorem~\ref{thm:tapered-cumulant-upper-bound}, the stochastic error is
governed by the effective block length $q+1$. Thus, under the same assumptions
as in Theorem~\ref{thm:tapered-cumulant-upper-bound}, with probability at least
$1-Ce^{-x}$,
\[
    \|\widehat{\bcK}_{d,B,q}-\bcK_d\|
    \lesssim
    \Delta_{q+1,x}(1+\Delta_{q+1,x})^{d-1}.
\]
For $q\ge1$, this is of the same order as
$\Delta_{q,x}(1+\Delta_{q,x})^{d-1}$ up to constants.

The tapered estimator has the same zero-bias property whenever the taper is
flat on the support of $\bcK_d$. For example, if $k=2q$, then
$w_{2q}(m)=1$ for all $m\le q$, so
\[
    \mathcal W_{2q}*\bcK_d=\bcK_d.
\]
Consequently, the tapered estimator also has no approximation bias and its
stochastic error is of order
\[
    \Delta_{2q,x}(1+\Delta_{2q,x})^{d-1},
\]
matching the hard-banded estimator up to constants.
\end{remark}

\begin{remark}[Global minimum-distance estimators]
\label{rem:ma-local-global}
Theorem~\ref{thm:cumulant-ma} is stated for a local minimum-distance estimator
over $\mathcal N_{\rho_0}$. This restriction is primarily a theoretical device,
since $\mathcal N_{\rho_0}$ depends on the unknown parameter $\eta_\star$.
In practice, one may instead compute a global estimator
\[
    \widetilde\eta
    \in
    \argmin_{\eta\in\Theta}
    \|\widehat b-M(\eta)\|_2^2.
\]
Assume that the conditions of Theorem~\ref{thm:cumulant-ma} hold, that
$\eta_\star\in\Theta$, and that a global minimizer exists. Suppose further
that the model is globally separated outside the local neighborhood: for some
$c_0>0$,
\[
    \inf_{\eta\in\Theta:\|\eta-\eta_\star\|_2>\rho_0}
    \|M(\eta)-M(\eta_\star)\|_2
    \ge
    c_0.
\]
Whenever
\[
    2\sqrt m\,
    \|\widehat{\bcK}_{d,T,k}-\bcK_d\|
    <
    c_0,
\]
every global minimizer $\widetilde\eta$ belongs to
$\mathcal N_{\rho_0}$ and satisfies
\[
    \|\widetilde\eta-\eta_\star\|_2
    \le
    \frac{4\sqrt m}{s}
    \|\widehat{\bcK}_{d,T,k}-\bcK_d\|.
\]

In the independent-replicate setting of
Theorem~\ref{thm:cumulant-ma}, the same conclusion holds with probability at
least $1-Ce^{-x}$ provided
\[
    2\sqrt m
    \left[
        C_{\alpha,d}\beta_{\mathrm{MA}}
        k^{-\alpha-d/2+1}
        +
        C_{\gamma,d,K}
        \Delta_{k,x}(1+\Delta_{k,x})^{d-1}
    \right]
    <
    c_0.
\]
\end{remark}

\subsection{Non-Gaussian Source Localization in Spatial Sensor Arrays}
\label{subsec:spatial-source-localization}

We next give an example where the coordinate ordering is spatial rather than
temporal. Source localization is a classical problem in array signal
processing, where measurements are collected from sensors with known spatial
locations and the goal is to infer the location, direction, or spatial footprint
of latent sources. Classical approaches include covariance-based subspace
methods such as MUSIC and ESPRIT \citep{schmidt1986multiple,roy1989esprit,krim1996sensor}.
These second-order methods can be sensitive to Gaussian background noise,
whereas cumulants of order $d\ge3$ vanish for Gaussian variables. Thus
higher-order cumulants can isolate non-Gaussian source structure that is
invisible or less stable in covariance-based analysis. Higher-order statistics
have therefore been used in array processing and blind source separation,
especially for non-Gaussian sources and in the presence of additive Gaussian
noise \citep{cardoso1993blind,cardoso1999high,hyvarinen2001independent}. The
procedure described below can be viewed as a tapered-cumulant plug-in version
of spatial matched filtering.

Let $X\in\R^p$ denote one centered snapshot from an ordered sensor array and
suppose that
\[
    X
    =
    \sum_{\ell=1}^L a_\ell S_\ell
    +
    \xi,
\]
where $S_1,\ldots,S_L$ are independent centered non-Gaussian source variables,
$a_\ell=(a_{\ell 1},\ldots,a_{\ell p})^\top\in\R^p$ is the spatial loading
vector of source $\ell$, and $\xi$ is centered Gaussian noise independent of the
sources. Let $\tau_{d,\ell}$ be the scalar order-$d$ cumulant of $S_\ell$.
By multilinearity of cumulants, independence of the sources, and the fact that
Gaussian noise has zero cumulants of order $d\ge3$, the order-$d$ cumulant
tensor of $X$ satisfies
\begin{equation}
\label{eq:spatial-cumulant-representation}
    \bcK_d
    =
    \sum_{\ell=1}^L
    \tau_{d,\ell}
    a_\ell^{\otimes d}.
\end{equation}
Thus the higher-order cumulant tensor removes the additive Gaussian measurement
noise and retains the non-Gaussian spatial source structure.

We assume that each source has a local spatial footprint. Specifically, suppose
that source $\ell$ has center $c_\ell\in[p]$ and that there exist constants
$A<\infty$ and $\rho\in(0,1)$ such that
\begin{equation}
\label{eq:spatial-loading-decay}
    |a_{\ell i}|
    \le
    A\rho^{|i-c_\ell|},
    \qquad
    i\in[p],
    \quad
    \ell=1,\ldots,L.
\end{equation}
This assumption is natural when a physical source affects nearby sensors more
strongly than distant sensors. It also implies that the corresponding cumulant
tensor is bandable in the spatial ordering.

\begin{lemma}[Bandability of local source cumulants]
\label{lem:source-cumulant-bandability}
Under the local source model above and the decay condition
\eqref{eq:spatial-loading-decay}, for every $\alpha>0$, the cumulant tensor
$\bcK_d$ belongs to $\mathcal K_{\alpha,\beta_{\mathrm{src}}}^d$, where
\[
    \beta_{\mathrm{src}}
    =
    L
    \max_{1\le \ell\le L}|\tau_{d,\ell}|
    A^d
    \sup_{D\ge0}
    (1+D)^{\alpha+d-1}\rho^D
    <
    \infty.
\]
\end{lemma}

We now describe a cumulant matched-filter localization procedure. For each
candidate location $c\in[p]$, let $g_c\in\R^p$ be a known local template
centered at $c$, normalized so that $\|g_c\|_2=1$. For example, one may take
$(g_c)_i$ proportional to $\rho_0^{|i-c|}$ for a fixed decay parameter
$\rho_0\in(0,1)$. Define the population cumulant score
\[
    Q(c)
    =
    \left\langle
        \bcK_d,
        g_c^{\otimes d}
    \right\rangle,
    \qquad
    c\in[p].
\]
Under the source model \eqref{eq:spatial-cumulant-representation},
\[
    Q(c)
    =
    \sum_{\ell=1}^L
    \tau_{d,\ell}
    \langle a_\ell,g_c\rangle^d.
\]
Thus $Q(c)$ is large when the template $g_c$ aligns with the spatial footprint
of a non-Gaussian source whose order-$d$ cumulant has the corresponding sign.
If the sign of the source cumulant is unknown, the same construction can be
applied to the absolute score $|Q(c)|$ or to separate positive and negative
scores.

Given the tapered cumulant estimator $\widehat{\bcK}_{d,T,k}$, define the
plug-in score
\[
    \widehat Q_k(c)
    =
    \left\langle
        \widehat{\bcK}_{d,T,k},
        g_c^{\otimes d}
    \right\rangle,
    \qquad
    c\in[p],
\]
and estimate the source location by
\begin{equation}
\label{eq:spatial-location-estimator}
    \widehat c
    \in
    \operatorname{argmax}_{1\le c\le p}
    \widehat Q_k(c).
\end{equation}
For multiple sources, one may instead take separated local maximizers of
$\widehat Q_k(c)$.

The following theorem separates a deterministic argmax perturbation result
from its probabilistic consequence. The deterministic result shows that every
maximizer of the plug-in score lies near the population source location
whenever the realized tensor estimation error is less than half of the
population score gap. The second part inserts the high-probability tapered
cumulant bound.

\begin{theorem}[Tapered cumulant matched-filter source localization]
\label{thm:source-localization}
Let $\{g_c:c\in[p]\}$ be a collection of templates satisfying
$\|g_c\|_2=1$. Suppose that $c_0$ is a population source location satisfying,
for some radius $r\ge0$,
\[
    \operatorname{Gap}(r) :=
    Q(c_0)
    -
    \max_{|c-c_0|>r} Q(c)
    >
    0.
\]
Whenever
\[
    \|\widehat{\bcK}_{d,T,k}-\bcK_d\|
    <
    \frac{\operatorname{Gap}(r)}{2},
\]
every maximizer $\widehat c$ in
\eqref{eq:spatial-location-estimator} satisfies
\[
    |\widehat c-c_0|
    \le
    r.
\]

Moreover, fix any $\alpha>0$, and let
$\beta_{\mathrm{src}}=\beta_{\mathrm{src}}(\alpha)$ be the constant in
Lemma~\ref{lem:source-cumulant-bandability}. In the setting where the
observations are i.i.d. copies of the sensor-array snapshot $X$ and $\sup_{u\in\mathbb S^{p-1}} \|u^\top X\|_{\Phi_\gamma}  \leq  K$ for some $0<\gamma\leq d$, for every
integer $2\le k\le p$ and every $x\ge0$, with probability at least
$1-Ce^{-x}$, every maximizer $\widehat c$ satisfies
\[
    |\widehat c-c_0|
    \le
    r,
\]
provided
\[
    C_{\alpha,d}
    \beta_{\mathrm{src}}
    k^{-\alpha-d/2+1}
    +
    C_{\gamma,d,K}
    \Delta_{k,x}
    (1+\Delta_{k,x})^{d-1}
    <
    \frac{\operatorname{Gap}(r)}{2}.
\]
\end{theorem}

Theorem~\ref{thm:source-localization} translates the spectral-norm accuracy of
the tapered cumulant estimator into a spatial localization guarantee. For
every $c\in[p]$, the normalization $\|g_c\|_2=1$ and the definition of the
tensor spectral norm imply
\[
    |\widehat Q_k(c)-Q(c)|
    \le
    \|\widehat{\bcK}_{d,T,k}-\bcK_d\|.
\]
Thus localization succeeds whenever the tensor estimation error is less than
half of the population separation gap. In this sense, tapering regularizes the
cumulant matched-filter score while preserving the Gaussian-noise cancellation
property of higher-order cumulants. By the reverse triangle inequality, the same bound yields
$\bigl||\widehat Q_k(c)|-|Q(c)|\bigr|
\le\|\widehat{\bcK}_{d,T,k}-\bcK_d\|$.
Consequently, Theorem~\ref{thm:source-localization} holds verbatim for the
absolute-score estimator
$\widehat c\in\operatorname{argmax}_{1\le c\le p}|\widehat Q_k(c)|$,
with $Q$ replaced by $|Q|$ in the gap condition; this is the variant used in
Section~\ref{subsec:spe1-spatial}.

\subsection{Broader Plug-in Uses of Tapered Cumulants}
\label{subsec:regularized-plugin}

The preceding examples follow a common regularized plug-in principle. Suppose
a downstream procedure depends on the population order-$d$ cumulant tensor
through a functional or estimating map $\mathcal T(\bcK_d)$. Instead of using
the raw empirical cumulant tensor, we first form the tapered estimator
$\widehat{\bcK}_{d,T,k}$ and then use the regularized plug-in quantity $\mathcal T(\widehat{\bcK}_{d,T,k})$. In the examples above, $\mathcal T$ corresponds respectively to a cumulant
Yule--Walker linear system, a minimum-distance cumulant criterion, and a
spatial matched-filter score.

This principle is not limited to the particular time-series and spatial models
considered above. It applies whenever the variables have a meaningful ordering
and the relevant higher-order dependence is approximately local. Examples
include temporal lags, spatial sensor arrays, frequency-indexed measurements,
genomic positions, and other sequence-ordered features. Existing procedures
based on empirical third-order or higher-order cumulants can therefore be
regularized by inserting the tapered cumulant estimator before the original
downstream step. Such procedures include higher-order spectral and bispectral
analysis \citep{nikias1987bispectrum,mendel1991tutorial,nikias1993signal,nikias1993higher,zhang2018asymptotic,sifft2025signalsnap},
detection and signal reconstruction under additive Gaussian noise
\citep{shamsunder1994detection,tugnait1989time,petropulu1990complex,hinich1992time,wang2021third},
phase reconstruction and nonminimum-phase system identification
\citep{lii1982deconvolution,matsuoka1984phase,pan1988complex,petropulu1995noncausal},
nonlinear system identification
\citep{brillinger1977identification,schetzen1989volterra,hickey2009higher,parker2020nonlinear,basti2024bicoherence},
and non-Gaussian ARMA estimation based on higher-order moment or cumulant
equations
\citep{huzii1981estimation,giannakis1989cumulant,friedlander1990asymptotically,taniguchi1991higher}.

The usefulness of this plug-in step comes from two complementary features of
higher-order cumulants. First, cumulants provide identifying information beyond
covariance. Second-order methods
\citep[e.g.][]{brockwell1991time,hamilton1994time,schmidt1986multiple,roy1989esprit,krim1996sensor}
may miss asymmetric, nonlinear, or otherwise non-Gaussian dependence. By
contrast, cumulants of order $d\ge3$ can capture such structure and vanish for
Gaussian variables, making them insensitive to additive Gaussian noise at those
orders.

Second, tapering regularizes cumulant-based procedures in ordered
high-dimensional problems. As discussed in Section~\ref{subsec:rate-comparison},
raw empirical cumulant tensors can be unstable in spectral norm. When
higher-order dependence is local in the coordinate ordering, distant cumulant
entries are often weaker and less informative. Tapering preserves the dominant
local entries while downweighting or removing long-range entries. Consequently,
any downstream procedure that is stable with respect to perturbations of
$\bcK_d$ inherits the bias--variance tradeoff of the tapered cumulant
estimator. This is analogous in spirit to banded covariance regularization for
ordered data \citep{wu2009banding}, but it targets higher-order non-Gaussian
dependence rather than second-order covariance structure.

\section{Computation and Tuning}
\label{sec:computation-tuning}

This section discusses computation and practical tuning. The bandwidth $k$ is
the main tuning parameter in the tapered cumulant estimator. In downstream
applications, additional model-specific parameters may also be selected, such
as the autoregressive order $r$ or the moving-average order $q$. The bandwidth
choice is guided by the bias--variance theory above, whereas the model-order
criteria below are intended as practical diagnostics for numerical work.

\subsection{Computing the Tapered Cumulant Estimator}

In implementation, one may first subtract the sample mean and then compute the
empirical cumulant from the centered empirical moments using the
moment--cumulant formula. 
After mean subtraction the third sample cumulant coincides with the third centered empirical moment, while for $d\ge 4$ the formula retains nontrivial corrections involving centered moments of orders 2 through $d-2$. The full order-$d$ sample cumulant tensor
has $p^d$ entries and is generally too large to form explicitly when $p$ is
large. The tapered estimator avoids this full storage cost because
$w_k(\operatorname{diam}(\bi))=0$ whenever
$\operatorname{diam}(\bi)\ge k$. Therefore, only entries satisfying
$\operatorname{diam}(\bi)<k$ need to be computed. The number of such entries
is at most $C_dpk^{d-1}$ for fixed $d$, so the tapered estimator can be stored
using $O(pk^{d-1})$ memory rather than $O(p^d)$ memory.

For each retained multi-index $\bi=(i_1,\ldots,i_d)$, the sample cumulant is
computed through the plug-in moment--cumulant formula
\eqref{eq_sample_cumulant_tensor}. Since $d$ is fixed, the number of
partitions of $[d]$ is a constant depending only on $d$. Thus computing all
entries in the tapered band has cost $O(npk^{d-1})$ up to constants depending
on $d$. The tapered estimator is then formed entrywise as
\[
    (\widehat{\bcK}_{d,T,k})_{\bi}
    =
    w_k(\operatorname{diam}(\bi))(\widehat{\bcK}_d)_{\bi}.
\]

\subsection{Computing the Downstream Estimators}
\label{subsec:computing-downstream}

The full tapered tensor need not be materialized if the downstream procedure
uses only selected entries. Once these selected cumulant entries have been
estimated, the downstream model-specific estimators are computed from
lower-dimensional systems. In the autoregressive cumulant Yule--Walker
estimator, the fitted cumulant equations have the form
$\widehat b_k\approx \widehat A_k\phi$, and the estimator is
\[
    \widehat\phi_{d,k}
    =
    (\widehat A_k^\top\widehat A_k)^{-1}\widehat A_k^\top\widehat b_k,
\]
provided $\widehat A_k$ has full column rank.

For the moving-average minimum-distance estimator, the criterion is nonlinear
in the moving-average coefficients but linear in the innovation cumulant
$\tau_d$. Hence $\tau_d$ can be profiled out. Let
$\Psi(\theta)=(\Psi_1(\theta),\ldots,\Psi_m(\theta))^\top$, where
$\Psi_a(\theta)$ is the model-implied cumulant equation defined in
Section~\ref{subsec:cumulant-ma}. For any fixed $\theta$ with
$\Psi(\theta)\ne0$, the least-squares minimizer over $\tau_d$ is
\[
    \widehat\tau_d(\theta)
    =
    \frac{\langle \widehat b,\Psi(\theta)\rangle}{\|\Psi(\theta)\|_2^2}.
\]
Therefore, the moving-average estimator can be computed by solving the
$q$-dimensional profiled problem
\begin{equation}
\label{eq:ma-op-theta}
    \widehat\theta
    \in
    \argmin_{\theta\in\Theta_\theta}
    \left\|
    \widehat b-\widehat\tau_d(\theta)\Psi(\theta)
    \right\|_2^2.
\end{equation}
For fixed small $q$, this profiled problem can be solved by multistart
nonlinear optimization, using quasi-Newton methods such as BFGS or its
bound-constrained variant L-BFGS-B
\citep{nocedal2006numerical,byrd1995limited}.

\subsection{Bandwidth Selection}
\label{subsec:select-bandwidth}

The bandwidth $k$ controls the bias--variance tradeoff. A smaller $k$ reduces
variance by keeping fewer empirical cumulant entries, but increases
approximation bias. The discussion in Section~\ref{subsec:rate-comparison}
suggests balancing the leading bias term
$\beta k^{-\alpha-d/2+1}$ and the leading stochastic term
$\sqrt{(k+\log p)/n}$. When $\log p$ is no larger than the resulting oracle
bandwidth, this reduces to balancing $\beta k^{-\alpha-d/2+1}$ with
$\sqrt{k/n}$, giving
\[
    k_\star
    \asymp
    (\beta^2 n)^{1/(2\alpha+d-1)}.
\]
Thus, if $\alpha$ and $\beta$ are known or can be specified and
$\log p\lesssim k_\star\lesssim p$, one may choose $k$ as the nearest
admissible integer to $k_\star$. Outside this regime, one may instead minimize
the full leading-order criterion
\[
    \beta k^{-\alpha-d/2+1}
    +
    \sqrt{\frac{k+\log p}{n}}
\]
over the admissible bandwidth grid.

In practice, $\alpha$ and $\beta$ are usually unknown. A data-driven approach
is to choose $k$ from a finite grid
$\mathcal K=\{k_1<k_2<\cdots<k_L\}$ using a Lepski-type stability criterion.
This criterion is motivated by adaptive bandwidth-selection methods that
compare estimators across bandwidths
\citep{lepskii1991,goldenshluger2011bandwidth}. For $k\le k'$, let
$R(k,k')$ be a nonnegative loss measuring how much the estimated object changes
when the bandwidth is increased from $k$ to $k'$. The loss should be chosen
according to the object used in the subsequent analysis. A generic stability
rule is
\begin{equation}
\label{eq:k-selection}
    \widehat k
    =
    \min
    \left\{
    k\in\mathcal K:
    R(k,k')
    \le
    A\,s(k')
    \text{ for all } k'\in\mathcal K,\ k'\ge k
    \right\},
\end{equation}
where $A>0$ is a tuning constant and $s(k')$ is a scale or threshold function.
The intuition is that, once $k$ is large enough for the bias to be below the
stochastic error, increasing the bandwidth should no longer change the relevant
estimated object substantially.

A theory-aligned choice of loss is the tensor-level stability loss
\[
    R_{\mathrm{tensor}}(k,k')
    =
    \|\widehat{\bcK}_{d,T,k}-\widehat{\bcK}_{d,T,k'}\|.
\]
Motivated by Theorem~\ref{thm:tapered-cumulant-upper-bound}, one may take
\[
    s(k)
    =
    \sqrt{\frac{k+\log p}{n}}
    +
    \frac{(k+\log p)^{d/\gamma}}{n},
\]
where we take $\gamma=2$ in all numerical work. The unknown constants in the high-probability bound are absorbed into the
multiplier $A$. This version is closest to the theoretical error bound because
it compares the full tapered cumulant tensors.

However, exact computation of the tensor spectral norm is generally difficult
for $d\ge3$ \citep{hillar2013most}. Therefore, in numerical work, the
tensor-level loss can be replaced by a computationally tractable proxy, such as
the maximum matricized operator norm over unfoldings, an approximate tensor
spectral norm computed by alternating maximization, or a loss defined directly
on the lower-dimensional object used in the final estimator.

For example, in the cumulant Yule--Walker application, the final estimator
depends on the selected cumulant equations through $\widehat A_k$ and
$\widehat b_k$. It is therefore natural to tune $k$ by the equation-level
stability loss
\[
    R_{\mathrm{YW}}(k,k')
    =
    \left(
    \|\widehat A_k-\widehat A_{k'}\|_F^2
    +
    \|\widehat b_k-\widehat b_{k'}\|_2^2
    \right)^{1/2}.
\]
The bandwidth is then selected as the smallest $k$ satisfying
\eqref{eq:k-selection}.

\subsection{Model-Order Selection}

In addition to the bandwidth $k$, some applications require choosing a
model-specific order parameter. For example, the model order may be the
autoregressive order in the cumulant Yule--Walker estimator or the
moving-average order in the cumulant minimum-distance estimator.

A general approach is to fit the model over a finite grid of candidate orders
$\mathcal M=\{m_1,\ldots,m_J\}$. For each $m\in\mathcal M$, we construct the
corresponding cumulant equations, estimate the model parameters, and evaluate
how well the fitted model reproduces the selected empirical cumulants. This
gives a residual diagnostic of the form
\[
    L(m)
    =
    \frac{1}{N_m}
    \|r_m(\widehat\eta_m)\|_2^2,
\]
where $r_m(\widehat\eta_m)$ is the residual vector from the selected cumulant
equations, and $N_m$ is the number of equations used for order $m$.

Since $L(m)$ often decreases as the model order increases, one may use an
information criterion of the form
\[
    \operatorname{IC}(m)
    =
    \log\{L(m)+\epsilon\}
    +
    \lambda
    \frac{\nu_m\log n_{\rm eff}}{n_{\rm eff}},
\]
where $\nu_m$ is the number of fitted parameters, $n_{\rm eff}$ is the
effective sample size, $\lambda>0$ controls the penalty strength, and
$\epsilon>0$ is a small numerical constant. For independent replicated
observations, $n_{\rm eff}=n$; for a single dependent trajectory, $n_{\rm eff}$
may be taken as the number of usable approximately independent blocks or
another effective sample-size estimate. The selected order is then chosen by
comparing the residual diagnostic and the information criterion over the candidate grid.

In Section~\ref{sec:real-data}, we illustrate this procedure for the
autoregressive and moving-average examples.

\section{Simulation Studies}
\label{sec:simulation}

We conduct three simulation studies. The first evaluates spectral-norm
estimation of bandable cumulant tensors and compares the raw sample cumulant,
hard banding, and tapering. The second studies the data-driven bandwidth
selector used in implementation. The third examines how tapered third-order
cumulants perform as inputs to the autoregressive and moving-average estimators
introduced in Section~\ref{sec:application}.

\subsection{Estimation of Bandable Cumulant Tensors}
\label{subsec:simulation-tensor-estimation}

We first examine the finite-sample performance of the tapered cumulant
estimator. The purpose of this experiment is to compare the raw sample
cumulant, hard banding, and tapering. We generate observations from a
non-Gaussian linear process. For $j=1,\ldots,p$, let
\[
    X_j
    =
    \sum_{\ell=0}^{L}
    \psi_\ell \varepsilon_{j+L-\ell}
    +
    \sigma Z_j,
\]
where $\{\varepsilon_s\}$ are independent standardized non-Gaussian
innovations, $Z=(Z_1,\ldots,Z_p)^\top$ is standard Gaussian noise independent
of the innovations, and the coefficients $\psi_\ell$ decay polynomially. Since
the Gaussian component has zero cumulants of order $d\ge3$, the higher-order
cumulant tensor is determined by the non-Gaussian linear component. The true
order-$d$ cumulant tensor is computed exactly from the loading representation
and is used as the target in the simulation.

We consider cumulant orders $d=3$ and $d=4$, and bandability parameters
$\alpha\in\{0.1,1,10\}$. For each configuration, we compare three estimators:
the raw sample cumulant tensor $\widehat{\bcK}_d$, the hard-banded estimator
$\mathcal B_k(\widehat{\bcK}_d)$ with oracle bandwidth, and the tapered
estimator $\widehat{\bcK}_{d,T,k}$ with oracle bandwidth. The hard-banded
estimator assesses the value of localization alone, whereas the tapered
estimator assesses the additional benefit of a smooth transition at the
boundary of the diagonal tube. The oracle bandwidth is selected from a finite
candidate grid by minimizing the empirical estimation error against the known
population cumulant tensor. It is used only as a benchmark, allowing us to
isolate the effect of localization from the additional difficulty of bandwidth
selection. The estimation error is measured by an approximate tensor spectral
norm using \texttt{CP-ALS} implemented by R package \texttt{rTensor} and averaged over 50 Monte Carlo replications. The CP-ALS value is a lower bound on the tensor spectral norm and is used as a common proxy for all three estimators.

\begin{figure}[htbp]
    \centering
    \includegraphics[width=0.7\textwidth]{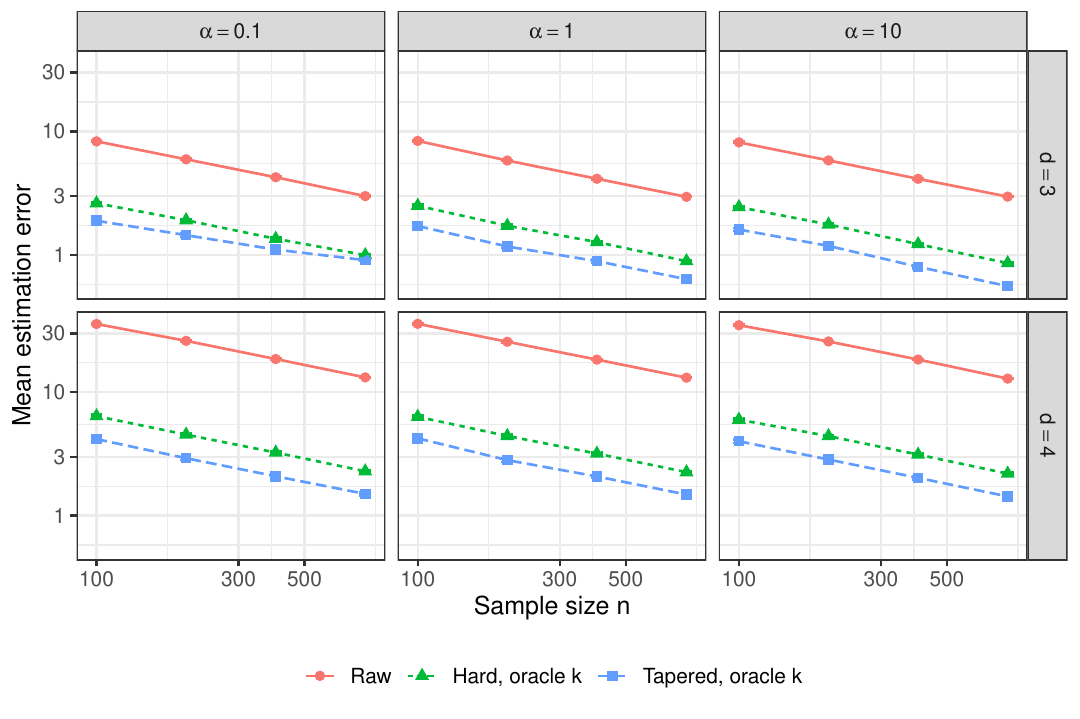}
    \caption{Mean tensor estimation error as a function of the sample size $n$.}
    \label{fig:sim_tensor_error_n}
\end{figure}

Figure~\ref{fig:sim_tensor_error_n} shows that the estimation error decreases
as the sample size increases. The raw sample cumulant has the largest error,
while both localized estimators substantially improve the estimation accuracy.
The tapered estimator gives the smallest error across the displayed settings.

\begin{figure}[htbp]
    \centering
    \includegraphics[width=0.7\textwidth]{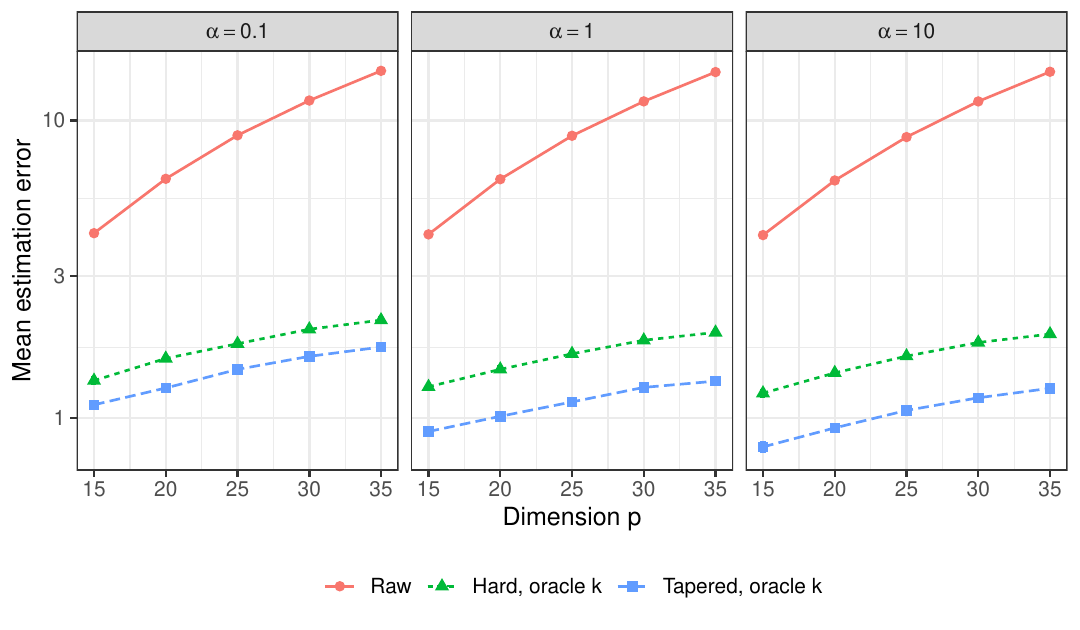}
    \caption{Mean tensor estimation error as a function of the dimension $p$.}
    \label{fig:sim_tensor_error_p}
\end{figure}

Figure~\ref{fig:sim_tensor_error_p} shows that the error increases with the
dimension $p$ for $d=3$, but the increase is much slower for the hard-banded
and tapered estimators than for the raw sample cumulant. The results also show
that larger values of $\alpha$ lead to slightly smaller estimation errors,
consistent with the fact that larger $\alpha$ corresponds to faster
off-diagonal decay and hence a more strongly bandable cumulant tensor.

Overall, these simulations corroborate the main qualitative implications of
Theorem~\ref{thm:tapered-cumulant-upper-bound}: exploiting bandability improves
cumulant tensor estimation, and tapering provides a stable localized estimator
in finite samples.

\subsection{Data-Driven Bandwidth Selection}
\label{subsec:sim-bandwidth-selection}

We next examine the data-driven bandwidth selector used in the tapered cumulant
estimator. For each simulated data set, we compute the tapered cumulant
estimator over a grid of candidate bandwidths and select $\widehat k$ using the
rule described in Section~\ref{subsec:select-bandwidth}. We vary the threshold
multiplier $A$ in \eqref{eq:k-selection} and compare the risk of the selected
estimator with the oracle risk over the same bandwidth grid.

The simulation is based on the third-order cumulants. We generate non-Gaussian
observations with polynomially decaying dependence so that the cumulant tensor
satisfies the bandable structure. The dimension is fixed at $p=20$, and we
consider sample sizes $n\in\{200,400,800\}$ and smoothness levels
$\alpha\in\{0.3,1,2\}$. For a candidate bandwidth $k$, the estimation error is
measured by the approximate tensor spectral norm
$\|\widehat{\bcK}_{3,T,k}-\bcK_3\|$. The oracle bandwidth is defined as the
minimizer of this error over the candidate grid. For each value of $A$, we
report the mean ratio between the error of the data-driven estimator and the
oracle error.

Figure~\ref{fig:bandwidth-risk-ratio} shows the resulting risk ratios between
the error of the Lepski-selected tapered estimator and the oracle error over
the candidate bandwidth grid. The risk ratio is close to one for a moderate
range of threshold multipliers, indicating that the Lepski-type rule can
recover a near-oracle bandwidth in this setting. The best performance is
obtained around $A=1$. When $A$ is substantially larger, the selected bandwidth
tends to be too small, especially for $\alpha=0.3$ and $\alpha=1$, leading to
an increased bias and a larger risk ratio. For $\alpha=2$, the method is less
sensitive once $A$ is moderately large, but very small values of $A$ still lead
to worse performance. Based on these simulations, we use $A=1$ as the default
threshold multiplier in the subsequent numerical analyses.

\begin{figure}[t]
    \centering
    \includegraphics[width=0.85\textwidth]{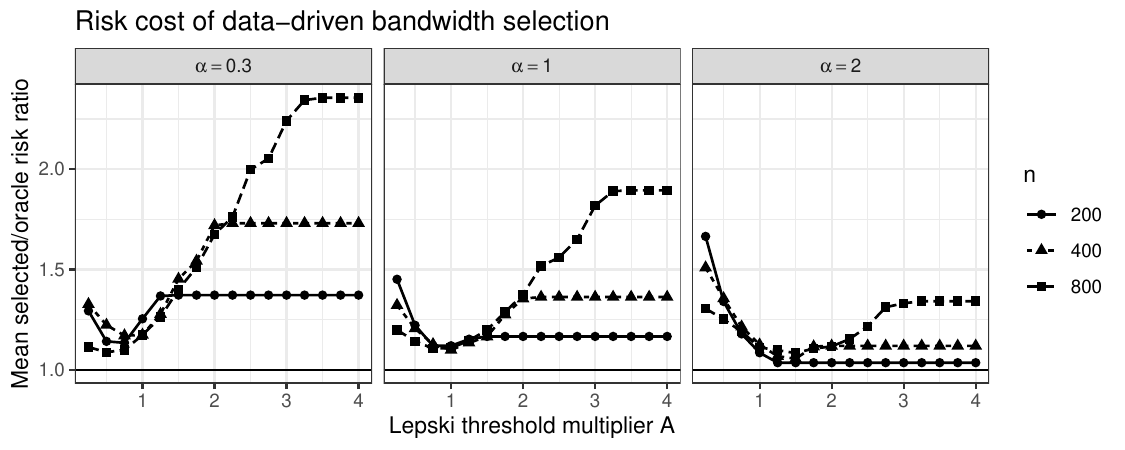}
    \caption{
    Risk ratio for data-driven bandwidth selection.
    }
    \label{fig:bandwidth-risk-ratio}
\end{figure}

\subsection{Time Series Estimation}
\label{subsec:simulation-time-series}

We next evaluate the proposed cumulant-based estimators in simple
autoregressive and moving-average models to illustrate that the cumulant
equations used in Section~\ref{subsec:cumulant-yule-walker} and
Section~\ref{subsec:cumulant-ma} lead to stable finite-sample coefficient
estimation. In both experiments, the empirical cumulant equations are formed
using tapered third-order cumulants.

We first consider autoregressive coefficient estimation. Data are generated
from a stable autoregressive model with centered non-Gaussian innovations. To
examine robustness to second-order contamination, the observed series is formed
by adding independent Gaussian measurement noise with standard deviation $0$,
$0.75$, or $1.25$. For each sample size, we estimate the autoregressive
coefficients using the proposed third-order cumulant Yule--Walker estimator and
the ordinary covariance Yule--Walker estimator. The performance metric is the
mean $\ell_2$ coefficient error over 100 Monte Carlo replications.

Figure~\ref{fig:ar_yw_comparison} reports the results. When there is no
measurement noise, both methods improve as the sample size increases, and the
ordinary Yule--Walker estimator performs slightly better, as expected from the
correctly specified second-order autoregressive equations. In the presence of
additive Gaussian measurement noise, however, the ordinary Yule--Walker
estimator exhibits a large persistent error, while the cumulant Yule--Walker
estimator continues to improve with the sample size. This behavior is
consistent with the fact that independent Gaussian noise has zero third-order
cumulant but changes the covariance structure. The experiment therefore
illustrates a setting in which higher-order cumulant equations can provide a
useful alternative to ordinary second-order Yule--Walker estimation.

\begin{figure}[t]
    \centering
    \includegraphics[width=0.8\textwidth]{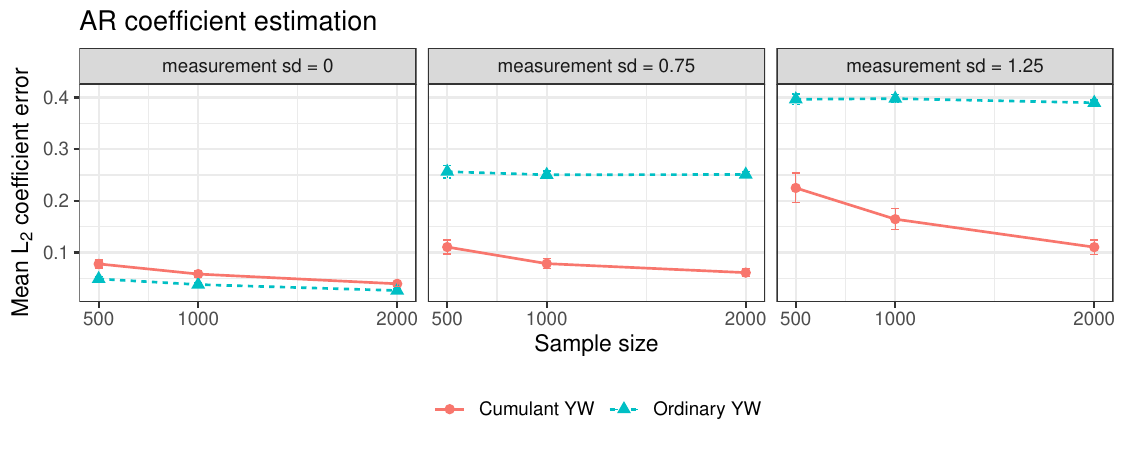}
    \caption{AR coefficient estimation.}
    \label{fig:ar_yw_comparison}
\end{figure}

We also consider moving-average coefficient estimation. Data are generated from
a fixed moving-average model with centered non-Gaussian innovations. The
coefficients are estimated by the third-order cumulant minimum-distance
estimator described in Section~\ref{subsec:cumulant-ma}. After profiling out
the innovation cumulant in closed form, we minimize the resulting nonlinear
least-squares criterion over the moving-average coefficients using a multistart
box-constrained quasi-Newton algorithm. The performance metric is again the
mean $\ell_2$ coefficient error over 100 Monte Carlo replications. The purpose of
this experiment is not to compare against a second-order baseline, but to
verify that the selected third-order cumulant equations lead to stable
finite-sample estimation.

Figure~\ref{fig:ma_md_error} shows that the estimation error decreases steadily
as the sample size increases. This indicates that the selected third-order
cumulant equations contain useful information for identifying the moving-average
coefficients. Together with the autoregressive experiment, this simulation
supports the use of tapered higher-order cumulants for finite-sample estimation
in non-Gaussian time series models.

\begin{figure}[t]
    \centering
    \includegraphics[width=0.5\textwidth]{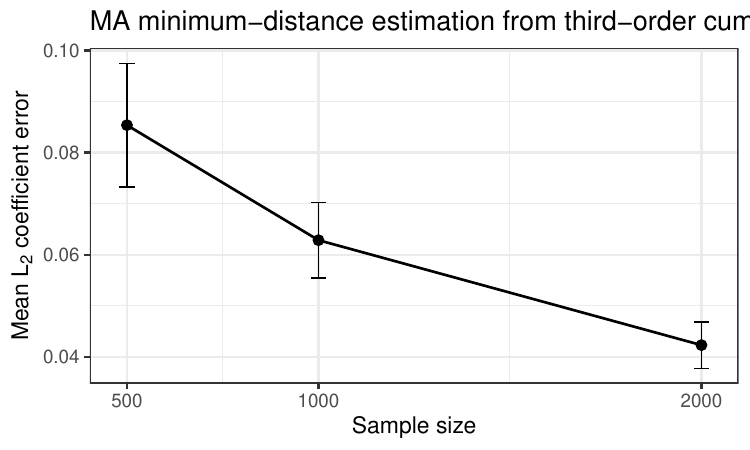}
    \caption{MA coefficient estimation using the third-order cumulant minimum-distance estimator.}
    \label{fig:ma_md_error}
\end{figure}

\section{Real Data Analyses}\label{sec:real-data}

We present three real-data examples illustrating the plug-in role of tapered
higher-order cumulants. The first two examples use temporal data and show how
tapered third-order cumulants can regularize estimating equations for
non-Gaussian time-series models. The third example uses spatially ordered
Neuropixels waveforms and illustrates the use of tapered cumulants for
localizing non-Gaussian source structure in the presence of Gaussian-like
background variation. These examples are intended as empirical demonstrations
of the proposed regularization principle rather than exact generative-model
claims.

\subsection{RR Interval Analysis}
\label{subsec:rr_interval_analysis}

We illustrate the cumulant Yule--Walker estimator in Section~\ref{subsec:cumulant-yule-walker} using the RR interval dataset of healthy subjects available through PhysioNet \citep{irurzun2021rr}.
An RR interval is the elapsed time between two consecutive R-wave peaks in an electrocardiogram, and is commonly used as a beat-to-beat measure of cardiac rhythm.
The dataset contains multiple independent subjects, each with a longitudinal sequence of RR intervals.

For each subject, we remove physiologically implausible RR intervals, discard large local jumps on the logarithmic scale, apply a logarithmic transformation, and then center and standardize the transformed trajectory.
We estimate third-order cumulants within each subject and average the subject-level cumulants, which uses each full trajectory while preserving the subject-level independence used for uncertainty quantification.
The preprocessing, cumulant aggregation, covariance baseline, bandwidth selector, and order diagnostic are given in Supplement Section~\ref{subsec:rr-implementation-supp}.

We fit the third-order cumulant Yule--Walker equations described in Section~\ref{subsec:cumulant-yule-walker}.
For the main fit, we use an AR(3) model and lag pairs in $\{1,\ldots,8\}^2$.
The taper bandwidth is selected from $\{3,5,7,9,11,13,15\}$ using the Lepski-type stability rule in Section~\ref{subsec:select-bandwidth}.
For the RR interval data, this procedure selects $\widehat k=13$.

\begin{figure}[htbp]
    \centering
    \begin{subfigure}[t]{0.53\textwidth}
        \centering
        \includegraphics[width=\linewidth]{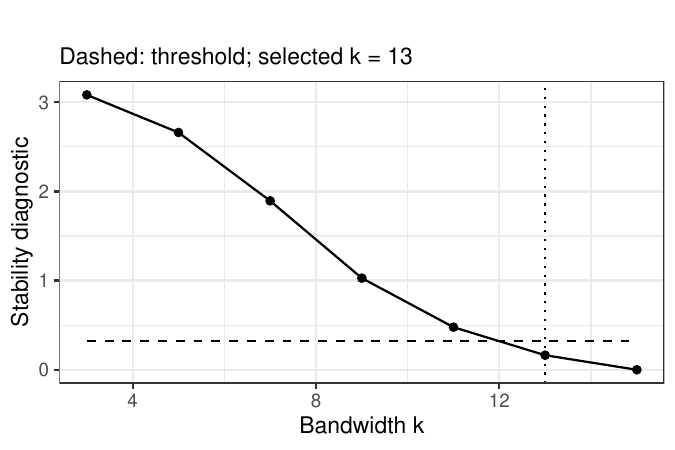}
        \caption{Bandwidth diagnostic.}
        \label{fig:rr_bandwidth}
    \end{subfigure}
    \hfill
    \begin{subfigure}[t]{0.45\textwidth}
        \centering
        \includegraphics[width=\linewidth]{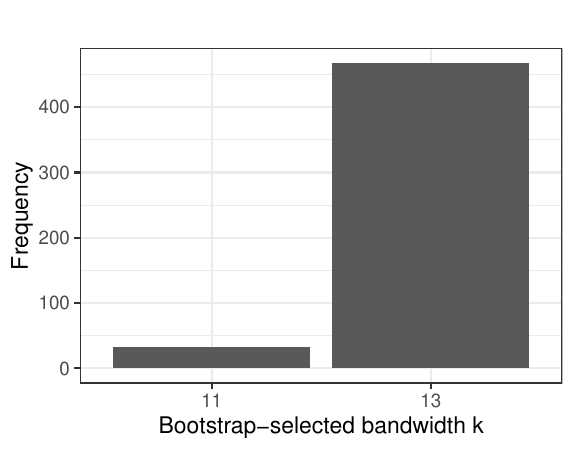}
        \caption{Bootstrap bandwidths.}
        \label{fig:rr_bootstrap_k}
    \end{subfigure}
    \caption{Bandwidth selection.}
    \label{fig:rr_bandwidth_selection}
\end{figure}

Figure~\ref{fig:rr_bandwidth_selection} shows that the stability diagnostic decreases with the bandwidth and first falls below the threshold at $k=13$.
The subject-level bootstrap also selects $k=13$ in most replications, with a smaller number selecting $k=11$.

\begin{figure}[htbp]
    \centering
    \begin{subfigure}[t]{0.44\textwidth}
        \centering
        \includegraphics[width=\linewidth]{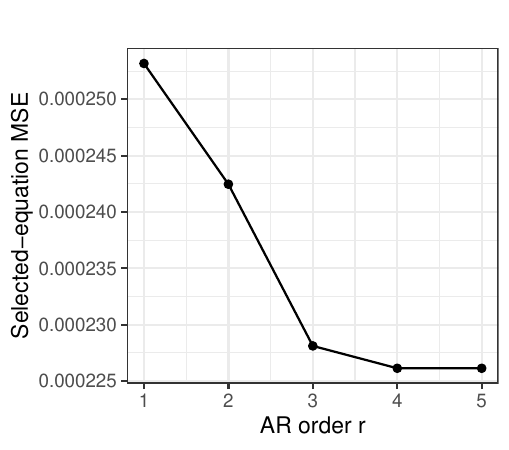}
        \caption{Order selection.}
        \label{fig:rr_order}
    \end{subfigure}
    \hfill
    \begin{subfigure}[t]{0.54\textwidth}
        \centering
        \includegraphics[width=\linewidth]{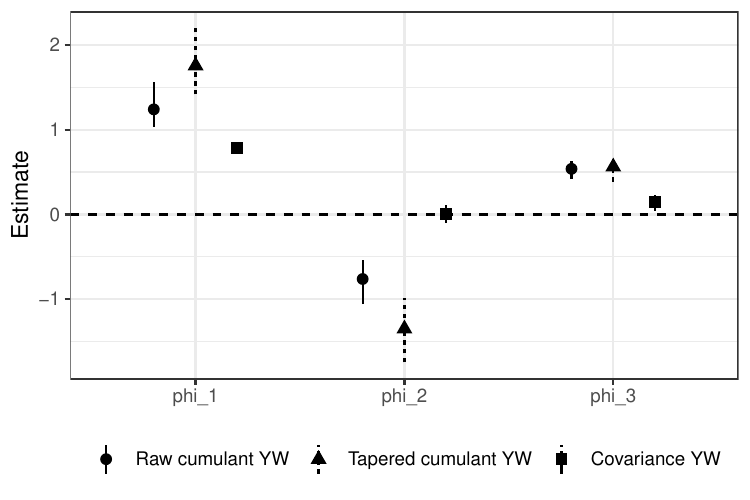}
        \caption{Coefficient estimates.}
        \label{fig:rr_coefficients}
    \end{subfigure}

    \par\medskip

    \begin{subfigure}[t]{0.48\textwidth}
        \centering
        \includegraphics[width=\linewidth]{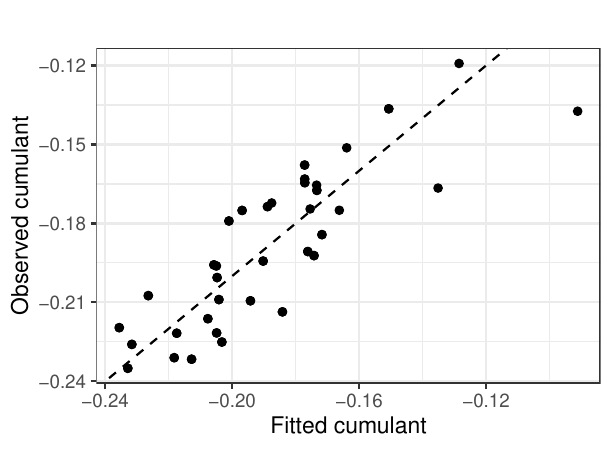}
        \caption{Tapered cumulant fit.}
        \label{fig:rr_fit}
    \end{subfigure}
    \hfill
    \begin{subfigure}[t]{0.48\textwidth}
        \centering
        \includegraphics[width=\linewidth]{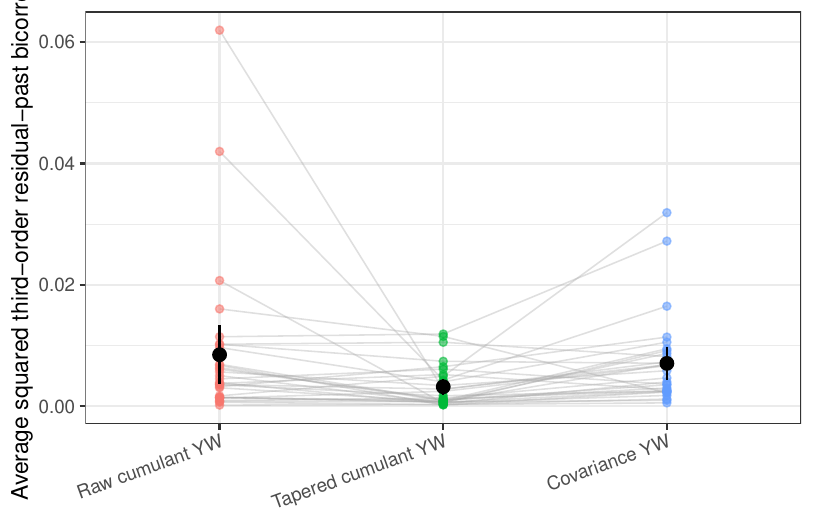}
        \caption{Third-order portmanteau scores.}
        \label{fig:rr_third_order_portmanteau}
    \end{subfigure}
    \caption{RR interval analysis. }
    \label{fig:rr_results}
\end{figure}

Figure~\ref{fig:rr_results} reports the order diagnostic, the three-method coefficient estimates, the fitted cumulant equations, and a held-out third-order goodness-of-fit comparison. The selected-equation MSE decreases through $r=3$ and changes little afterward, supporting an AR(3) model as a parsimonious description of the dominant dependence structure. Figure~\ref{fig:rr_coefficients} compares the tapered cumulant estimator with the raw cumulant Yule--Walker estimator obtained from the same lag pairs without tapering and with the ordinary covariance Yule--Walker estimator. This comparison isolates the effect of tapering while keeping the cumulant order, autoregressive order, and estimating equations fixed for the two cumulant estimators. The covariance and cumulant fits emphasize different features of the data: the former summarizes second-order linear dependence, whereas the latter is driven by asymmetric non-Gaussian dependence among lagged RR intervals. The observed and fitted selected cumulants in Figure~\ref{fig:rr_fit} mostly follow the diagonal trend, indicating that the tapered cumulant Yule--Walker system captures the main selected third-order equations.

To evaluate the feature of the fitted model that is directly targeted by cumulant Yule--Walker estimation, Figure~\ref{fig:rr_third_order_portmanteau} reports a held-out third-order portmanteau score. The construction adapts the bicorrelation diagnostics of \cite{ashley1986diagnostic} and \cite{hinich1996testing}. The first 80\% of each subject's trajectory is used for estimation, and the final 20\% is used only for evaluation. For each method, we standardize the residual--past third-order moment for every unique lag pair $1\le h_1\le h_2\le8$ and average its square. Smaller scores therefore indicate better satisfaction of the held-out third-order moment restrictions. 

The tapered cumulant estimator has the lowest across-subject average score, and most paired subject trajectories decrease from the raw to the tapered estimator. The raw cumulant estimator also produces several substantially larger subject-level scores, whereas the tapered scores are more concentrated near zero. The covariance estimator lies between the two cumulant estimators on average under this third-order criterion. These results support the intended role of tapering as a regularizer of noisy higher-lag cumulant equations.

\subsection{Air Quality Sensor Analysis Using a Moving-Average Cumulant Model}
\label{subsec:ma_air_quality}

We illustrate the cumulant minimum-distance estimator in Section~\ref{subsec:cumulant-ma} using the UCI Air Quality sensor data \citep{air_quality_360}.
The dataset consists of hourly gas-sensor measurements collected from an air-quality monitoring device.
We use the CO-targeted sensor response \texttt{PT08.S1(CO)} as a scalar time series.
Because the raw sensor record contains missing values, low-frequency variation, and calendar effects, we preprocess the series before fitting the moving-average cumulant model.

The preprocessing removes missing sensor values, applies a logarithmic transformation, regresses out hour-of-day effects, day-of-week effects, and a smooth time trend, and then differences the residual series.
The processed sequence is centered and standardized and is treated as one long realization of a stationary non-Gaussian time series.
The exact preprocessing, cumulant equations, profiled minimum-distance objective, invertibility constraint, and order criterion are given in Supplement Section~\ref{subsec:ma-air-implementation-supp}.

We fit the third-order moving-average cumulant equations described in Section~\ref{subsec:cumulant-ma}.
The empirical cumulants are estimated by time averages from the processed trajectory, and the innovation cumulant is profiled out before optimizing over the moving-average coefficients.
We consider candidate orders $q=1,\ldots,10$, and the BIC-type criterion selects $q=4$.

\begin{figure}[htbp]
    \centering

    \begin{subfigure}[t]{0.8\textwidth}
        \centering
        \includegraphics[width=\textwidth]{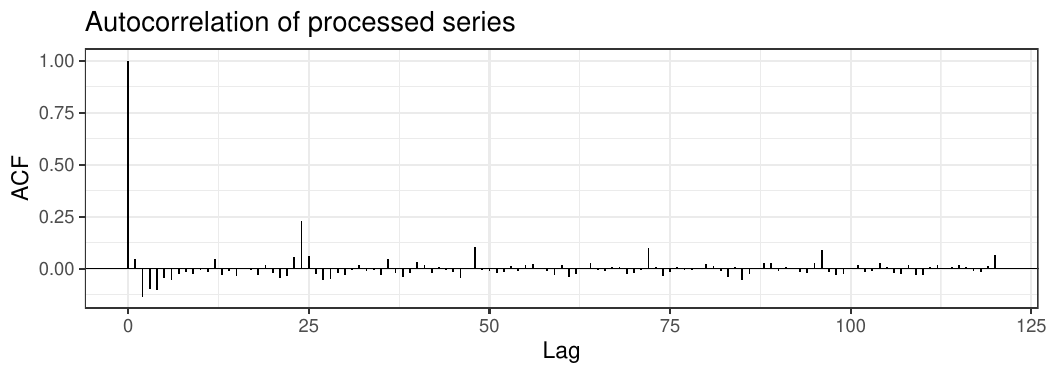}
        \caption{Processed ACF.}
        \label{fig:ma_air_acf}
    \end{subfigure}

    \begin{subfigure}[t]{0.8\textwidth}
        \centering
        \includegraphics[width=\textwidth]{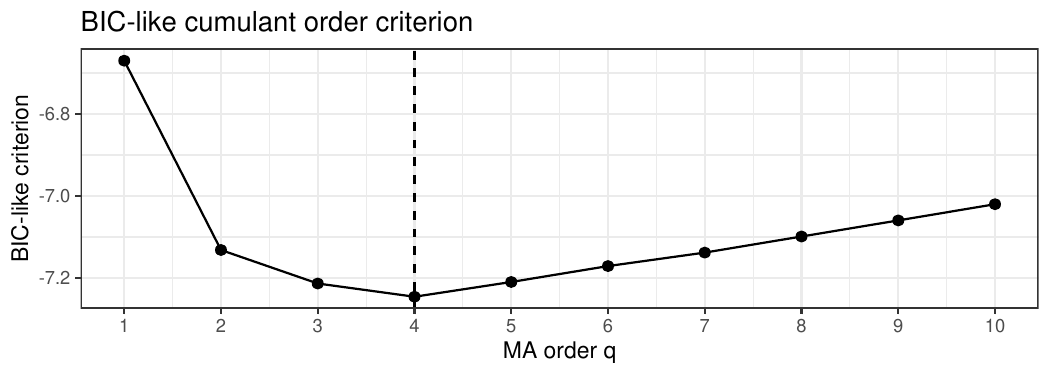}
        \caption{Order criterion.}
        \label{fig:ma_air_order_bic}
    \end{subfigure}

    \caption{Preprocessing and order selection.}
    \label{fig:ma_air_order}
\end{figure}

Figure~\ref{fig:ma_air_order} summarizes the preprocessing and order-selection diagnostics.
The autocorrelation of the processed series decays quickly after differencing, with only mild remaining seasonal spikes.
The BIC-type criterion selects an MA(4) model.

\begin{figure}[htbp]
    \centering

    \begin{subfigure}[t]{0.56\textwidth}
        \centering
        \includegraphics[width=\textwidth]{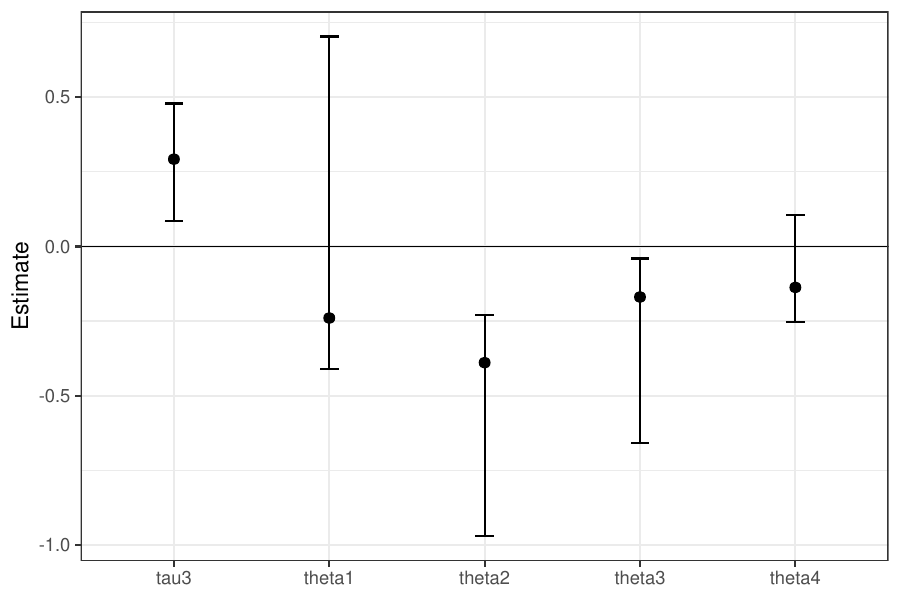}
        \caption{Coefficient estimates.}
        \label{fig:ma_air_coef}
    \end{subfigure}
    \hfill
    \begin{subfigure}[t]{0.38\textwidth}
        \centering
        \includegraphics[width=\textwidth]{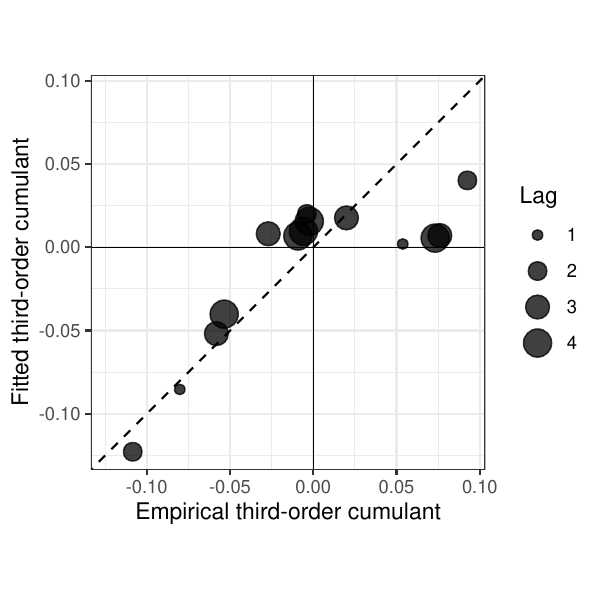}
        \caption{Cumulant fit.}
        \label{fig:ma_air_fit}
    \end{subfigure}

    \caption{Air Quality analysis.}
    \label{fig:ma_air_fit_all}
\end{figure}

Figure~\ref{fig:ma_air_fit_all} reports the fitted MA(4) cumulant model.
The coefficient plot shows the estimated moving-average coefficients with bootstrap confidence intervals.
The fitted-cumulant plot compares the empirical third-order cumulants with the fitted cumulants from the selected MA(4) model.
Most points follow the diagonal trend, although several outliers remain, indicating that the main low-lag third-order dependence is largely captured by a low-order moving-average model.
We interpret the fitted MA(4) model as a parsimonious short-memory approximation rather than an exact data-generating model, since it does not model dependence beyond lag 4.

\subsection{Spatial Localization on Paired Patch-Clamp and Neuropixels Waveforms}
\label{subsec:spe1-spatial}

We evaluate the proposed tapered cumulant localization procedure in Section~\ref{subsec:spatial-source-localization} on the SPE-1 paired patch-clamp and Neuropixels benchmark archive \citep{spe1figshare,zhao2026benchmarking}.
Neuropixels probes provide high-density extracellular recordings along a densely sampled silicon shank \citep{jun2017neuropixels}.
The SPE-1 archive stores extracted extracellular waveforms, probe geometry, template waveforms, sparse waveform-channel maps, and ground-truth unit locations in a SpikeInterface-style intermediate format \citep{buccino2020spikeinterface}.
We use the \texttt{seed\_42} copy of the archive and retain all 11 available cells across the 16 degradation conditions \texttt{D0}--\texttt{D15}, giving 176 cell-condition samples.

The purpose of this experiment is not to benchmark a full spike-localization pipeline, but to test the specific statistical claim that a tapered third-order score can preserve local non-Gaussian spatial information that may be attenuated in a covariance-based score. For each cell-condition, we convert the extracted spike waveforms into an ordered spatial sample on a local Neuropixels depth grid.
The conversion uses the waveform peak and pre-spike baseline segments to create a sparse non-Gaussian signal-plus-noise sample, embeds the 18 recorded sparse waveform channels into a 64-position local probe window using the SpikeInterface channel-sparsity metadata, and adds controlled Gaussian background variation at unobserved local positions. Full mathematical details are given in Supplement Section~\ref{subsec:spe1-preprocess-supp}.
We then apply the cumulant matched-filter localization rule in \eqref{eq:spatial-location-estimator}, using the tapered estimator in \eqref{eq:tapered-cumulant-estimator}, and compare it with an analogous covariance matched-filter baseline. 

The taper bandwidth is chosen adaptively over $k\in\{4,6,\ldots,32\}$ using the stability rule described in Section~\ref{subsec:select-bandwidth}. Figure~\ref{fig:spe1-main} summarizes the resulting localization performance from three complementary views: panel (a) shows median absolute localization error across degradation conditions, panel (b) shows paired differences relative to the covariance baseline, and panel (c) shows a representative score curve. The comparison with the covariance score is paired, since both scores are computed on the same cell-condition samples. We therefore focus on paired error differences in addition to marginal mean and median errors.

Across the 176 cell-condition samples, the tapered cumulant estimator achieves a mean absolute depth error of 254.1 $\mu$m and a median absolute depth error of 246.3 $\mu$m. For comparison, the covariance baseline has mean and median absolute errors of 356.5 $\mu$m and 330.0 $\mu$m, respectively. The tapered cumulant estimate is closer to the ground truth in 99 of the 176 paired samples, with 67 ties and 10 worse cases. Its localization performance remains stable across a range of degradation conditions, suggesting that the adaptively tapered third-order score can retain useful spatial information even when the observed waveforms are substantially degraded.

The selected cell-condition example in Figure~\ref{fig:spe1-main}(c) illustrates this behavior. For cell c28 under degradation condition D11, the ground-truth unit depth is 2373.4 $\mu$m, indicated by the dotted vertical line. The tapered cumulant score attains its maximum at 2380 $\mu$m, corresponding to an absolute localization error of only 6.6 $\mu$m. In the same example, the covariance score selects 2700 $\mu$m, corresponding to an error of 326.6 $\mu$m. The score curve therefore remains closely aligned with the true unit location despite the spatial degradation.

\begin{figure}[ht]
\centering
\includegraphics[width=\linewidth]{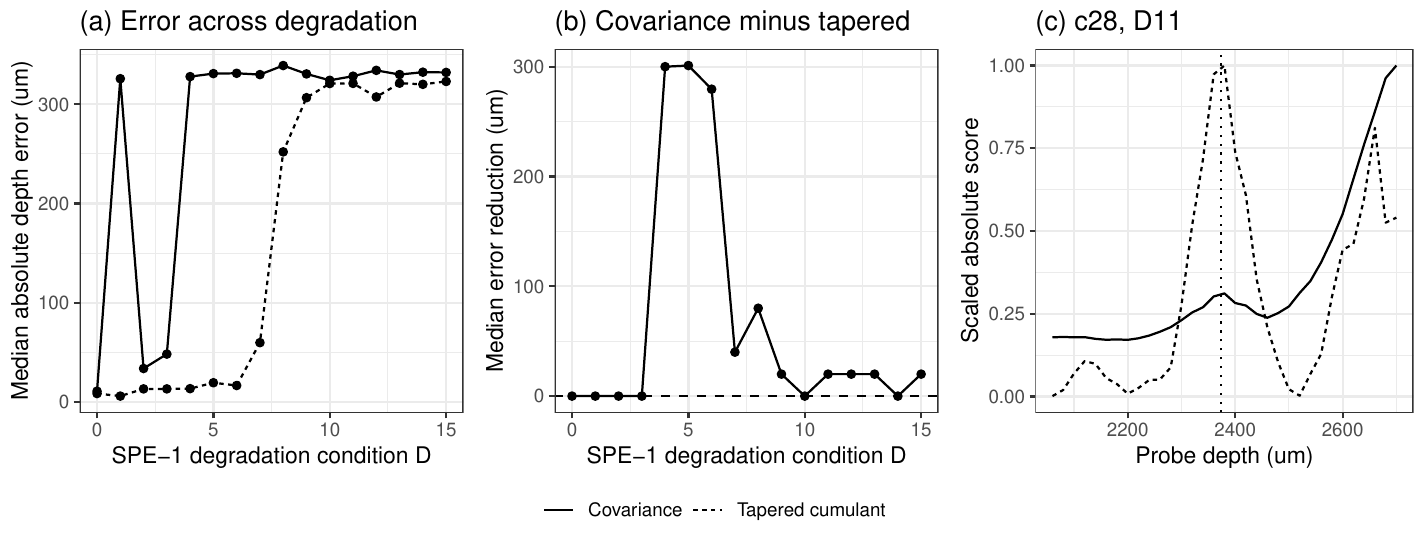}
\caption{SPE-1 localization experiment.}
\label{fig:spe1-main}
\end{figure}

\section{Discussion}
\label{sec:discussion}

This paper develops a bandability framework for estimating higher-order cumulant tensors in ordered high-dimensional problems.
The main idea is that, in many non-Gaussian settings, higher-order dependence is strongest among nearby coordinates and weakens with distance.
Using the diameter of a tensor index to measure this distance leads to a natural extension of bandable covariance models to higher-order cumulants.
The proposed tapered estimator exploits this structure by localizing the empirical cumulant tensor, which improves statistical stability and avoids forming the full $p^d$ array.

The theory gives a bias--variance characterization under tensor spectral norm.
The upper bound separates tapering bias from stochastic error, while the lower bound shows that the same tradeoff is unavoidable over the bandable cumulant class.
Thus, when the effective bandwidth is much smaller than the ambient dimension, the structured estimator can be substantially more accurate than an unstructured cumulant estimator.

The estimator also provides a regularized input for cumulant-based methods.
The autoregressive and moving-average examples show how tapered cumulants can
lead to more stable estimating equations in non-Gaussian time-series models,
and the matched-filter example shows that a tapered third-order score can
localize non-Gaussian sources in spatial sensor arrays more accurately than a
covariance-based score.

Several questions remain open.
One direction is to relax the assumption that the coordinate ordering is known.
Another is to extend the concentration theory from independent observations to dependent time-series data.
It would also be useful to develop adaptive bandwidth-selection theory and scalable proxies for tensor spectral norms in large problems.

\bibliography{reference}
\bibliographystyle{plain}

\appendix

% \crefalias{section}{appendix} % uncomment if you are using cleveref
\newpage

% \section{Additional Figures}

\section{Upper Bounds for Bandable Moment Tensors}

\begin{definition}[Bandable moment tensor]
For $\alpha>0$ and $\beta>0$, we say that an order-$d$ moment tensor $\bcM_d$ is $(\alpha,\beta)$-bandable if
\[
    |(\bcM_d)_{\bi}|
    \le
    \beta\left(1+\operatorname{diam}(\bi)\right)^{-(\alpha+d-1)}
\]
for all $\bi=(i_1,\ldots,i_d)\in[p]^d$.
We denote the corresponding class by
\[
    \mathcal F_{\alpha,\beta}^d
    =
    \left\{
    \bcM_d\in(\R^p)^{\otimes d}:
    |(\bcM_d)_{\bi}|
    \le
    \beta\left(1+\operatorname{diam}(\bi)\right)^{-(\alpha+d-1)},
    \ \forall \bi \in[p]^d
    \right\}.
\]
\end{definition}

Define the tapered estimator as
\[
    (\widehat{\bcM}_{d,T,k})_{\bi}
    =
    w_k\!\left(\operatorname{diam}(\bi)\right)
    (\widehat{\bcM}_d)_{\bi}. 
\]

\begin{theorem}[Upper bound for the tapered moment estimator]
\label{thm:tapered-moment-upper-bound}
Let $d\ge2$ be fixed, and let $X_1,\ldots,X_n$ be i.i.d. copies of a centered random vector $X\in\R^p$.
Let $\bcM_d = \E X^{\otimes d}$. 
Suppose that $\bcM_d\in\mathcal F_{\alpha,\beta}^d$ and, for some $0<\gamma\le d$, 
$
    \sup_{u\in\mathbb S^{p-1}}
    \|u^\top X\|_{\Phi_\gamma}
    \le
    K.
$
Then for every integer $2\le k\le p$ and every $x\ge0$, with probability at least $1-Ce^{-x}$,
\[
    \left\|
    \widehat{\bcM}_{d,T,k}
    -
    \bcM_d
    \right\|
    \le
    C_{\alpha,d}\beta k^{-\alpha-d/2+1}
    +
    C_{\gamma,d,K}
    \left\{
    \sqrt{\frac{k+\log p+x}{n}}
    +
    \frac{(k+\log p+x)^{d/\gamma}}{n}
    \right\}.
\]
Consequently,
\[
    \E
    \left\|
    \widehat{\bcM}_{d,T,k}
    -
    \bcM_d
    \right\|
    \le
    C_{\alpha,d}\beta k^{-\alpha-d/2+1}
    +
    C_{\gamma,d,K}
    \left\{
    \sqrt{\frac{k+\log p}{n}}
    +
    \frac{(k+\log p)^{d/\gamma}}{n}
    \right\}.
\]
\end{theorem}

\section{A Further Plug-in Application: Tapered Bispectrum Estimation}
\label{subsec:tapered-bispectrum}

Section~\ref{subsec:regularized-plugin} describes a broader plug-in perspective: many procedures for non-Gaussian processes begin by estimating third-order cumulants and then use these estimates in a downstream task. This section gives one concrete example through bispectrum estimation. For a non-Gaussian stationary process, the bispectrum is the frequency-domain representation of the third-order lag cumulants and can capture nonlinear dependence and phase information that is not available from the autocovariance or the ordinary power spectrum \citep{nikias1987bispectrum,mendel1991tutorial,nikias1993signal,nikias1993higher,zhang2018asymptotic}. The proposed tapered cumulant estimator can therefore be inserted before the Fourier transform, yielding a regularized lag-window bispectrum estimator.

Let $\{Y_t\}_{t\in\mathbb Z}$ be a centered stationary scalar time series. For lags $h_0,h_1,h_2\in\mathbb Z$, write $\kappa_3(h_0,h_1,h_2)$
for the joint third-order cumulant of $Y_{t-h_0},Y_{t-h_1},Y_{t-h_2}$. Since the process is centered, this equals
\[
    \kappa_3(h_0,h_1,h_2)
    =
    \E\left[
        Y_{t-h_0}Y_{t-h_1}Y_{t-h_2}
    \right].
\]
When the cumulants are absolutely summable, the bispectrum is defined by the two-dimensional Fourier transform
\[
    f_3(\omega_1,\omega_2)
    =
    \frac{1}{(2\pi)^2}
    \sum_{h_1,h_2\in\mathbb Z}
    \gamma(h_1,h_2)
    \exp\{-i(h_1\omega_1+h_2\omega_2)\},
    \qquad
    (\omega_1,\omega_2)\in[-\pi,\pi]^2,
\]
where $\gamma(h_1,h_2)=\kappa_3(0,h_1,h_2)$. 
For Gaussian time series, all third-order cumulants vanish, and hence $f_3\equiv0$. Therefore a nonzero bispectrum is a direct diagnostic for non-Gaussian and asymmetric temporal dependence.

Our tapered cumulant estimator gives a natural regularized estimator of the bispectrum. Given empirical third-order lag cumulants $\widehat\kappa_3(0,h_1,h_2)$, define
\[
    \widehat\gamma_k(h_1,h_2)
    =
    w_k\!\left(\operatorname{diam}(0,h_1,h_2)\right)
    \widehat\kappa_3(0,h_1,h_2),
\]
where
\[
    \operatorname{diam}(0,h_1,h_2)
    =
    \max\{|h_1|,|h_2|,|h_1-h_2|\}.
\]
The tapered bispectrum estimator is then
\[
    \widehat f_{3,k}(\omega_1,\omega_2)
    =
    \frac{1}{(2\pi)^2}
    \sum_{h_1,h_2}
    \widehat\gamma_k(h_1,h_2)
    \exp\{-i(h_1\omega_1+h_2\omega_2)\},
\]
where the sum is taken over the available lag region. The taper keeps cumulants whose lag indices are close to the main diagonal and shrinks cumulants corresponding to widely separated time points. 

The error of the tapered bispectrum can be decomposed into a stochastic cumulant-estimation error and a tapering bias. Formally,
\[
\begin{split}
    \left|
    \widehat f_{3,k}(\omega_1,\omega_2)
    -
    f_3(\omega_1,\omega_2)
    \right|
    &\le
    \frac{1}{(2\pi)^2}
    \sum_{h_1,h_2}
    w_k\!\left(\operatorname{diam}(0,h_1,h_2)\right)
    \left|
    \widehat\kappa_3(0,h_1,h_2)
    -
    \kappa_3(0,h_1,h_2)
    \right|  \\
    &\quad+
    \frac{1}{(2\pi)^2}
    \sum_{h_1,h_2}
    \left\{1-w_k\!\left(\operatorname{diam}(0,h_1,h_2)\right)\right\}
    \left|
    \gamma(h_1,h_2)
    \right|.
\end{split}
\]
The first term is controlled by the accuracy of the empirical third-order cumulants on the retained lag region. The second term is the tapering bias; it includes both the transition region $\lfloor k/2\rfloor<\operatorname{diam}(0,h_1,h_2)<k$ and the tail region $\operatorname{diam}(0,h_1,h_2)\ge k$. This bias is small when the third-order cumulants are bandable or summable. Therefore, the proposed tapering scheme provides a direct way to regularize bispectrum estimation in settings where raw high-order cumulant estimates are unstable at large lags.

\section{Spectral-Norm versus Entrywise Control for Cumulant Yule--Walker Estimation}
\label{sec:supp-yw-spectral-entrywise}

This section compares tensor spectral-norm control with entrywise control for the growing-order cumulant Yule--Walker system.

Let $\bcE$ be an order-$d$ tensor perturbation.
For $\mathcal H=\{h^{(1)},\ldots,h^{(m)}\}$ and working order $r$, define
\[
    \mathcal T_b(\bcE)_a
    =
    \bcE_{1,h_1^{(a)}+1,\ldots,h_{d-1}^{(a)}+1},
\]
and
\[
    \mathcal T_A(\bcE)_{aj}
    =
    \bcE_{j+1,h_1^{(a)}+1,\ldots,h_{d-1}^{(a)}+1},
    \qquad
    a=1,\ldots,m,
    \quad
    j=1,\ldots,r.
\]
Thus, for $\bcE=\widehat{\bcK}-\bcK_d$,
\[
    \widehat b-b=\mathcal T_b(\bcE),
    \qquad
    \widehat A-A=\mathcal T_A(\bcE).
\]
Define the maximum error over the selected entries by
\[
\begin{split}
    \|\bcE\|_{\max,\mathcal H,r}
    =
    \max
    \Bigg\{
    &
    \max_{1\leq a\leq m}
    \left|
    \bcE_{1,h_1^{(a)}+1,\ldots,h_{d-1}^{(a)}+1}
    \right|,
    \\
    &
    \max_{\substack{1\leq a\leq m\\1\leq j\leq r}}
    \left|
    \bcE_{j+1,h_1^{(a)}+1,\ldots,h_{d-1}^{(a)}+1}
    \right|
    \Bigg\},
\end{split}
\]
and let
\[
    s_{A,r}
    =
    \sigma_{\min}(A).
\]

\begin{lemma}[Entrywise perturbation bound]
\label{lem:ar-yw-entrywise-bound}
Let
\[
    \delta
    =
    \|\bcE\|_{\max,\mathcal H,r}.
\]
Then
\[
    \|\mathcal T_b(\bcE)\|_2
    \leq
    \sqrt m\,\delta,
    \qquad
    \|\mathcal T_A(\bcE)\|
    \leq
    \sqrt{mr}\,\delta,
\]
and, for every $v\in\mathbb R^r$,
\[
    \|\mathcal T_A(\bcE)v\|_2
    \leq
    \sqrt m\,\delta\|v\|_1.
\]
If $s_{A,r}>0$ and
\[
    \sqrt{mr}\,\delta
    \leq
    \frac{s_{A,r}}{2},
\]
then $\widehat A$ has full column rank and
\begin{equation}
\label{eq:ar-yw-entrywise-coefficient-bound}
\begin{split}
    \left\|
    \widehat\phi^{(r)}
    -
    \phi^{(r)}
    \right\|_2
    \leq
    \frac{2}{s_{A,r}}
    \left[
    \sqrt m
    \left\{
    1+
    \min
    \left(
    \sqrt r\,\|\phi^{(r)}\|_2,
    \|\phi^{(r)}\|_1
    \right)
    \right\}
    \delta
    +
    \|q_r\|_2
    \right].
\end{split}
\end{equation}
\end{lemma}

\begin{proof}
The first two inequalities follow from
\[
    \|\mathcal T_b(\bcE)\|_2^2
    \leq
    m\delta^2,
    \qquad
    \|\mathcal T_A(\bcE)\|_{\mathrm F}^2
    \leq
    mr\delta^2.
\]
Moreover, for each $a$,
\[
    \left|
    \mathcal T_A(\bcE)_{a,\cdot}v
    \right|
    \leq
    \delta\|v\|_1,
\]
which gives
\[
    \|\mathcal T_A(\bcE)v\|_2
    \leq
    \sqrt m\,\delta\|v\|_1.
\]
Weyl's inequality implies
\[
    \sigma_{\min}(\widehat A)
    \geq
    s_{A,r}
    -
    \|\mathcal T_A(\bcE)\|
    \geq
    \frac{s_{A,r}}{2}.
\]
Hence $\|\widehat A^\dagger\|\leq 2/s_{A,r}$.
Using
\[
    \widehat\phi^{(r)}-\phi^{(r)}
    =
    \widehat A^\dagger
    \left[
    \mathcal T_b(\bcE)
    -
    \mathcal T_A(\bcE)\phi^{(r)}
    +
    q_r
    \right]
\]
and the bounds above proves \eqref{eq:ar-yw-entrywise-coefficient-bound}.
\end{proof}

Tensor spectral-norm control gives
\begin{equation}
\label{eq:ar-yw-extraction-comparison}
    \|\mathcal T_b(\bcE)\|_2
    \leq
    \sqrt m\,\|\bcE\|,
    \qquad
    \|\mathcal T_A(\bcE)\|
    \leq
    \sqrt m\,\|\bcE\|.
\end{equation}
Indeed, the $a$th row satisfies
\[
\begin{split}
    \|\mathcal T_A(\bcE)_{a,\cdot}\|_2
    &=
    \sup_{v\in\mathbb S^{r-1}}
    \left|
    \left\langle
    \bcE,
    \widetilde v
    \otimes
    e_{h_1^{(a)}+1}
    \otimes\cdots\otimes
    e_{h_{d-1}^{(a)}+1}
    \right\rangle
    \right|
    \\
    &\leq
    \|\bcE\|,
\end{split}
\]
where $\widetilde v=\sum_{j=1}^rv_je_{j+1}$.
Consequently, if
\[
    \sqrt m\,\|\bcE\|
    \leq
    \frac{s_{A,r}}{2},
\]
then
\begin{equation}
\label{eq:ar-yw-spectral-coefficient-bound-supp}
    \left\|
    \widehat\phi^{(r)}
    -
    \phi^{(r)}
    \right\|_2
    \leq
    \frac{2}{s_{A,r}}
    \left[
    \sqrt m
    \left(
    1+\|\phi^{(r)}\|_2
    \right)
    \|\bcE\|
    +
    \|q_r\|_2
    \right].
\end{equation}

The corresponding rank-stability conditions are therefore
\[
    \sqrt m\,\|\bcE\|
    \leq
    \frac{s_{A,r}}{2}
\]
under spectral-norm control and
\[
    \sqrt{mr}\,
    \|\bcE\|_{\max,\mathcal H,r}
    \leq
    \frac{s_{A,r}}{2}
\]
under entrywise control.
Thus spectral-norm control avoids the worst-case $\sqrt r$ accumulation arising when entrywise errors are converted to operator norm.

For the raw sample cumulant, let
\[
    \delta_{n,x}^{\mathrm{raw}}
    =
    \left\|
    \widehat{\bcK}_d-\bcK_d
    \right\|_{\max,\mathcal H,r}.
\]
Also, there are
\[
    N_r=m(r+1)
\]
selected entries.
A direct union bound typically gives
\[
    \delta_{n,x}^{\mathrm{raw}}
    \lesssim
    \eta_{N_r,x}
    (1+\eta_{N_r,x})^{d-1},
\]
where
\[
    \eta_{N_r,x}
    =
    \sqrt{
    \frac{\log N_r+x}{n}
    }
    +
    \frac{
    (\log N_r+x)^{d/\gamma}
    }{
    n
    }.
\]
Under $\sum_{j\geq1}|\phi_j|<\infty$, the coefficient error, conditional on rank stability, is therefore
\begin{equation*}
\label{eq:ar-yw-entrywise-leading-rate}
    O_{\mathbb P}
    \left[
    \frac{\sqrt m}{s_{A,r}}
    \left\{
    \sqrt{
    \frac{
    \log\{m(r+1)\}
    }{
    n
    }
    }
    +
    \frac{
    [\log\{m(r+1)\}]^{d/\gamma}
    }{
    n
    }
    \right\}
    +
    \frac{\|q_r\|_2}{s_{A,r}}
    \right].
\end{equation*}
The corresponding rank-stability requirement is
\[
    \frac{\sqrt{mr}}{s_{A,r}}
    \left\{
    \sqrt{
    \frac{
    \log\{m(r+1)\}
    }{
    n
    }
    }
    +
    \frac{
    [\log\{m(r+1)\}]^{d/\gamma}
    }{
    n
    }
    \right\}
    \to0.
\]

For the tapered estimator, Theorem~\ref{thm:cumulant-yule-walker} gives
\begin{equation*}
\label{eq:ar-yw-spectral-leading-rate}
\begin{split}
    \left\|
    \widehat\phi_{d,k}^{(r)}
    -
    \phi^{(r)}
    \right\|_2
    \lesssim
    \frac{\sqrt m}{s_{A,r}}
    \left[
    \beta_{\mathrm{AR}}
    k^{-\alpha-d/2+1}
    +
    \Delta_{k,x}
    (1+\Delta_{k,x})^{d-1}
    \right]
    +
    \frac{\|q_r\|_2}{s_{A,r}}.
\end{split}
\end{equation*}
For sub-Gaussian observations and in the small-error regime, the tensor estimation error has order
\[
    \beta_{\mathrm{AR}}
    k^{-\alpha-d/2+1}
    +
    \sqrt{
    \frac{k+\log p}{n}
    }
    +
    \frac{
    (k+\log p)^{d/2}
    }{
    n
    }.
\]
Ignoring tapering bias and higher-order fluctuation terms, the leading rank-stability requirements are
\[
    \text{raw entrywise:}
    \qquad
    \frac{\sqrt{mr}}{s_{A,r}}
    \sqrt{
    \frac{
    \log\{m(r+1)\}
    }{
    n
    }
    },
\]
versus
\[
    \text{tapered spectral:}
    \qquad
    \frac{\sqrt m}{s_{A,r}}
    \sqrt{
    \frac{
    k+\log p
    }{
    n
    }
    }.
\]
Thus, at the level of the leading square-root terms, the tapered spectral condition is weaker when
\[
    k+\log p
    \ll
    r\log\{m(r+1)\},
\]
provided the tapering bias and higher-order fluctuation terms are negligible.
This comparison concerns sufficient conditions obtained from a union bound for the selected raw cumulant entries and a tensor spectral-norm bound for the tapered estimator; however, it does not imply that spectral-norm estimation uniformly dominates entrywise estimation.
Its principal advantage is that it controls the entire design perturbation uniformly over coefficient directions and avoids the worst-case $\sqrt r$ conversion required by entrywise-only control of rank stability.
Conditional on invertibility, however, the $\ell_1$ refinement in Lemma~\ref{lem:ar-yw-entrywise-bound} shows that the entrywise coefficient bound need not contain an explicit $\sqrt r$ factor.
The advantage is therefore most relevant for growing working order when invertibility must be established from the same cumulant estimation bound.

\section{Additional Details and Results for the Real Data Analyses}

\subsection{RR Interval Preprocessing and Cumulant Yule--Walker Implementation}
\label{subsec:rr-implementation-supp}

This subsection gives the preprocessing and implementation details for Section~\ref{subsec:rr_interval_analysis}.
Let $X_{i,t}$ denote the raw RR interval for subject $i$, where $i=1,\ldots,n$ and $t$ indexes beats within subject.
The implementation first removes physiologically implausible intervals and large local jumps on the logarithmic scale.
Let $\mathcal I_i$ be the retained ordered index set after this deterministic filtering step, and write the retained log intervals as
\[
    U_{i,s}
    =
    \log X_{i,t_s},
    \qquad
    t_s\in\mathcal I_i,\quad s=1,\ldots,T_i.
\]
Each subject is centered and standardized separately,
\[
    Y_{i,s}
    =
    \frac{U_{i,s}-\bar U_i}{\widehat\sigma_i},
    \qquad
    \bar U_i
    =
    T_i^{-1}\sum_{s=1}^{T_i} U_{i,s},
    \qquad
    \widehat\sigma_i^2
    =
    T_i^{-1}\sum_{s=1}^{T_i}(U_{i,s}-\bar U_i)^2 .
\]

For lags $h_0,h_1,h_2\ge0$, let $L(h)=\max(h_0,h_1,h_2)$.
The subject-level centered third cumulant is estimated by
\[
    \widehat\kappa_{3,i}(h_0,h_1,h_2)
    =
    \frac{1}{T_i-L(h)}
    \sum_{s=L(h)+1}^{T_i}
    Y_{i,s-h_0}Y_{i,s-h_1}Y_{i,s-h_2}.
\]
The aggregated cumulant gives equal weight to each subject,
\[
    \widehat\kappa_3(h_0,h_1,h_2)
    =
    \frac1n
    \sum_{i=1}^n
    \widehat\kappa_{3,i}(h_0,h_1,h_2).
\]
For bandwidth $k$, the tapered version is
\[
    \widehat\kappa_{3,k}(h_0,h_1,h_2)
    =
    w_k\{\operatorname{diam}(h_0,h_1,h_2)\}
    \widehat\kappa_3(h_0,h_1,h_2),
\]
where $w_k(\cdot)$ is the taper weight in \eqref{eq:taper-estimator-w} and $\operatorname{diam}(h_0,h_1,h_2)=\max_\ell h_\ell-\min_\ell h_\ell$.

For an AR order $r$, the RR analysis uses the lag set
\[
    \mathcal H_{\mathrm{RR}}
    =
    \{1,\ldots,8\}^2.
\]
For each $h^{(a)}=(h_1^{(a)},h_2^{(a)})\in\mathcal H_{\mathrm{RR}}$, define
\[
    \widehat b_{k,a}
    =
    \widehat\kappa_{3,k}(0,h_1^{(a)},h_2^{(a)})
\]
and
\[
    \widehat A_{k,aj}
    =
    \widehat\kappa_{3,k}(j,h_1^{(a)},h_2^{(a)}),
    \qquad
    j=1,\ldots,r.
\]
These are the empirical versions of the cumulant Yule--Walker equations in \eqref{eq:ar-cumulant-yw-equation}.
The cumulant Yule--Walker coefficient estimate is
\[
    \widehat\phi_{r,k}
    =
    (\widehat A_k^\top\widehat A_k)^{-1}
    \widehat A_k^\top\widehat b_k,
\]
whenever $\widehat A_k$ has full column rank.

The bandwidth grid is
\[
    \mathcal K
    =
    \{3,5,7,9,11,13,15\}.
\]
For $k\le k'$, the equation-level stability loss is
\[
    D(k,k')
    =
    \left(
    \|\widehat A_k-\widehat A_{k'}\|_F^2
    +
    \|\widehat b_k-\widehat b_{k'}\|_2^2
    \right)^{1/2}.
\]
The selected bandwidth is the smallest $k\in\mathcal K$ satisfying the Lepski-type stability rule in \eqref{eq:k-selection}, with $R(k,k')=D(k,k')$.

The autoregressive order is assessed by the residual diagnostic
\[
    L(r)
    =
    \frac{1}{|\mathcal H_{\mathrm{RR}}|}
    \left\|
    \widehat b_{\widehat k}
    -
    \widehat A_{\widehat k,r}\widehat\phi_{r,\widehat k}
    \right\|_2^2,
    \qquad
    r=1,\ldots,5.
\]
The main analysis uses the smallest order after which this diagnostic changes little, which gives $r=3$ for the RR interval data.

For the raw cumulant comparison, let $\widehat A_{\mathrm{raw}}$ and $\widehat b_{\mathrm{raw}}$ be constructed from the same lag pairs with the taper weights replaced by one. The raw cumulant Yule--Walker estimator is
\[
    \widehat\phi_r^{\mathrm{raw}}
    =
    (\widehat A_{\mathrm{raw}}^\top\widehat A_{\mathrm{raw}})^{-1}
    \widehat A_{\mathrm{raw}}^\top\widehat b_{\mathrm{raw}}.
\]
Thus, the raw and tapered cumulant estimators differ only in whether the empirical third-order cumulants entering the same estimating equations are tapered.

For the ordinary covariance Yule--Walker comparison, the subject-level autocovariance estimate is
\[
    \widehat\gamma_i(h)
    =
    \frac{1}{T_i-h}
    \sum_{s=h+1}^{T_i}
    Y_{i,s}Y_{i,s-h},
    \qquad
    h\ge0,
\]
and the aggregated autocovariance is $\widehat\gamma(h)=n^{-1}\sum_{i=1}^n\widehat\gamma_i(h)$.
Let $\widehat\Gamma_r$ be the Toeplitz matrix with $(a,b)$ entry $\widehat\gamma(|a-b|)$ and let $\widehat g_r=(\widehat\gamma(1),\ldots,\widehat\gamma(r))^\top$.
The covariance Yule--Walker estimator is
\[
    \widehat\phi^{\mathrm{cov}}_r
    =
    \widehat\Gamma_r^{-1}\widehat g_r .
\]
The coefficient intervals in Figure~\ref{fig:rr_coefficients} are computed by resampling subjects with replacement and rerunning the same preprocessing-to-estimation map on each bootstrap replicate.

The third-order portmanteau comparison follows the use of squared bicorrelations as diagnostics for nonlinear serial dependence in time-series fitting errors \citep{ashley1986diagnostic,hinich1996testing}, but it is adapted to the particular moment restrictions used here. Let $T_i^{\mathrm{tr}}=\lfloor0.8T_i\rfloor$. The first $T_i^{\mathrm{tr}}$ observations of every subject are used to compute the centering and scaling constants and to estimate common AR(3) coefficient vectors for the tapered cumulant, raw cumulant, and covariance methods. The taper bandwidth is also reselected using only these training segments. The final observations are not used at any stage of model fitting.

For method $M\in\{\mathrm{tap},\mathrm{raw},\mathrm{cov}\}$, define the one-step residual
\[
    \widehat\varepsilon_{i,t}^{M}
    =
    Y_{i,t}
    -
    \sum_{j=1}^3
    \widehat\phi_j^{M}Y_{i,t-j}.
\]
Set $H=8$ and let
\[
    \mathcal T_i^{\mathrm{ev}}
    =
    \{T_i^{\mathrm{tr}}+H+1,\ldots,T_i\}.
\]
Starting the evaluation at $T_i^{\mathrm{tr}}+H+1$ ensures that the response and every lagged variable entering the score belong to the held-out segment. For each unique lag pair $1\le h_1\le h_2\le H$, define the standardized residual--past third-order moment
\[
    \widehat c_i^{M}(h_1,h_2)
    =
    \frac{
    |\mathcal T_i^{\mathrm{ev}}|^{-1}
    \sum_{t\in\mathcal T_i^{\mathrm{ev}}}
    (\widehat\varepsilon_{i,t}^{M}-\overline{\widehat\varepsilon}_i^{M})
    (Y_{i,t-h_1}-\overline Y_{i,h_1})
    (Y_{i,t-h_2}-\overline Y_{i,h_2})
    }{
    s(\widehat\varepsilon_i^{M})
    s(Y_{i,\cdot-h_1})
    s(Y_{i,\cdot-h_2})
    },
\]
where the bars and standard deviations are computed over the common evaluation indices $\mathcal T_i^{\mathrm{ev}}$. Since a third-order cumulant equals the corresponding centered third-order product, $\widehat c_i^{M}(h_1,h_2)$ is the dimensionless empirical version of the restriction $\operatorname{cum}(\varepsilon_{i,t}^{M},Y_{i,t-h_1},Y_{i,t-h_2})=0$.

The subject-level score is
\[
    Q_i^{M}
    =
    \frac{2}{H(H+1)}
    \sum_{1\le h_1\le h_2\le H}
    \{\widehat c_i^{M}(h_1,h_2)\}^2.
\]
The restriction $h_1\le h_2$ removes duplicate equations implied by symmetry of the third-order cumulant. The score equally weights the resulting $H(H+1)/2=36$ distinct restrictions. At diameter $D=\max(h_1,h_2)$, there are exactly $D$ unique lag pairs, so the total contribution of diameter $D$ grows linearly with $D$ without introducing an additional data-tuned weighting exponent. This gives the noisier large-diameter restrictions greater aggregate representation and makes the score directly relevant for assessing the regularization supplied by tapering.

Unlike the formally calibrated portmanteau tests in \cite{ashley1986diagnostic} and \cite{hinich1996testing}, $Q_i^M$ is used only as an out-of-sample comparison score; no null distribution or $p$-value is attached to it. Figure~\ref{fig:rr_third_order_portmanteau} displays the subject-level scores with lines pairing the three methods within subject, together with the across-subject means and 95\% $t$ intervals. The implementation also reports paired subject-bootstrap 95\% intervals for each difference in mean score.

\subsection{Air Quality Preprocessing and Moving-Average Cumulant Implementation}
\label{subsec:ma-air-implementation-supp}

This subsection gives the preprocessing and implementation details for Section~\ref{subsec:ma_air_quality}.
Let $X_t$ denote the raw hourly \texttt{PT08.S1(CO)} sensor response after removing observations with missing sensor values.
The logarithmic series is
\[
    V_t
    =
    \log(X_t).
\]
We remove calendar and low-frequency components by the regression
\[
    V_t
    =
    \alpha
    +
    \sum_{a=1}^{23}
    \gamma_a
    \mathbf 1\{\operatorname{hour}(t)=a\}
    +
    \sum_{b=1}^{6}
    \delta_b
    \mathbf 1\{\operatorname{dow}(t)=b\}
    +
    f(t)
    +
    R_t,
\]
where $f$ is a natural cubic spline in time with 28 degrees of freedom and $R_t$ is the fitted residual.
The analysis then uses the first-differenced residual
\[
    Y_t^\circ
    =
    R_t-R_{t-1}.
\]
To limit the influence of isolated extreme observations on the empirical third-order cumulants, the differenced residuals are winsorized at their empirical $0.001$ and $0.999$ quantiles.
Let $\widetilde Y_t^\circ$ denote the resulting winsorized series.
After centering and standardization,
\[
    Y_t
    =
    \frac{
        \widetilde Y_t^\circ
        -
        \overline{\widetilde Y^\circ}
    }{
        \widehat\sigma_Y
    },
\]
the sequence $Y_t$ is treated as one long realization of a stationary non-Gaussian time series.

For a fixed maximum lag $L$, define
\[
    Z_t^{(L)}
    =
    (Y_t,Y_{t-1},\ldots,Y_{t-L})^\top
\]
and let $\bcK_3^{(L)}$ be the third-order cumulant tensor of $Z_t^{(L)}$.
For lags $0\leq h_1,h_2\leq L$, write
\[
    \kappa_3(0,h_1,h_2)
    =
    (\bcK_3^{(L)})_{1,h_1+1,h_2+1}.
\]
Since the processed series is centered, the empirical cumulant is the time average
\[
    \widehat\kappa_3(0,h_1,h_2)
    =
    \frac{1}{T-L}
    \sum_{t=L+1}^T
    Y_tY_{t-h_1}Y_{t-h_2}.
\]

The Air Quality analysis uses the lag set
\[
    \mathcal H_{\mathrm{air}}
    =
    \left\{
        (h_1,h_2):
        0\leq h_1\leq h_2\leq h_{\max}
    \right\}
    \setminus
    \{(0,0)\},
    \qquad
    h_{\max}=10.
\]
The exclusion of $(0,0)$ prevents the minimum-distance criterion from being dominated by the marginal skewness term.

We use the tapered empirical cumulants with bandwidth $k=10$.
Because
\[
    \operatorname{diam}(0,h_1,h_2)
    =
    \max(h_1,h_2),
\]
the taper weight is
\[
    w_{10}(m)
    =
    \begin{cases}
        1,
        & 0\leq m\leq 5,\\
        \dfrac{10-m}{5},
        & 5<m<10,\\
        0,
        & m\geq 10.
    \end{cases}
\]
For $(h_1,h_2)\in\mathcal H_{\mathrm{air}}$, define
\[
    \widehat\kappa_{3,T,10}(0,h_1,h_2)
    =
    w_{10}\left(\max(h_1,h_2)\right)
    \widehat\kappa_3(0,h_1,h_2),
\]
and let $\widehat b_{T,10}$ collect these tapered empirical cumulants over $\mathcal H_{\mathrm{air}}$.

In addition to tapering the cumulants, we balance the minimum-distance equations across lag diameters.
For $r=1,\ldots,10$, let
\[
    N_r
    =
    \left|
        \left\{
            (h_1,h_2)\in\mathcal H_{\mathrm{air}}:
            \max(h_1,h_2)=r
        \right\}
    \right|.
\]
The equation corresponding to $(h_1,h_2)$ receives weight
\[
    \omega_{h_1,h_2}
    =
    \frac{1}{
        N_{\max(h_1,h_2)}
    }.
\]
Thus every lag diameter has the same total weight in the criterion.
Let $\Omega$ be the diagonal matrix containing these equation weights, and define
\[
    \langle u,v\rangle_\Omega
    =
    u^\top\Omega v,
    \qquad
    \|u\|_\Omega^2
    =
    u^\top\Omega u.
\]
The equation weights $\omega_{h_1,h_2}$ are distinct from the taper weights $w_{10}(\max(h_1,h_2))$: the former balance the minimum-distance equations, whereas the latter regularize the empirical cumulants.

For candidate moving-average order $q$, the fitted model is
\[
    Y_t
    =
    \sum_{\ell=0}^q
    \theta_\ell\varepsilon_{t-\ell},
    \qquad
    \theta_0=1,
\]
where the innovations are centered and have third cumulant $\tau_3$.
By the moving-average cumulant identity \eqref{eq:ma-cumulant-identity}, the model-implied third cumulant for $(h_1,h_2)\in\mathcal H_{\mathrm{air}}$ is
\[
    \tau_3
    \Psi_{h_1,h_2}(\theta)
    =
    \tau_3
    \sum_{s\in\mathbb Z}
    \theta_s
    \theta_{s-h_1}
    \theta_{s-h_2},
\]
where $\theta_0=1$ and $\theta_s=0$ for $s\notin\{0,\ldots,q\}$.
Let $\Psi(\theta)$ collect these entries over $\mathcal H_{\mathrm{air}}$.
For fixed $\theta$ with $\Psi(\theta)\ne0$, the innovation cumulant is profiled out as
\[
    \widehat\tau_{3,10}(\theta)
    =
    \frac{
        \left\langle
            \widehat b_{T,10},
            \Psi(\theta)
        \right\rangle_\Omega
    }{
        \|\Psi(\theta)\|_\Omega^2
    }.
\]
The moving-average coefficient estimator solves
\[
    \widehat\theta_{q,10}
    \in
    \argmin_{\theta\in\Theta_q}
    \left\|
        \widehat b_{T,10}
        -
        \widehat\tau_{3,10}(\theta)
        \Psi(\theta)
    \right\|_\Omega^2,
\]
which is the weighted tapered version of the third-order minimum-distance problem in \eqref{eq:ma-op-theta}.

The optimization is carried out by multistart L-BFGS-B with 80 initial values and componentwise box constraints $|\theta_j|\leq0.97$ \citep{byrd1995limited,nocedal2006numerical}.
The feasible set $\Theta_q$ is restricted to invertible moving-average coefficients.
In computation, parameter values for which the moving-average polynomial has a root with modulus at most $1.03$ are assigned a large penalty.

For $q=1,\ldots,10$, define the weighted tapered residual loss
\[
    L_T(q)
    =
    \frac{
        \left\|
            \widehat b_{T,10}
            -
            \widehat\tau_{3,10}(\widehat\theta_{q,10})
            \Psi(\widehat\theta_{q,10})
        \right\|_\Omega^2
    }{
        \sum_{(h_1,h_2)\in\mathcal H_{\mathrm{air}}}
        \omega_{h_1,h_2}
    }.
\]
The BIC-type order criterion is
\[
    \operatorname{IC}_T(q)
    =
    \log\left(
        L_T(q)+\epsilon
    \right)
    +
    \lambda
    \frac{
        (q+1)\log(n_{\mathrm{eff}})
    }{
        n_{\mathrm{eff}}
    },
\]
where $q+1$ counts the $q$ moving-average coefficients and the profiled innovation cumulant, $\lambda=1$, $\epsilon=10^{-12}$, and
\[
    n_{\mathrm{eff}}
    =
    \max\left(
        10,
        \left\lfloor
            \frac{T}{72}
        \right\rfloor
    \right).
\]
Here $n_{\mathrm{eff}}$ approximates the number of length-72 dependence blocks in the processed trajectory.
The selected order is the minimizer of this criterion, which gives $\widehat q=4$ in the Air Quality analysis.

The bootstrap intervals in Figure~\ref{fig:ma_air_fit_all} are based on 100 circular moving-block bootstrap replicates with block length 72.
Each bootstrap series is recentered and restandardized, after which the tapered empirical cumulants are recomputed using the same bandwidth $k=10$ and the selected MA(4) minimum-distance objective is refitted.
The reported intervals are the $2.5\%$ and $97.5\%$ empirical quantiles of the successful bootstrap estimates.

\subsection{SPE-1 Preprocessing and Localization Implementation}
\label{subsec:spe1-preprocess-supp}

This subsection gives the exact construction used in Section~\ref{subsec:spe1-spatial}.
The external preprocessing reference is the SpikeInterface-style waveform representation in the SPE-1 archive \citep{spe1figshare,buccino2020spikeinterface}; the local-grid embedding and Gaussian background described below are experiment-specific stress-test steps.

For each cell-condition $s$, let
\[
    W^{(s)}_{j,t,m},
    \qquad
    j=1,\ldots,N_s,\quad
    t=1,\ldots,T,\quad
    m=1,\ldots,q_s,
\]
denote the extracted waveform array, where $j$ indexes spike events, $t$ indexes waveform time, and $m$ indexes the sparse waveform channels.
In this experiment, $T=90$ and $q_s=18$ for all retained cell-conditions.
Let $B=\{1,\ldots,\lfloor 0.2T\rfloor\}$ be the pre-spike baseline window.
We first baseline-correct each waveform by
\[
    \widetilde W^{(s)}_{j,t,m}
    =
    W^{(s)}_{j,t,m}
    -
    |B|^{-1}\sum_{u\in B} W^{(s)}_{j,u,m}.
\]

Let $A^{(s)}_{t,r}$ be the average template stored in the archive on the full probe channel index $r$.
The SpikeInterface sparsity file maps the sparse waveform channel $m$ to a full probe channel $r_m$.
We therefore use $A^{(s)}_{t,r_m}$ when aligning sparse waveforms with full-probe geometry.
Let
\[
    (t_s^\star,m_s^\star)
    =
    \argmax_{t,m}
    |A^{(s)}_{t,r_m}|
\]
be the template peak time and reference sparse channel.
For each spike event $j$, draw $Z_j\sim \mathrm{Bernoulli}(\pi)$ independently with $\pi=0.25$.
If $Z_j=1$, set $U_j=t_s^\star$; otherwise draw $U_j$ uniformly from $B$.
The sparse spatial snapshot is
\[
    X^{\mathrm{sp}}_{j,m}
    =
    \widetilde W^{(s)}_{j,U_j,m},
    \qquad
    m=1,\ldots,q_s.
\]
The Bernoulli probability is kept away from $1/2$ so that the centered spike/baseline mixture has a nonzero third cumulant.

Next we embed the sparse channels into a local probe window.
Let $y_r$ and $x_r$ denote the vertical and horizontal coordinates of full probe channel $r$.
We select $p=64$ full probe channels in a local window around the largest template amplitudes and sort them by $(y_r,x_r)$.
Write the selected ordered full-probe indices as $r_1,\ldots,r_p$ and define
\[
    \mathcal A_s
    =
    \{a\in[p]: r_a=r_m\text{ for some sparse waveform channel }m\}.
\]
For observed local positions $a\in\mathcal A_s$, we place the corresponding sparse snapshot into the local grid.
For unobserved positions $a\notin\mathcal A_s$, we draw independent Gaussian background variables with standard deviation $1.25\,\widehat\sigma_s$, where $\widehat\sigma_s$ is the median channelwise standard deviation of the observed sparse snapshots.
Thus the local-grid observation before final scaling is
\[
    X^{\mathrm{loc}}_{j,a}
    =
    \begin{cases}
    X^{\mathrm{sp}}_{j,m}, & r_a=r_m\text{ for some }m,\\
    \eta_{j,a}, & a\notin\mathcal A_s,
    \end{cases}
    \qquad
    \eta_{j,a}\sim N(0,1.25^2\widehat\sigma_s^2).
\]
We center each local coordinate and divide all entries by one global empirical standard deviation, giving $\bar X_{j,a}$.
Finally, to stress covariance-based localization, we add independent Gaussian noise with location-dependent standard deviation
\[
    d_s^\star
    =
    y_{r_{a^\star}},
    \qquad
    a^\star
    \in
    \argmax_{a\in[p]} |y_{r_a}-c_s|,
\]
and
\[
    \sigma_a
    =
    0.75
    \left[
    1
    +
    3
    \exp\left\{
    -
    \frac{(y_{r_a}-d_s^\star)^2}{2(80)^2}
    \right\}
    \right],
\]
where $c_s$ is the ground-truth unit depth.
The final observation is
\[
    X_{j,a}
    =
    \bar X_{j,a}
    +
    \xi_{j,a},
    \qquad
    \xi_{j,a}\sim N(0,\sigma_a^2).
\]
After adding this noise, the observations are centered and globally scaled once more.
The resulting sample is $X_1,\ldots,X_{n_s}\in\mathbb R^p$, where $n_s\le 500$ is the fitting subsample size and $p=64$.

The tapered cumulant score follows the source-localization plug-in rule in \eqref{eq:spatial-location-estimator}.
For each candidate grid position $c\in[p]$, define a normalized spatial template
\[
    g_c(a)
    =
    \frac{\exp\{-|y_{r_a}-y_{r_c}|/80\}\mathbf 1\{|y_{r_a}-y_{r_c}|\le 500\}}
    {\left[\sum_{b=1}^p \exp\{-2|y_{r_b}-y_{r_c}|/80\}\mathbf 1\{|y_{r_b}-y_{r_c}|\le 500\}\right]^{1/2}}.
\]
Let $\widehat{\mathcal K}_{3,T,k}$ be the third-order tapered sample cumulant estimator in \eqref{eq:tapered-cumulant-estimator}, computed from $X_1,\ldots,X_{n_s}$ using the taper weights in \eqref{eq:taper-estimator-w}.
For a fixed $k$, the tapered matched-filter score is
\[
    \widehat Q_k(c)
    =
    \left\langle
    \widehat{\mathcal K}_{3,T,k},
    g_c^{\otimes 3}
    \right\rangle.
\]
In computation we use the equivalent range-decomposed formula
\[
    \widehat Q_k(c)
    =
    \sum_{r=0}^{p-1}
    w_k(r)\,
    \widehat S_r(c),
\]
where $\widehat S_r(c)$ is the sum of empirical third-order cumulant contributions whose three indices have diameter $r$ after weighting by $g_c$.
This avoids forming the full $p^3$ empirical cumulant tensor.
The selected bandwidth $\widehat k$ is chosen from $\{4,6,\ldots,32\}$ by the Lepski-type rule in Section~\ref{subsec:select-bandwidth}, using the same stability threshold as in the numerical implementation.
The tapered location estimate is
\[
    \widehat c_{\mathrm{tap}}
    \in
    \argmax_{c\in[p]}
    |\widehat Q_{\widehat k}(c)|.
\]
We report the estimated depth $y_{r_{\widehat c_{\mathrm{tap}}}}$ and compare it with the ground-truth depth $c_s$ from the archive.

The covariance baseline uses the same spatial templates and the same preprocessed observations.
Let
\[
    \widehat\Sigma
    =
    n_s^{-1}
    \sum_{j=1}^{n_s}
    X_jX_j^\top
\]
be the empirical covariance matrix.
The covariance matched-filter score is
\[
    \widehat Q_{\mathrm{cov}}(c)
    =
    g_c^\top \widehat\Sigma g_c
    =
    n_s^{-1}
    \sum_{j=1}^{n_s}
    (X_j^\top g_c)^2.
\]
The covariance location estimate is
\[
    \widehat c_{\mathrm{cov}}
    \in
    \argmax_{c\in[p]}
    \widehat Q_{\mathrm{cov}}(c).
\]
Both methods are therefore evaluated on the same observations, same depth grid, and same template family; the only difference is whether the matched filter uses third-order tapered cumulants or second-order covariance.

\subsection{Raw and Tapered Third-order Cumulant Scores}
\label{subsec:spe1-raw-taper-supp}

This subsection compares the proposed tapered third-order score with the corresponding untapered score on the same SPE-1 preprocessed observations.
Let $\widehat{\mathcal K}_3$ denote the ordinary plug-in third-order sample cumulant tensor computed from $X_1,\ldots,X_{n_s}$.
The raw cumulant matched-filter score is
\[
    \widehat Q_{\mathrm{raw}}(c)
    =
    \left\langle
    \widehat{\mathcal K}_3,
    g_c^{\otimes 3}
    \right\rangle.
\]
Using the range decomposition introduced in Supplement~\ref{subsec:spe1-preprocess-supp}, this can be written as
\[
    \widehat Q_{\mathrm{raw}}(c)
    =
    \sum_{r=0}^{p-1}
    \widehat S_r(c).
\]
The tapered score replaces the all-one raw weights by the taper weights in \eqref{eq:taper-estimator-w},
\[
    \widehat Q_{\mathrm{tap}}(c)
    =
    \sum_{r=0}^{p-1}
    w_{\widehat k}(r)\widehat S_r(c),
\]
where $\widehat k$ is selected by the same stability rule used in the main experiment.
Thus the raw score and tapered score differ only in how much long-range third-order interaction is used in the matched filter.

Figure~\ref{fig:spe1-raw-taper-supp}(a) shows that the two procedures have similar localization accuracy across degradation conditions.
Over the 176 cell-condition samples, the raw cumulant has mean/median absolute errors 249.7/202.7 $\mu$m, while the tapered cumulant has mean/median absolute errors 254.1/246.3 $\mu$m.
In paired comparisons, the two methods select the same grid point in 147 samples.
The tapered score is more accurate in 11 samples, and the raw score is more accurate in 18 samples.
These numbers indicate that tapering does not substantially change the aggregate localization error in this experiment.

The main advantage of tapering is computational.
The raw score requires all spatial ranges $r=0,\ldots,p-1$, while the tapered score only needs ranges with $r<\widehat k$ because $w_{\widehat k}(r)=0$ for $r\ge \widehat k$.
For the range-decomposed implementation, the raw score therefore stores approximately $p^2$ range summaries, whereas the tapered score stores approximately $p\widehat k$ range summaries.
In this experiment, $p=64$ and the median selected bandwidth is $\widehat k=12$, giving a median raw-to-tapered storage ratio of 5.3 for the range summaries.
For an explicit third-order tensor representation, the analogous comparison is $p^3$ raw tensor entries versus $O(p\widehat k^2)$ tapered entries, with a median raw-to-tapered proxy ratio of 28.4.

We also measured the score-construction cost on the selected \texttt{c28}, \texttt{D11} example.
For this sample, $n_s=354$, $p=64$, and $\widehat k=14$.
The raw score construction took 7.84 seconds and used an 11.6 MB flat score object.
The tapered construction took 2.97 seconds and used a 2.54 MB flat score object.
This corresponds to 2.6$\times$ less time and 4.6$\times$ less memory for the score-building step.

Figure~\ref{fig:spe1-raw-taper-supp}(b) shows why the tapered score can also be preferable for individual samples.
In the \texttt{c28}, \texttt{D11} example, the raw cumulant score has a large off-target peak at 2660 $\mu$m.
The ground-truth depth is 2373.4 $\mu$m, so the raw estimate has error 286.6 $\mu$m.
The tapered score suppresses the off-target peak and selects 2380 $\mu$m, giving error 6.6 $\mu$m.
Thus, in this example, tapering preserves the useful local third-order signal while removing a long-range false peak.

\begin{figure}[ht]
\centering
\includegraphics[width=.8\linewidth]{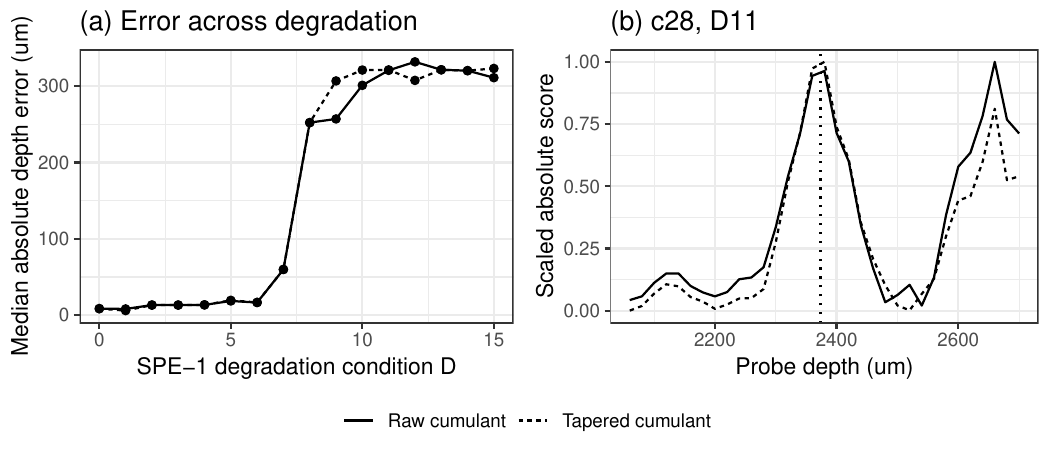}
\caption{Raw and tapered cumulant comparison.}
\label{fig:spe1-raw-taper-supp}
\end{figure}

\section{Proofs}

\subsection{Proof of Theorem \ref{thm:tapered-moment-upper-bound}}

\begin{proof}

Let $\mathcal W_k$ be the deterministic weight tensor defined by
\[
    (\mathcal W_k)_{\bi}
    =
    w_k(\operatorname{diam}(\bi)).
\]
Then
\[
    \widehat{\bcM}_{d,T,k}
    =
    \mathcal W_k*\widehat{\bcM}_d.
\]
Define
\[
    \mathcal E
    =
    \widehat{\bcM}_d
    -
    \bcM_d.
\]
By the triangle inequality,
\begin{equation}
\label{eq:tapered_triangle_decomposition}
    \left\|
    \widehat{\bcM}_{d,T,k}
    -
    \bcM_d
    \right\|
    \le
    \|\mathcal W_k*\mathcal E\|
    +
    \|(\mathcal W_k-\mathbf 1)*\bcM_d\|.
\end{equation}

We first control the bias term.
Set
\[
    \mathcal R_k
    =
    (\mathcal W_k-\mathbf 1)*\bcM_d.
\]
Let $h_k=\lfloor k/2\rfloor$. Since $w_k(m)=1$ for $0\le m\le h_k$,
\[
    (\mathcal R_k)_{\bi}=0
    \qquad
    \text{whenever }
    \operatorname{diam}(\bi)\le h_k.
\]
For $q\ge0$, define the dyadic annular tensor $\mathcal R_{k,q}$ by
\[
    (\mathcal R_{k,q})_{\bi}
    =
    (\mathcal R_k)_{\bi}
    \mathbf I
    \left\{
    2^{q-1}k
    <
    1+\operatorname{diam}(\bi)
    \le
    2^q k
    \right\},
\]
where, for $q=0$, the lower threshold $2^{q-1}k$ is interpreted as $h_k$.
Then
\[
    \mathcal R_k
    =
    \sum_{q\ge0}
    \mathcal R_{k,q}.
\]
For every $\bi$ in the support of $\mathcal R_{k,q}$,
\[
    1+\operatorname{diam}(\bi)
    >
    2^{q-1}k.
\]
Since $\bcM_d\in\mathcal F_{\alpha,\beta}^d$ and $|1-w_k(m)|\le1$,
\[
    |(\mathcal R_{k,q})_{\bi}|
    \le
    \beta
    (2^{q-1}k)^{-(\alpha+d-1)}
    \le
    C_{\alpha,d}\beta
    (2^q k)^{-(\alpha+d-1)}.
\]
Moreover, $\mathcal R_{k,q}$ is supported on entries satisfying
\[
    \operatorname{diam}(\bi)
    \le
    2^q k.
\]
Applying Lemma~\ref{lem:local-support-tensor-norm} with
\[
    m=2^q k,
    \qquad
    a=C_{\alpha,d}\beta(2^q k)^{-(\alpha+d-1)},
\]
gives
\[
\begin{split}
    \|\mathcal R_{k,q}\|
    &\le
    C_d
    \beta
    (2^q k)^{-(\alpha+d-1)}
    (2^q k)^{d/2}  \\
    &=
    C_d
    \beta
    (2^q k)^{-\alpha-d/2+1}.
\end{split}
\]
Since $\alpha+d/2-1>0$, summing over $q\ge0$ yields
\begin{equation}
\label{eq:tapered_bias_bound}
\begin{split}
    \|(\mathcal W_k-\mathbf 1)*\bcM_d\|
    &\le
    \sum_{q\ge0}
    \|\mathcal R_{k,q}\|  \\
    &\le
    C_{\alpha,d}
    \beta
    k^{-\alpha-d/2+1}.
\end{split}
\end{equation}

We now control the stochastic term.
For an integer $m\ge1$ and an integer $\ell$, define
\[
    I_{\ell,m}
    =
    \{j\in[p]:\ell\le j<\ell+m\},
    \qquad
    \ell=1-m,\ldots,p.
\]
Let $\Pi_{\ell,m}$ denote the projection that keeps only entries whose indices all belong to $I_{\ell,m}$.
For any tensor $\mathcal A\in(\R^p)^{\otimes d}$, define
\[
    S_m(\mathcal A)
    =
    \sum_{\ell=1-m}^p
    \Pi_{\ell,m}(\mathcal A).
\]
Let $a(\bi)=\min_{j\in[d]}i_j$ and $b(\bi)=\max_{j\in[d]}i_j$.
The interval $I_{\ell,m}$ contains all coordinates of $\bi$ if and only if
\[
    b(\bi)-m+1
    \le
    \ell
    \le
    a(\bi).
\]
Therefore the number of such intervals is
\[
    \left(m-\operatorname{diam}(\bi)\right)_+.
\]
For any tensor $\mathcal A$, the preceding counting argument gives
\[
    S_m(\mathcal A)_{\bi}
    =
    \left(m-\operatorname{diam}(\bi)\right)_+
    \mathcal A_{\bi}.
\]
Therefore, for $r=\operatorname{diam}(\bi)$,
\[
    \frac{
    (k-r)_+-(h_k-r)_+
    }{
    k-h_k
    }
    =
    w_k(r).
\]
Hence
\begin{equation}
\label{eq:tapered_weight_identity}
    \mathcal W_k*\mathcal A
    =
    \frac{1}{k-h_k}
    \left\{
    S_k(\mathcal A)
    -
    S_{h_k}(\mathcal A)
    \right\}.
\end{equation}
Applying \eqref{eq:tapered_weight_identity} to $\mathcal E$ gives
\begin{equation}
\label{eq:tapered_stochastic_split}
    \|\mathcal W_k*\mathcal E\|
    \le
    \frac{1}{k-h_k}
    \|S_k(\mathcal E)\|
    +
    \frac{1}{k-h_k}
    \|S_{h_k}(\mathcal E)\|.
\end{equation}

For fixed $m$, split the intervals $I_{\ell,m}$ according to the residue class of $\ell$ modulo $m$.
Within each residue class, the intervals are disjoint.
By Lemma~\ref{lem:block-diagonal-tensor}, for each residue class, the spectral norm of the sum of the corresponding block tensors is bounded by the maximum block norm.
Since there are at most $m$ residue classes,
\begin{equation}
\label{eq:tapered_Sm_block_bound}
    \|S_m(\mathcal E)\|
    \le
    m
    \max_{\ell}
    \|\Pi_{\ell,m}(\mathcal E)\|.
\end{equation}
Combining \eqref{eq:tapered_stochastic_split} and \eqref{eq:tapered_Sm_block_bound}, and using $h_k\le k$ and $k-h_k\ge k/2$,
\begin{equation}
\label{eq:tapered_weighted_noise_local_bound}
    \|\mathcal W_k*\mathcal E\|
    \le
    3
    \max_{m\in\{k,h_k\}}
    \max_{\ell}
    \|\Pi_{\ell,m}(\mathcal E)\|.
\end{equation}
By Lemma~\ref{lem:local-empirical-moment-bound}, applied with $m=k$ and $m=h_k$, and by a union bound over the two choices of $m$, \eqref{eq:tapered_weighted_noise_local_bound} implies that, with probability at least $1-Ce^{-x}$,
\begin{equation}
\label{eq:tapered_stochastic_bound}
    \|\mathcal W_k*\mathcal E\|
    \le
    C_{\gamma,d,K}
    \left\{
    \sqrt{\frac{k+\log p+x}{n}}
    +
    \frac{(k+\log p+x)^{d/\gamma}}{n}
    \right\}.
\end{equation}
Combining \eqref{eq:tapered_triangle_decomposition}, \eqref{eq:tapered_bias_bound}, and \eqref{eq:tapered_stochastic_bound} proves the high-probability bound.

It remains to prove the expectation bound.
Let
\[
    Y
    =
    \left\|
    \widehat{\bcM}_{d,T,k}
    -
    \bcM_d
    \right\|,
    \qquad
    a
    =
    k+\log p,
    \qquad
    q
    =
    \frac d\gamma.
\]
The high-probability bound implies that, after enlarging constants,
\[
    \PP
    \left\{
    Y
    >
    C_{\alpha,d}\beta k^{-\alpha-d/2+1}
    +
    C_{\gamma,d,K}
    \left(
    \sqrt{\frac{a+x}{n}}
    +
    \frac{(a+x)^q}{n}
    \right)
    \right\}
    \le
    Ce^{-x}
\]
for all $x\ge0$.
Since $q\ge1$,
\[
    \sqrt{a+x}
    \le
    \sqrt a+\sqrt x,
    \qquad
    (a+x)^q
    \le
    2^{q-1}(a^q+x^q).
\]
Therefore,
\[
    \PP
    \left\{
    Y
    >
    B
    +
    C_{\gamma,d,K}
    \left(
    \sqrt{\frac{x}{n}}
    +
    \frac{x^q}{n}
    \right)
    \right\}
    \le
    Ce^{-x},
\]
where
\[
    B
    =
    C_{\alpha,d}\beta k^{-\alpha-d/2+1}
    +
    C_{\gamma,d,K}
    \left(
    \sqrt{\frac{a}{n}}
    +
    \frac{a^q}{n}
    \right).
\]
Let
\[
    W
    =
    (Y-B)_+,
    \qquad
    h(x)
    =
    \sqrt{\frac{x}{n}}
    +
    \frac{x^q}{n}.
\]
The preceding display gives
\[
    \PP\{W>C_{\gamma,d,K}h(x)\}
    \le
    Ce^{-x},
    \qquad
    x\ge0.
\]

By tail integration, we have
\[
    \E W
    =
    \int_0^\infty \PP(W>t)\,dt .
\]
Since
\[
    \PP\{W>C_{\gamma,d,K}h(x)\}
    \le
    Ce^{-x},
    \qquad x\ge0,
\]
and $h$ is nondecreasing, the usual change-of-variables argument gives
\[
    \E W
    \le
    C_{\gamma,d,K}
    \E h(\xi),
\]
where $\xi$ is a nonnegative random variable satisfying
\[
    \PP(\xi>x)\le Ce^{-x},
    \qquad x\ge0.
\]
Therefore,
\[
    \E W
    \le
    C_{\gamma,d,K}
    \left(
    \frac{\E \xi^{1/2}}{\sqrt n}
    +
    \frac{\E \xi^q}{n}
    \right)
    \le
    C_{\gamma,d,K}
    \left(
    \frac1{\sqrt n}
    +
    \frac1n
    \right),
\]
because $\E \xi^{1/2}\le C$ and $\E \xi^q\le C_q$. 

Since $a=k+\log p\ge1$, we have
\[
    \frac1{\sqrt n}
    +
    \frac1n
    \le
    C
    \left(
    \sqrt{\frac{a}{n}}
    +
    \frac{a^q}{n}
    \right).
\]
Therefore,
\[
    \E Y
    \le
    C_{\alpha,d}\beta k^{-\alpha-d/2+1}
    +
    C_{\gamma,d,K}
    \left(
    \sqrt{\frac{k+\log p}{n}}
    +
    \frac{(k+\log p)^{d/\gamma}}{n}
    \right).
\]
This proves the expectation bound and completes the proof.
\end{proof}

\subsection{Proof of Theorem \ref{thm:tapered-cumulant-upper-bound}}

\begin{proof}
Let $\mathcal W_k$ be the deterministic weight tensor defined in the proof of Theorem~\ref{thm:tapered-moment-upper-bound}, and define
\[
    \mathcal E_{\bcK}
    =
    \widehat{\bcK}_d-\bcK_d.
\]
Since $\widehat{\bcK}_{d,T,k}=\mathcal W_k*\widehat{\bcK}_d$, the triangle inequality gives
\begin{equation}
\label{eq:tapered_cumulant_triangle}
    \left\|
    \widehat{\bcK}_{d,T,k}
    -
    \bcK_d
    \right\|
    \le
    \|\mathcal W_k*\mathcal E_{\bcK}\|
    +
    \|(\mathcal W_k-\mathbf 1)*\bcK_d\|.
\end{equation}

We first control the bias term. Since $\bcK_d\in\mathcal K_{\alpha,\beta}^d$, the dyadic annulus argument leading to \eqref{eq:tapered_bias_bound}, with $\bcM_d$ replaced by $\bcK_d$, gives
\begin{equation}
\label{eq:tapered_cumulant_bias}
    \|(\mathcal W_k-\mathbf 1)*\bcK_d\|
    \le
    C_{\alpha,d}
    \beta
    k^{-\alpha-d/2+1}.
\end{equation}
This step uses only the entrywise bandability of the target tensor and the deterministic tapering weights.

We now control the stochastic term. For an integer $m\ge1$ and an integer $\ell$, define
\[
    I_{\ell,m}
    =
    \{j\in[p]:\ell\le j<\ell+m\},
    \qquad
    \ell=1-m,\ldots,p.
\]
Let $\Pi_{\ell,m}$ denote the order-$d$ tensor projection that keeps only entries whose indices all belong to $I_{\ell,m}$. For any tensor $\bcA\in(\R^p)^{\otimes d}$, define
\[
    S_m(\bcA)
    =
    \sum_{\ell=1-m}^p
    \Pi_{\ell,m}(\bcA).
\]
By the counting identity \eqref{eq:tapered_weight_identity}, for any order-$d$ tensor $\bcA$,
\begin{equation}
\label{eq:tapered_cumulant_weight_identity}
    \mathcal W_k*\bcA
    =
    \frac{1}{k-h_k}
    \left\{
    S_k(\bcA)
    -
    S_{h_k}(\bcA)
    \right\}.
\end{equation}
Applying \eqref{eq:tapered_cumulant_weight_identity} to $\mathcal E_{\bcK}$ gives
\[
    \|\mathcal W_k*\mathcal E_{\bcK}\|
    \le
    \frac{1}{k-h_k}
    \|S_k(\mathcal E_{\bcK})\|
    +
    \frac{1}{k-h_k}
    \|S_{h_k}(\mathcal E_{\bcK})\|.
\]
For fixed $m$, the same residue-class block decomposition used in \eqref{eq:tapered_Sm_block_bound} gives
\[
    \|S_m(\mathcal E_{\bcK})\|
    \le
    m
    \max_{\ell}
    \|\Pi_{\ell,m}(\mathcal E_{\bcK})\|.
\]
Therefore, as in \eqref{eq:tapered_weighted_noise_local_bound},
\begin{equation}
\label{eq:tapered_cumulant_local_reduction}
    \|\mathcal W_k*\mathcal E_{\bcK}\|
    \le
    3
    \max_{m\in\{k,h_k\}}
    \max_{\ell}
    \|\Pi_{\ell,m}(\mathcal E_{\bcK})\|.
\end{equation}

It remains to prove a local empirical cumulant bound. Fix $m\in[p]$ and $x\ge0$. For an interval $I\subset[p]$, let $\Pi_I^{(r)}$ denote the order-$r$ tensor projection that keeps only entries whose indices all belong to $I$. Define
\[
    \mathcal I_m
    =
    \{I\subset[p]: I \text{ is an interval and } |I|\le m\}
\]
and
\begin{equation}
\label{eq:local_delta_mx}
    \Delta_{m,x}
    =
    \sqrt{\frac{m+\log p+x}{n}}
    +
    \frac{(m+\log p+x)^{d/\gamma}}{n}.
\end{equation}
For $r=1,\ldots,d$, define
\[
    \mathcal E_r
    =
    \widehat{\bcM}_r-\bcM_r.
\]
By Lemma~\ref{lem:uniform-local-empirical-moment-all-orders}, with probability at least $1-Ce^{-x}$,
\begin{equation}
\label{eq:uniform_local_moment_errors_all_orders}
    \max_{1\le r\le d}
    \max_{I\in\mathcal I_m}
    \|\Pi_I^{(r)}(\mathcal E_r)\|
    \le
    C_{\gamma,d,K}
    \Delta_{m,x}.
\end{equation}

We also need deterministic bounds for the population moment tensors. For every $1\le r\le d$ and every interval $I\subset[p]$,
\[
\begin{split}
    \|\Pi_I^{(r)}(\bcM_r)\|
    &\le
    \|\bcM_r\|  \\
    &=
    \sup_{u_1,\ldots,u_r\in\mathbb S^{p-1}}
    \left|
    \E
    \prod_{a=1}^r u_a^\top X
    \right|  \\
    &\le
    \sup_{u_1,\ldots,u_r\in\mathbb S^{p-1}}
    \prod_{a=1}^r
    \left(
    \E |u_a^\top X|^r
    \right)^{1/r}
    \le
    C_{\gamma,d,K}.
\end{split}
\]
Hence, on the event \eqref{eq:uniform_local_moment_errors_all_orders},
\begin{equation}
\label{eq:local_population_and_sample_moment_bounds}
    \max_{1\le r\le d}
    \max_{I\in\mathcal I_m}
    \|\Pi_I^{(r)}(\bcM_r)\|
    \le
    C_{\gamma,d,K},
    \qquad
    \max_{1\le r\le d}
    \max_{I\in\mathcal I_m}
    \|\Pi_I^{(r)}(\widehat{\bcM}_r)\|
    \le
    C_{\gamma,d,K}(1+\Delta_{m,x}).
\end{equation}

Let $\Pi_d$ be the set of all partitions of $[d]$. For a partition $\pi=\{B_1,\ldots,B_s\}\in\Pi_d$, define
\[
    c_\pi=(-1)^{s-1}(s-1)!.
\]
Let $\bcM_\pi$ denote the order-$d$ tensor obtained by placing $\bcM_{|B_j|}$ on the modes indexed by $B_j$ and taking the tensor product over $j=1,\ldots,s$. Define $\widehat{\bcM}_\pi$ in the same way with $\widehat{\bcM}_{|B_j|}$ in place of $\bcM_{|B_j|}$. Then
\[
    \bcK_d
    =
    \sum_{\pi\in\Pi_d} c_\pi \bcM_\pi,
    \qquad
    \widehat{\bcK}_d
    =
    \sum_{\pi\in\Pi_d} c_\pi \widehat{\bcM}_\pi.
\]

Now fix $I\in\mathcal I_m$ and a partition $\pi=\{B_1,\ldots,B_s\}\in\Pi_d$. After a permutation of tensor modes,
\[
    \Pi_I^{(d)}(\bcM_\pi)
    =
    \bigotimes_{j=1}^s
    \Pi_I^{(|B_j|)}(\bcM_{|B_j|}),
\]
and the same identity holds with $\widehat{\bcM}$ in place of $\bcM$. The spectral norm is invariant under permutations of tensor modes, and the spectral norm of a tensor product is the product of the spectral norms of its factors. Therefore the telescoping identity gives
\[
\begin{split}
    &
    \left\|
    \Pi_I^{(d)}
    \left(
    \widehat{\bcM}_\pi-\bcM_\pi
    \right)
    \right\|  \\
    &\le
    \sum_{q=1}^s
    \left(
    \prod_{j<q}
    \|\Pi_I^{(|B_j|)}(\widehat{\bcM}_{|B_j|})\|
    \right)
    \|\Pi_I^{(|B_q|)}(\widehat{\bcM}_{|B_q|}-\bcM_{|B_q|})\|
    \left(
    \prod_{j>q}
    \|\Pi_I^{(|B_j|)}(\bcM_{|B_j|})\|
    \right).
\end{split}
\]
Using \eqref{eq:uniform_local_moment_errors_all_orders} and \eqref{eq:local_population_and_sample_moment_bounds}, we obtain
\begin{equation}
\label{eq:single_partition_local_bound}
    \left\|
    \Pi_I^{(d)}
    \left(
    \widehat{\bcM}_\pi-\bcM_\pi
    \right)
    \right\|
    \le
    C_{\gamma,d,K}
    \Delta_{m,x}
    (1+\Delta_{m,x})^{s-1}.
\end{equation}
Since $s\le d$ and $|\Pi_d|$ depends only on $d$, summing \eqref{eq:single_partition_local_bound} over $\pi\in\Pi_d$ yields
\begin{equation}
\label{eq:local_empirical_cumulant_bound}
    \max_{I\in\mathcal I_m}
    \|\Pi_I^{(d)}(\widehat{\bcK}_d-\bcK_d)\|
    \le
    C_{\gamma,d,K}
    \Delta_{m,x}
    (1+\Delta_{m,x})^{d-1}
\end{equation}
with probability at least $1-Ce^{-x}$.

Apply \eqref{eq:local_empirical_cumulant_bound} with $m=k$ and $m=h_k$, and take a union bound over these two choices of $m$. Since $\Delta_{h_k,x}\le\Delta_{k,x}$, \eqref{eq:tapered_cumulant_local_reduction} gives
\begin{equation}
\label{eq:tapered_cumulant_stochastic_bound}
    \|\mathcal W_k*\mathcal E_{\bcK}\|
    \le
    C_{\gamma,d,K}
    \Delta_{k,x}
    (1+\Delta_{k,x})^{d-1}
\end{equation}
with probability at least $1-Ce^{-x}$. Combining \eqref{eq:tapered_cumulant_triangle}, \eqref{eq:tapered_cumulant_bias}, and \eqref{eq:tapered_cumulant_stochastic_bound} proves the high-probability bound.

We now prove the expectation bound. Let
\[
    Y
    =
    \left\|
    \widehat{\bcK}_{d,T,k}
    -
    \bcK_d
    \right\|,
    \qquad
    a=k+\log p,
    \qquad
    q=\frac d\gamma.
\]
The high-probability bound implies that, after enlarging constants,
\[
    \PP
    \left\{
    Y
    >
    C_{\alpha,d}\beta k^{-\alpha-d/2+1}
    +
    C_{\gamma,d,K}
    H(x)
    \right\}
    \le
    Ce^{-x},
    \qquad x\ge0,
\]
where
\[
    H(x)
    =
    \delta_x(1+\delta_x)^{d-1},
    \qquad
    \delta_x
    =
    \sqrt{\frac{a+x}{n}}
    +
    \frac{(a+x)^q}{n}.
\]
Since $H(x)$ is nondecreasing, the standard tail-integration argument gives
\[
    \E
    \left[
    Y
    -
    C_{\alpha,d}\beta k^{-\alpha-d/2+1}
    \right]_+
    \le
    C_{\gamma,d,K}
    \E H(\xi),
\]
where $\xi$ is a nonnegative random variable satisfying
\[
    \PP(\xi>x)\le Ce^{-x},
    \qquad x\ge0.
\]
Since
\[
    H(\xi)
    \le
    C_d
    \sum_{j=1}^d \delta_\xi^j,
\]
it remains to bound the moments of $\delta_\xi$. For each fixed $1\le j\le d$,
\[
    \E \delta_\xi^j
    \le
    C_{\gamma,d}
    \left\{
    \left(\frac{a}{n}\right)^{j/2}
    +
    \left(\frac{a^q}{n}\right)^j
    \right\}.
\]
Here we used $a\ge1$ and the fact that all moments of $\xi$ are bounded by constants depending only on $j$ and $q$. Therefore, with
\[
    \Delta_{k,0}
    =
    \sqrt{\frac{a}{n}}
    +
    \frac{a^q}{n},
\]
we have
\[
    \E H(\xi)
    \le
    C_{\gamma,d}
    \sum_{j=1}^d
    \Delta_{k,0}^j
    \le
    C_{\gamma,d}
    \Delta_{k,0}(1+\Delta_{k,0})^{d-1}.
\]
Thus
\[
    \E Y
    \le
    C_{\alpha,d}\beta k^{-\alpha-d/2+1}
    +
    C_{\gamma,d,K}
    \Delta_{k,0}(1+\Delta_{k,0})^{d-1}.
\]
This proves the expectation bound.

Finally, if $\Delta_{k,x}\le1$, then
\[
    \Delta_{k,x}(1+\Delta_{k,x})^{d-1}
    \le
    C_d\Delta_{k,x},
\]
which gives the simplified high-probability bound. The simplified expectation bound follows in the same way when $\Delta_{k,0}\le1$. This completes the proof.
\end{proof}

\subsection{Proof of Theorem \ref{thm_bandable_cumulant_lower_bound}}

\begin{proof}
We first state the perturbation claim used in the proof. Its proof is given after the main argument.

\medskip
\noindent
\textbf{Claim.}
There exist constants $\theta_0>0$, $c_d>0$, and $C_d>0$, depending only on $d$, such that, for every $p\ge1$, every $u\in\mathbb S^{p-1}$, and every $\theta\in(0,\theta_0]$, there exists a random vector $X_{u,\theta}\in\R^p$ with law $P_{u,\theta}$ satisfying
\begin{equation}
\label{eq_hermite_tilt_orlicz}
    \sup_{v\in\mathbb S^{p-1}}
    \|v^\top X_{u,\theta}\|_{\Phi_2}
    \le
    C_d,
\end{equation}
\begin{equation}
\label{eq_hermite_tilt_cumulant_rank_one}
    \bcK_d(X_{u,\theta})
    =
    \lambda_\theta u^{\otimes d},
    \qquad
    c_d\theta
    \le
    |\lambda_\theta|
    \le
    C_d\theta,
\end{equation}
and
\begin{equation}
\label{eq_hermite_tilt_KL_to_gaussian}
    D(P_{u,\theta}\|\gamma_p)
    \le
    C_d\theta^2,
\end{equation}
where $\gamma_p=N(0,I_p)$ and $\lambda_\theta$ does not depend on $u$.

\medskip
\noindent
\textbf{Proof of the theorem.}
Fix an integer $k_0\le k\le p$, where $k_0=k_0(d)$ is a sufficiently large constant. Let
\[
    L_k
    =
    \left\lfloor\frac pk\right\rfloor
\]
and define the disjoint coordinate blocks
\[
    I_\ell
    =
    \{(\ell-1)k+1,\ldots,\ell k\},
    \qquad
    \ell=1,\ldots,L_k.
\]
By Lemma~\ref{lemma_balanced_sign_packing}, there exist constants $c_{\mathrm{code}}>0$ and $\rho_0\in(0,1)$ and a set $\mathcal E_k\subset\{-1,1\}^k$ such that
\[
    |\mathcal E_k|
    \ge
    \exp(c_{\mathrm{code}}k)
\]
and, for every distinct $\varepsilon,\varepsilon'\in\mathcal E_k$,
\[
    \left|
    \frac1k
    \sum_{r=1}^k
    \varepsilon_r\varepsilon_r'
    \right|
    \le
    \rho_0.
\]
For every $\ell\in[L_k]$ and $\varepsilon\in\mathcal E_k$, define
\[
    u_{\ell,\varepsilon}
    =
    \frac1{\sqrt k}
    \sum_{r=1}^k
    \varepsilon_r
    e_{(\ell-1)k+r}
    \in
    \mathbb S^{p-1}.
\]
If $\ell=\ell'$ and $\varepsilon\ne\varepsilon'$, then
\[
    \left|
    \left\langle
    u_{\ell,\varepsilon},
    u_{\ell',\varepsilon'}
    \right\rangle
    \right|
    \le
    \rho_0.
\]
If $\ell\ne\ell'$, then the two vectors have disjoint supports and hence are orthogonal. Therefore, for every distinct pair $(\ell,\varepsilon)\ne(\ell',\varepsilon')$,
\begin{equation}
\label{eq_location_packing_inner_product}
    \left|
    \left\langle
    u_{\ell,\varepsilon},
    u_{\ell',\varepsilon'}
    \right\rangle
    \right|
    \le
    \rho_0.
\end{equation}

The total number of alternatives is
\[
    M_k
    =
    L_k|\mathcal E_k|.
\]
There exists a constant $c>0$ depending only on $d$ such that
\begin{equation}
\label{eq_location_packing_entropy}
    \log M_k
    \ge
    c
    \left\{
    k+\log(ep/k)
    \right\}.
\end{equation}
Indeed, if $p/k\ge2$, then $L_k\ge p/(2k)$ and
\[
    \log M_k
    \ge
    c_{\mathrm{code}}k
    +
    \log(p/k)
    -
    \log2.
\]
If $1\le p/k<2$, then $L_k=1$ and $\log(ep/k)$ is bounded by an absolute constant, which is absorbed by the term $c_{\mathrm{code}}k$. Increasing $k_0$ if necessary proves \eqref{eq_location_packing_entropy}.

Choose
\begin{equation}
\label{eq_location_lb_theta_choice}
    \theta
    =
    c_\theta
    \min
    \left\{
    \beta k^{-\alpha-d/2+1},
    \sqrt{\frac{k+\log(ep/k)}{n}},
    1
    \right\},
\end{equation}
where $c_\theta>0$ is a sufficiently small constant depending only on $d$. Taking $c_\theta\le\theta_0$ ensures that $\theta\le\theta_0$.

For each $(\ell,\varepsilon)$, let $X_{\ell,\varepsilon}=X_{u_{\ell,\varepsilon},\theta}$ be the random vector from the perturbation claim, and define
\[
    Y_{\ell,\varepsilon}
    =
    a_d
    \left(
    X_{\ell,\varepsilon}
    -
    \E X_{\ell,\varepsilon}
    \right),
\]
where $a_d\in(0,1]$ is a sufficiently small constant depending only on $d$. Let $Q_{\ell,\varepsilon}$ denote the law of $Y_{\ell,\varepsilon}$. Then $\E Y_{\ell,\varepsilon}=0$, and invariance of higher-order cumulants under deterministic shifts gives
\begin{equation}
\label{eq_location_lb_cumulant}
    \bcK_d(Y_{\ell,\varepsilon})
    =
    a_d^d\lambda_\theta
    u_{\ell,\varepsilon}^{\otimes d}.
\end{equation}
By \eqref{eq_hermite_tilt_orlicz} and the centering inequality for Orlicz norms,
\[
    \left\|
    v^\top
    \left(
    X_{\ell,\varepsilon}
    -
    \E X_{\ell,\varepsilon}
    \right)
    \right\|_{\Phi_2}
    \le
    C_d,
    \qquad
    v\in\mathbb S^{p-1}.
\]
Taking $a_d>0$ sufficiently small gives
\[
    \|v^\top Y_{\ell,\varepsilon}\|_{\Phi_2}
    \le
    2,
    \qquad
    v\in\mathbb S^{p-1}.
\]

We next verify the bandability constraint. If
\[
    \left(
    u_{\ell,\varepsilon}^{\otimes d}
    \right)_{\bi}
    \ne
    0,
\]
then all coordinates of $\bi$ belong to $I_\ell$, so
\[
    \operatorname{diam}(\bi)
    \le
    k-1.
\]
Moreover,
\[
    \left|
    \left(
    u_{\ell,\varepsilon}^{\otimes d}
    \right)_{\bi}
    \right|
    =
    k^{-d/2}.
\]
By \eqref{eq_hermite_tilt_cumulant_rank_one}, \eqref{eq_location_lb_cumulant}, and \eqref{eq_location_lb_theta_choice},
\[
\begin{split}
    \left|
    \left(
    \bcK_d(Y_{\ell,\varepsilon})
    \right)_{\bi}
    \right|
    &\le
    C_d\theta k^{-d/2} \\
    &\le
    C_dc_\theta
    \beta
    k^{-(\alpha+d-1)}.
\end{split}
\]
After decreasing $c_\theta$ if necessary,
\[
    \left|
    \left(
    \bcK_d(Y_{\ell,\varepsilon})
    \right)_{\bi}
    \right|
    \le
    \beta
    k^{-(\alpha+d-1)}
    \le
    \beta
    \left(
    1+\operatorname{diam}(\bi)
    \right)^{-(\alpha+d-1)}.
\]
Outside the support, the entries are zero. Therefore,
\begin{equation}
\label{eq_location_lb_class_membership}
    Y_{\ell,\varepsilon}
    \in
    \mathcal P_{\alpha,\beta}^d.
\end{equation}

For two distinct alternatives, \eqref{eq_location_packing_inner_product} gives
\[
\begin{split}
    \left\|
    u_{\ell,\varepsilon}^{\otimes d}
    -
    u_{\ell',\varepsilon'}^{\otimes d}
    \right\|
    &\ge
    \left|
    \left\langle
    u_{\ell,\varepsilon}^{\otimes d}
    -
    u_{\ell',\varepsilon'}^{\otimes d},
    u_{\ell,\varepsilon}^{\otimes d}
    \right\rangle
    \right| \\
    &=
    \left|
    1-
    \left\langle
    u_{\ell,\varepsilon},
    u_{\ell',\varepsilon'}
    \right\rangle^d
    \right| \\
    &\ge
    1-\rho_0^d.
\end{split}
\]
Hence, by \eqref{eq_hermite_tilt_cumulant_rank_one} and \eqref{eq_location_lb_cumulant},
\begin{equation}
\label{eq_location_lb_pairwise_separation}
    \left\|
    \bcK_d(Y_{\ell,\varepsilon})
    -
    \bcK_d(Y_{\ell',\varepsilon'})
    \right\|
    \ge
    c_d\theta.
\end{equation}

Let
\[
    \mu_{\ell,\varepsilon}
    =
    \E X_{\ell,\varepsilon}
\]
and let $\widetilde P_{\ell,\varepsilon}$ be the law of $X_{\ell,\varepsilon}-\mu_{\ell,\varepsilon}$. Translation invariance of KL divergence gives
\[
    D(\widetilde P_{\ell,\varepsilon}\|\gamma_p)
    =
    D(P_{u_{\ell,\varepsilon},\theta}\|N(\mu_{\ell,\varepsilon},I_p)).
\]
For any probability distribution $P$ with mean $\mu$,
\[
    \log
    \frac{dN(\mu,I_p)}{d\gamma_p}(x)
    =
    \mu^\top x
    -
    \frac12\|\mu\|_2^2,
\]
and therefore
\[
\begin{split}
    D(P\|N(\mu,I_p))
    &=
    D(P\|\gamma_p)
    -
    \E_P
    \log
    \frac{dN(\mu,I_p)}{d\gamma_p}(X) \\
    &=
    D(P\|\gamma_p)
    -
    \frac12\|\mu\|_2^2 \\
    &\le
    D(P\|\gamma_p).
\end{split}
\]
It follows from \eqref{eq_hermite_tilt_KL_to_gaussian} that
\[
    D(\widetilde P_{\ell,\varepsilon}\|\gamma_p)
    \le
    C_d\theta^2.
\]
Let
\[
    Q_0
    =
    N(0,a_d^2I_p).
\]
Invariance of KL divergence under invertible affine transformations gives
\[
    D(Q_{\ell,\varepsilon}\|Q_0)
    =
    D(\widetilde P_{\ell,\varepsilon}\|\gamma_p)
    \le
    C_d\theta^2.
\]
Thus, for $n$ independent observations,
\[
    D(Q_{\ell,\varepsilon}^{\otimes n}\|Q_0^{\otimes n})
    \le
    C_dn\theta^2.
\]
By \eqref{eq_location_lb_theta_choice},
\[
    n\theta^2
    \le
    c_\theta^2
    \left\{
    k+\log(ep/k)
    \right\}.
\]
Combining this inequality with \eqref{eq_location_packing_entropy} and taking $c_\theta>0$ sufficiently small gives
\[
    \frac1{M_k}
    \sum_{\ell=1}^{L_k}
    \sum_{\varepsilon\in\mathcal E_k}
    D(Q_{\ell,\varepsilon}^{\otimes n}\|Q_0^{\otimes n})
    \le
    \frac18\log M_k.
\]
Increasing $k_0$ if necessary also ensures that
\[
    \log2
    \le
    \frac14\log M_k.
\]
Applying the generalized Fano lemma with the reference distribution $Q_0^{\otimes n}$, the class membership \eqref{eq_location_lb_class_membership}, and the separation \eqref{eq_location_lb_pairwise_separation}, we obtain
\[
\begin{split}
    &
    \inf_{\widehat{\bcK}_d}
    \sup_{X\in\mathcal P_{\alpha,\beta}^d}
    \E
    \left\|
    \widehat{\bcK}_d-\bcK_d(X)
    \right\| \\
    &\qquad\ge
    c_d
    \min
    \left\{
    \beta k^{-\alpha-d/2+1},
    \sqrt{\frac{k+\log(ep/k)}{n}},
    1
    \right\}.
\end{split}
\]
For every $1\le k\le p$,
\[
    \frac12
    \left(
    k+\log p
    \right)
    \le
    k+\log(ep/k)
    \le
    2
    \left(
    k+\log p
    \right).
\]
Consequently, for every $k_0\le k\le p$,
\begin{equation}
\label{eq_location_lb_large_k}
\begin{split}
    &
    \inf_{\widehat{\bcK}_d}
    \sup_{X\in\mathcal P_{\alpha,\beta}^d}
    \E
    \left\|
    \widehat{\bcK}_d-\bcK_d(X)
    \right\| \\
    &\qquad\ge
    c_d
    \min
    \left\{
    \beta k^{-\alpha-d/2+1},
    \sqrt{\frac{k+\log p}{n}},
    1
    \right\}.
\end{split}
\end{equation}

It remains to control the finitely many values $1\le k<k_0$. For $j=1,\ldots,p$, take $u_j=e_j$ and choose
\[
    \theta_{\mathrm{diag}}
    =
    c_\theta
    \min
    \left\{
    \beta,
    \sqrt{\frac{\log p}{n}},
    1
    \right\}.
\]
Apply the perturbation claim with $u=u_j$ and $\theta=\theta_{\mathrm{diag}}$, followed by the same centering and rescaling construction. Denote the resulting centered random vector by $Y_j$ and its law by $Q_j$. Then
\[
    \bcK_d(Y_j)
    =
    a_d^d\lambda_{\theta_{\mathrm{diag}}}
    e_j^{\otimes d}.
\]
The only nonzero tensor entry has diameter zero, and its magnitude is bounded by $\beta$ after decreasing $c_\theta$ if necessary. Therefore,
\[
    Y_j
    \in
    \mathcal P_{\alpha,\beta}^d.
\]
For $j\ne j'$,
\[
\begin{split}
    \left\|
    \bcK_d(Y_j)
    -
    \bcK_d(Y_{j'})
    \right\|
    &=
    a_d^d
    |\lambda_{\theta_{\mathrm{diag}}}|
    \left\|
    e_j^{\otimes d}
    -
    e_{j'}^{\otimes d}
    \right\| \\
    &\ge
    c_d\theta_{\mathrm{diag}},
\end{split}
\]
where
\[
    \left\|
    e_j^{\otimes d}
    -
    e_{j'}^{\otimes d}
    \right\|
    \ge
    \left|
    \left\langle
    e_j^{\otimes d}
    -
    e_{j'}^{\otimes d},
    e_j^{\otimes d}
    \right\rangle
    \right|
    =
    1.
\]
The same KL calculation gives
\[
    D(Q_j^{\otimes n}\|Q_0^{\otimes n})
    \le
    C_dn\theta_{\mathrm{diag}}^2
    \le
    C_dc_\theta^2\log p.
\]
After decreasing $c_\theta$ if necessary, generalized Fano's inequality over the $p$ alternatives gives
\begin{equation}
\label{eq_location_lb_diagonal}
\begin{split}
    &
    \inf_{\widehat{\bcK}_d}
    \sup_{X\in\mathcal P_{\alpha,\beta}^d}
    \E
    \left\|
    \widehat{\bcK}_d-\bcK_d(X)
    \right\| \\
    &\qquad\ge
    c_d
    \min
    \left\{
    \beta,
    \sqrt{\frac{\log p}{n}},
    1
    \right\}.
\end{split}
\end{equation}

For every $1\le k<k_0$,
\[
    \beta k^{-\alpha-d/2+1}
    \le
    \beta
\]
and, since $p\ge k_0$,
\[
    k+\log p
    \le
    C_d\log p.
\]
It follows from \eqref{eq_location_lb_diagonal} that, for every $1\le k<k_0$,
\[
\begin{split}
    &
    \inf_{\widehat{\bcK}_d}
    \sup_{X\in\mathcal P_{\alpha,\beta}^d}
    \E
    \left\|
    \widehat{\bcK}_d-\bcK_d(X)
    \right\| \\
    &\qquad\ge
    c_d
    \min
    \left\{
    \beta k^{-\alpha-d/2+1},
    \sqrt{\frac{k+\log p}{n}},
    1
    \right\}.
\end{split}
\]
Combining this inequality with \eqref{eq_location_lb_large_k} and taking the supremum over $1\le k\le p$ proves the theorem.

\medskip
\noindent
\textbf{Proof of the perturbation claim.}

Let $H_d$ be the probabilists' Hermite polynomial of degree $d$. For $M>0$, define
\[
    h_M(t)
    =
    H_d(t)\mathbf I\{|H_d(t)|\le M\}
    +
    M\operatorname{sign}(H_d(t))
    \mathbf I\{|H_d(t)|>M\}.
\]
We choose $M=M_d$ sufficiently large below. Let $G\sim\gamma_p=N(0,I_p)$. For every $u\in\mathbb S^{p-1}$ and $\theta\in\R$, define the tilted distribution $P_{u,\theta}$ by
\[
    \frac{dP_{u,\theta}}{d\gamma_p}(x)
    =
    \frac{
    \exp\{\theta h_M(u^\top x)\}
    }{
    \E_{\gamma_p}
    \exp\{\theta h_M(u^\top G)\}
    }.
\]
Let $X_{u,\theta}\sim P_{u,\theta}$.

We first establish the uniform sub-Gaussian bound. Since $|h_M|\le M$, for every $x\in\R^p$,
\[
    \frac{dP_{u,\theta}}{d\gamma_p}(x)
    \le
    e^{2|\theta|M}.
\]
Therefore, for every $v\in\mathbb S^{p-1}$ and every $L>\sqrt2$,
\[
\begin{split}
    \E_{P_{u,\theta}}
    \exp
    \left\{
    \frac{(v^\top X_{u,\theta})^2}{L^2}
    \right\}
    &\le
    e^{2|\theta|M}
    \E_{\gamma_p}
    \exp
    \left\{
    \frac{(v^\top G)^2}{L^2}
    \right\} \\
    &=
    e^{2|\theta|M}
    \left(
    1-\frac2{L^2}
    \right)^{-1/2}.
\end{split}
\]
Choose $L=L_d$ sufficiently large and then choose $\theta_0>0$ sufficiently small so that
\[
    e^{2\theta_0M}
    \left(
    1-\frac2{L_d^2}
    \right)^{-1/2}
    \le
    2.
\]
It follows that, uniformly over $u\in\mathbb S^{p-1}$, $\theta\in(0,\theta_0]$, and $v\in\mathbb S^{p-1}$,
\[
    \E_{P_{u,\theta}}
    \exp
    \left\{
    \frac{(v^\top X_{u,\theta})^2}{L_d^2}
    \right\}
    \le
    2.
\]
This proves \eqref{eq_hermite_tilt_orlicz}.

We next compute the order-$d$ cumulant tensor. Under $P_{u,\theta}$, the orthogonal decomposition of $X_{u,\theta}$ satisfies
\[
    X_{u,\theta}
    =
    Su+Z_\perp,
\]
where $S=u^\top X_{u,\theta}$, $Z_\perp=(I_p-uu^\top)X_{u,\theta}$, $S$ and $Z_\perp$ are independent, and
\[
    Z_\perp
    \sim
    N(0,I_p-uu^\top).
\]
The density of $S$ is
\[
    f_{S,\theta}(s)
    =
    \frac{
    e^{\theta h_M(s)}\phi(s)
    }{
    \E_{U\sim N(0,1)}
    e^{\theta h_M(U)}
    }.
\]
For $t\in\R^p$, write $t=t_uu+t_\perp$, where $t_u=\langle t,u\rangle$ and $t_\perp=(I_p-uu^\top)t$. The cumulant generating function of $X_{u,\theta}$ is
\[
    \kappa_X(t)
    =
    \kappa_S(t_u)
    +
    \frac12\|t_\perp\|_2^2,
\]
where
\[
    \kappa_S(t)
    =
    \log
    \E_{U\sim N(0,1)}
    \exp\{\theta h_M(U)+tU\}
    -
    \log
    \E_{U\sim N(0,1)}
    \exp\{\theta h_M(U)\}.
\]
Since $d\ge3$, the Gaussian quadratic term contributes no order-$d$ cumulant. Therefore,
\[
    \bcK_d(X_{u,\theta})
    =
    \kappa_S^{(d)}(0)
    u^{\otimes d}.
\]

We now expand $\kappa_S^{(d)}(0)$ in $\theta$. At $\theta=0$,
\[
    \kappa_S(t)
    =
    \frac{t^2}{2},
\]
so $\kappa_S^{(d)}(0)=0$. Moreover,
\[
    \left.
    \frac{\partial}{\partial\theta}
    \kappa_S(t)
    \right|_{\theta=0}
    =
    \E
    \left[
    h_M(U)e^{tU-t^2/2}
    \right]
    -
    \E h_M(U).
\]
Using the generating function of the probabilists' Hermite polynomials,
\[
    \left.
    \frac{d^d}{dt^d}
    e^{tU-t^2/2}
    \right|_{t=0}
    =
    H_d(U).
\]
Hence,
\[
    \left.
    \frac{\partial^{d+1}}
    {\partial t^d\partial\theta}
    \kappa_S(t)
    \right|_{t=0,\theta=0}
    =
    \E
    \left[
    h_M(U)H_d(U)
    \right]
    =
    C_{M,d},
\]
where
\[
    C_{M,d}
    =
    \E
    \left[
    H_d(U)^2\mathbf I\{|H_d(U)|\le M\}
    +
    M|H_d(U)|\mathbf I\{|H_d(U)|>M\}
    \right].
\]
By dominated convergence,
\[
    C_{M,d}
    \longrightarrow
    \E H_d(U)^2
    >
    0
\]
as $M\to\infty$. Choose $M=M_d$ sufficiently large that $C_{M,d}>0$. Since $h_M$ is bounded, the Taylor remainder is uniform over $u$ and $p$, and
\[
    \kappa_S^{(d)}(0)
    =
    C_{M,d}\theta
    +
    O_d(\theta^2).
\]
After decreasing $\theta_0$ if necessary,
\[
    \bcK_d(X_{u,\theta})
    =
    \lambda_\theta u^{\otimes d},
    \qquad
    c_d\theta
    \le
    |\lambda_\theta|
    \le
    C_d\theta.
\]
The tilted one-dimensional distribution of $u^\top X_{u,\theta}$ depends only on $\theta$ and not on $u$. Therefore, $\lambda_\theta$ is the same for every $u\in\mathbb S^{p-1}$.

It remains to prove the KL bound. Let
\[
    A(\theta)
    =
    \log
    \E_{U\sim N(0,1)}
    \exp\{\theta h_M(U)\}.
\]
Then
\[
    D(P_{u,\theta}\|\gamma_p)
    =
    \theta
    \E_{P_{u,\theta}}
    h_M(u^\top X_{u,\theta})
    -
    A(\theta).
\]
Taylor expansion gives
\[
    A(\theta)
    =
    \theta\E h_M(U)
    +
    \frac{\theta^2}{2}
    \operatorname{Var}(h_M(U))
    +
    O_d(\theta^3)
\]
and
\[
    \E_{P_{u,\theta}}
    h_M(u^\top X_{u,\theta})
    =
    \E h_M(U)
    +
    \theta
    \operatorname{Var}(h_M(U))
    +
    O_d(\theta^2).
\]
Therefore,
\[
    D(P_{u,\theta}\|\gamma_p)
    =
    \frac{\theta^2}{2}
    \operatorname{Var}(h_M(U))
    +
    O_d(\theta^3)
    \le
    C_d\theta^2
\]
for every $\theta\in(0,\theta_0]$, after decreasing $\theta_0$ if necessary. This proves the perturbation claim and completes the proof.
\end{proof}

\subsection{Proof of Theorem \ref{thm:cumulant-yule-walker}}

\begin{proof}
Let $E_A=\widehat A-A$ and $e_b=\widehat b-b$.
Tensor spectral-norm control gives
\[
    \|(E_A)_{a,\cdot}\|_2
    \leq
    \delta_k,
    \qquad
    |(e_b)_a|
    \leq
    \delta_k,
    \qquad
    a=1,\ldots,m.
\]
Therefore,
\[
    \|E_A\|
    \leq
    \sqrt m\,\delta_k,
    \qquad
    \|e_b\|_2
    \leq
    \sqrt m\,\delta_k.
\]
Under \eqref{eq:ar-yw-spectral-small-error}, the singular-value perturbation inequality gives
\[
    \sigma_{\min}(\widehat A)
    \geq
    s_{A,r}-\|E_A\|
    \geq
    \frac{s_{A,r}}{2},
\]
and hence $\|\widehat A^\dagger\|\leq2/s_{A,r}$.
Using \eqref{eq:ar-truncated-linear-system},
\[
    \widehat\phi_{d,k}^{(r)}
    -
    \phi^{(r)}
    =
    \widehat A^\dagger
    \left(
    e_b-E_A\phi^{(r)}+q_r
    \right),
\]
which proves \eqref{eq:ar-yw-spectral-coefficient-bound}.
The high-probability conclusion follows from Lemma~\ref{lem:ar-cumulant-bandability} and Theorem~\ref{thm:tapered-cumulant-upper-bound}.
\end{proof}

\subsection{Proofs of Theorem \ref{thm:cumulant-ma} and Remark \ref{rem:ma-local-global}}
\label{sec:proof-cumulant-ma}

\begin{proof}
Because $\mathcal N_{\rho_0}$ is compact and the objective function
$\eta\mapsto\|\widehat b-M(\eta)\|_2^2$ is continuous, the set of minimizers
over $\mathcal N_{\rho_0}$ is nonempty. Fix any such minimizer
$\widehat\eta$ and define
\[
    \bcE_k
    =
    \widehat{\bcK}_{d,T,k}
    -
    \bcK_d.
\]

For every $a=1,\ldots,m$, the difference $\widehat b_a-b_a$ is an entry of
$\bcE_k$. Since the tensor spectral norm dominates the absolute value of each
entry,
\[
    |\widehat b_a-b_a|
    \le
    \|\bcE_k\|.
\]
Therefore,
\begin{equation}
\label{eq:ma-cumulant-md-b-error}
    \|\widehat b-b\|_2
    \le
    \sqrt m\,\|\bcE_k\|.
\end{equation}

Since $\widehat\eta$ minimizes
$\|\widehat b-M(\eta)\|_2^2$ over $\mathcal N_{\rho_0}$ and
$\eta_\star\in\mathcal N_{\rho_0}$,
\[
    \|\widehat b-M(\widehat\eta)\|_2
    \le
    \|\widehat b-M(\eta_\star)\|_2
    =
    \|\widehat b-b\|_2.
\]
Hence, by the triangle inequality and
\eqref{eq:ma-cumulant-md-b-error},
\begin{equation}
\label{eq:ma-cumulant-md-map-upper}
    \|M(\widehat\eta)-M(\eta_\star)\|_2
    \le
    2\|\widehat b-b\|_2
    \le
    2\sqrt m\,\|\bcE_k\|.
\end{equation}

For any $\eta\in\mathcal N_{\rho_0}$,
\[
    M(\eta)-M(\eta_\star)
    =
    J_\star(\eta-\eta_\star)
    +
    R(\eta),
\]
where
\[
    R(\eta)
    =
    \int_0^1
    \left\{
    \nabla M\bigl(\eta_\star+t(\eta-\eta_\star)\bigr)
    -
    J_\star
    \right\}
    (\eta-\eta_\star)\,dt.
\]
The Lipschitz condition on $\nabla M$ gives
\[
\begin{split}
    \|R(\eta)\|_2
    &\le
    \int_0^1
    Lt\|\eta-\eta_\star\|_2^2\,dt
    \\
    &=
    \frac{L}{2}
    \|\eta-\eta_\star\|_2^2.
\end{split}
\]
Consequently,
\[
\begin{split}
    \|M(\eta)-M(\eta_\star)\|_2
    &\ge
    \|J_\star(\eta-\eta_\star)\|_2
    -
    \|R(\eta)\|_2
    \\
    &\ge
    s\|\eta-\eta_\star\|_2
    -
    \frac{L}{2}
    \|\eta-\eta_\star\|_2^2.
\end{split}
\]
Since $\eta\in\mathcal N_{\rho_0}$ and
$\rho_0\le s/(2L)$,
\[
\begin{split}
    \|M(\eta)-M(\eta_\star)\|_2
    &\ge
    \left(
    s-\frac{L\rho_0}{2}
    \right)
    \|\eta-\eta_\star\|_2
    \\
    &\ge
    \frac{3s}{4}
    \|\eta-\eta_\star\|_2.
\end{split}
\]
Applying this inequality with $\eta=\widehat\eta$ and combining it with
\eqref{eq:ma-cumulant-md-map-upper} yields
\[
    \frac{3s}{4}
    \|\widehat\eta-\eta_\star\|_2
    \le
    2\sqrt m\,\|\bcE_k\|.
\]
Thus
\[
\begin{split}
    \|\widehat\eta-\eta_\star\|_2
    &\le
    \frac{8\sqrt m}{3s}
    \|\bcE_k\|
    \\
    &\le
    \frac{4\sqrt m}{s}
    \|\widehat{\bcK}_{d,T,k}-\bcK_d\|.
\end{split}
\]
This proves the deterministic claim.

For the independent-replicate result,
Lemma~\ref{lem:ma-cumulant-bandability} gives
\[
    \bcK_d
    \in
    \mathcal K_{\alpha,\beta_{\mathrm{MA}}}^d.
\]
Under the assumptions of
Theorem~\ref{thm:tapered-cumulant-upper-bound}, with probability at least
$1-Ce^{-x}$,
\[
    \|\widehat{\bcK}_{d,T,k}-\bcK_d\|
    \le
    C_{\alpha,d}
    \beta_{\mathrm{MA}}
    k^{-\alpha-d/2+1}
    +
    C_{\gamma,d,K}
    \Delta_{k,x}(1+\Delta_{k,x})^{d-1}.
\]
Substituting this inequality into the deterministic perturbation bound gives
\[
    \|\widehat\eta-\eta_\star\|_2
    \le
    \frac{4\sqrt m}{s}
    \left[
    C_{\alpha,d}
    \beta_{\mathrm{MA}}
    k^{-\alpha-d/2+1}
    +
    C_{\gamma,d,K}
    \Delta_{k,x}(1+\Delta_{k,x})^{d-1}
    \right]
\]
on the same event. This completes the proof.
\end{proof}

\begin{proof}[Proof of Remark~\ref{rem:ma-local-global}]
Define
\[
    \bcE_k
    =
    \widehat{\bcK}_{d,T,k}
    -
    \bcK_d.
\]
For every selected cumulant entry,
\[
    |\widehat b_a-b_a|
    \le
    \|\bcE_k\|,
\]
and hence
\[
    \|\widehat b-b\|_2
    \le
    \sqrt m\,\|\bcE_k\|.
\]

Since $\widetilde\eta$ is a global minimizer over $\Theta$ and
$\eta_\star\in\Theta$,
\[
    \|\widehat b-M(\widetilde\eta)\|_2
    \le
    \|\widehat b-M(\eta_\star)\|_2
    =
    \|\widehat b-b\|_2.
\]
It follows that
\[
\begin{split}
    \|M(\widetilde\eta)-M(\eta_\star)\|_2
    &\le
    \|\widehat b-M(\widetilde\eta)\|_2
    +
    \|\widehat b-M(\eta_\star)\|_2
    \\
    &\le
    2\|\widehat b-b\|_2
    \\
    &\le
    2\sqrt m\,\|\bcE_k\|.
\end{split}
\]
Suppose that
\[
    2\sqrt m\,\|\bcE_k\|
    <
    c_0.
\]
If $\|\widetilde\eta-\eta_\star\|_2>\rho_0$, then the separation condition
would imply
\[
    \|M(\widetilde\eta)-M(\eta_\star)\|_2
    \ge
    c_0,
\]
contradicting the preceding upper bound. Therefore,
$\widetilde\eta\in\mathcal N_{\rho_0}$.

Since $\widetilde\eta\in\mathcal N_{\rho_0}$ and
\[
    \|\widehat b-M(\widetilde\eta)\|_2
    \le
    \|\widehat b-M(\eta_\star)\|_2,
\]
the deterministic argument in the proof of
Theorem~\ref{thm:cumulant-ma} applies with $\widehat\eta$ replaced by
$\widetilde\eta$. Consequently,
\[
    \|\widetilde\eta-\eta_\star\|_2
    \le
    \frac{4\sqrt m}{s}
    \|\widehat{\bcK}_{d,T,k}-\bcK_d\|.
\]

Finally, under the independent-replicate assumptions of
Theorem~\ref{thm:cumulant-ma}, with probability at least $1-Ce^{-x}$,
\[
    \|\widehat{\bcK}_{d,T,k}-\bcK_d\|
    \le
    C_{\alpha,d}
    \beta_{\mathrm{MA}}
    k^{-\alpha-d/2+1}
    +
    C_{\gamma,d,K}
    \Delta_{k,x}(1+\Delta_{k,x})^{d-1}.
\]
The displayed condition in Remark~\ref{rem:ma-local-global} ensures that the
deterministic separation condition holds on this event. Substitution into the
deterministic bound completes the proof.
\end{proof}

\subsection{Proof of Theorem~\ref{thm:source-localization}}
\label{subsec:proof-source-localization}

\begin{proof}
We first prove the deterministic localization claim. Consider any realization
for which
\[
    \|\widehat{\bcK}_{d,T,k}-\bcK_d\|
    <
    \frac{\operatorname{Gap}(r)}{2}.
\]
For every $c\in[p]$, the normalization $\|g_c\|_2=1$ and the definition of
the tensor spectral norm give
\[
\begin{split}
    |\widehat Q_k(c)-Q(c)|
    &=
    \left|
    \left\langle
    \widehat{\bcK}_{d,T,k}-\bcK_d,
    g_c^{\otimes d}
    \right\rangle
    \right|
    \\
    &\le
    \|\widehat{\bcK}_{d,T,k}-\bcK_d\|.
\end{split}
\]
Consequently, for every $c$ satisfying $|c-c_0|>r$,
\[
\begin{split}
    \widehat Q_k(c)
    &\le
    Q(c)
    +
    \|\widehat{\bcK}_{d,T,k}-\bcK_d\|
    \\
    &\le
    Q(c_0)
    -
    \operatorname{Gap}(r)
    +
    \|\widehat{\bcK}_{d,T,k}-\bcK_d\|.
\end{split}
\]
At the population source location,
\[
    \widehat Q_k(c_0)
    \ge
    Q(c_0)
    -
    \|\widehat{\bcK}_{d,T,k}-\bcK_d\|.
\]
Since
\[
    \|\widehat{\bcK}_{d,T,k}-\bcK_d\|
    <
    \frac{\operatorname{Gap}(r)}{2},
\]
we have
\[
\begin{split}
    Q(c_0)
    -
    \operatorname{Gap}(r)
    +
    \|\widehat{\bcK}_{d,T,k}-\bcK_d\|
    &<
    Q(c_0)
    -
    \|\widehat{\bcK}_{d,T,k}-\bcK_d\|
    \\
    &\le
    \widehat Q_k(c_0).
\end{split}
\]
Therefore, every $c$ with $|c-c_0|>r$ has strictly smaller empirical score
than $c_0$. No maximizer in
\eqref{eq:spatial-location-estimator} can lie outside the radius-$r$
neighborhood of $c_0$, and every maximizer $\widehat c$ satisfies
\[
    |\widehat c-c_0|
    \le
    r.
\]
This proves the deterministic claim.

For the probabilistic result,
Lemma~\ref{lem:source-cumulant-bandability} gives
\[
    \bcK_d
    \in
    \mathcal K_{\alpha,\beta_{\mathrm{src}}}^d.
\]
Under the assumptions of
Theorem~\ref{thm:tapered-cumulant-upper-bound}, with probability at least
$1-Ce^{-x}$,
\[
    \|\widehat{\bcK}_{d,T,k}-\bcK_d\|
    \le
    C_{\alpha,d}
    \beta_{\mathrm{src}}
    k^{-\alpha-d/2+1}
    +
    C_{\gamma,d,K}
    \Delta_{k,x}
    (1+\Delta_{k,x})^{d-1}.
\]
If
\[
    C_{\alpha,d}
    \beta_{\mathrm{src}}
    k^{-\alpha-d/2+1}
    +
    C_{\gamma,d,K}
    \Delta_{k,x}
    (1+\Delta_{k,x})^{d-1}
    <
    \frac{\operatorname{Gap}(r)}{2},
\]
then the deterministic condition holds on this event. Applying the first part
of the theorem shows that, with probability at least $1-Ce^{-x}$, every
maximizer $\widehat c$ satisfies
\[
    |\widehat c-c_0|
    \le
    r.
\]
This completes the proof.
\end{proof}

\subsection{Proofs of Lemmas in the Main Text}

\subsubsection{Proof of Lemma \ref{lem:ar-cumulant-bandability}}

\begin{proof}
Let
\[
    H=\max_{0\leq i\leq d-1}h_i,
    \qquad
    L=\min_{0\leq i\leq d-1}h_i.
\]
The cumulant formula for a linear process gives
\[
    \kappa_d(h_0,\ldots,h_{d-1})
    =
    \tau_d
    \sum_{s=H}^{\infty}
    \prod_{i=0}^{d-1}
    \psi_{s-h_i}.
\]
Writing $s=H+u$ and using the geometric bound on $\psi_\ell$ yields
\[
\begin{aligned}
    \left|
    \kappa_d(h_0,\ldots,h_{d-1})
    \right|
    &\leq
    |\tau_d|C_\psi^d
    \sum_{u=0}^{\infty}
    \rho^{
    \sum_{i=0}^{d-1}(H+u-h_i)
    }\\
    &=
    |\tau_d|C_\psi^d
    \sum_{u=0}^{\infty}
    \rho^{
    du+\sum_{i=0}^{d-1}(H-h_i)
    }.
\end{aligned}
\]
Because
\[
    \sum_{i=0}^{d-1}(H-h_i)
    \geq
    H-L,
\]
we obtain
\[
    \left|
    \kappa_d(h_0,\ldots,h_{d-1})
    \right|
    \leq
    |\tau_d|C_\psi^d
    \rho^{H-L}
    \sum_{u=0}^{\infty}\rho^{du}
    =
    \frac{
    |\tau_d|C_\psi^d
    }{
    1-\rho^d
    }
    \rho^{H-L}.
\]
The asserted polynomial bandability bound follows from
\[
    \rho^D
    \leq
    \left\{
    \sup_{x\geq 0}
    (1+x)^{\alpha+d-1}\rho^x
    \right\}
    (1+D)^{-\alpha-d+1}.
\]
\end{proof}

\subsubsection{Proof of Lemma \ref{lem:ma-cumulant-bandability}}

\begin{proof}
By \eqref{eq:ma-cumulant-identity}, a summand is nonzero only if
\[
    a-h_\ell
    \in
    \{0,\ldots,q\},
    \qquad
    \ell=0,\ldots,d-1.
\]
If $\operatorname{diam}(h_0,\ldots,h_{d-1})>q$, no integer $a$ can satisfy these constraints for all $\ell=0,\ldots,d-1$. Hence $\kappa_d(h_0,\ldots,h_{d-1})=0$.

If $\operatorname{diam}(h_0,\ldots,h_{d-1})\le q$, then the number of nonzero terms in \eqref{eq:ma-cumulant-identity} is at most $q+1$. Therefore,
\[
    |\kappa_d(h_0,\ldots,h_{d-1})|
    \le
    |\tau_d|
    (q+1)
    (1\vee\|\theta\|_\infty)^d.
\]
Combining this bound with exact bandedness gives, for all $h_0,\ldots,h_{d-1}$,
\[
    |\kappa_d(h_0,\ldots,h_{d-1})|
    \le
    |\tau_d|
    (q+1)
    (1\vee\|\theta\|_\infty)^d
    (1+q)^{\alpha+d-1}
    \left(
    1+\operatorname{diam}(h_0,\ldots,h_{d-1})
    \right)^{-(\alpha+d-1)}.
\]
This proves the claim.
\end{proof}

\subsubsection{Proof of Lemma~\ref{lem:source-cumulant-bandability}}
\label{subsec:proof-source-cumulant-bandability}

\begin{proof}
By \eqref{eq:spatial-cumulant-representation}, for any multi-index $i=(i_1,\ldots,i_d)\in[p]^d$,
\[
    (\mathcal K_d)_i
    =
    \sum_{\ell=1}^L
    \tau_{d,\ell}
    \prod_{a=1}^d
    a_{\ell i_a}.
\]
Let $D=\operatorname{diam}(i)$. For each $\ell$, the triangle inequality gives
\[
    D
    =
    \max_{a,b\in[d]} |i_a-i_b|
    \le
    \sum_{a=1}^d |i_a-c_\ell|.
\]
Using the decay condition \eqref{eq:spatial-loading-decay}, we obtain
\[
    \prod_{a=1}^d |a_{\ell i_a}|
    \le
    A^d
    \rho^{\sum_{a=1}^d |i_a-c_\ell|}
    \le
    A^d\rho^D,
\]
where the last inequality uses $\rho\in(0,1)$. Therefore
\[
    |(\mathcal K_d)_i|
    \le
    \sum_{\ell=1}^L
    |\tau_{d,\ell}|
    \prod_{a=1}^d |a_{\ell i_a}|
    \le
    L
    \max_{1\le \ell\le L}|\tau_{d,\ell}|
    A^d
    \rho^D.
\]
By the definition of $\beta_{\mathrm{src}}$,
\[
    L
    \max_{1\le \ell\le L}|\tau_{d,\ell}|
    A^d
    \rho^D
    \le
    \beta_{\mathrm{src}}
    (1+D)^{-(\alpha+d-1)}.
\]
Since $i\in[p]^d$ was arbitrary, $\mathcal K_d\in\mathcal K_{\alpha,\beta_{\mathrm{src}}}^d$.
\end{proof}

\section{Technical Lemmas}

\begin{lemma} \label{lemma_balanced_sign_packing} There exist absolute constants $k_0\ge 2$, $c_{\rm code}>0$, and $\rho_0\in(0,1)$ such that, for every integer $k\ge k_0$, there exists a set $\mathcal E_k\subset\{-1,1\}^k$ satisfying 
\[ |\mathcal E_k| \ge \exp(c_{\rm code}k), \] 
and, for every distinct $\varepsilon,\varepsilon'\in\mathcal E_k$, 
\[ \left| 1-\frac{2d_H(\varepsilon,\varepsilon')}{k} \right| \le \rho_0. \] 
\end{lemma}

\begin{proof}
We prove the result by the probabilistic method. Let
$N=\lfloor \exp(c k)\rfloor$, where $c>0$ is a sufficiently small absolute constant to be chosen. Draw
$\varepsilon^1,\ldots,\varepsilon^N$ independently and uniformly from
$\{-1,1\}^k$.

For any fixed pair $a\ne b$, the Hamming distance
$d_H(\varepsilon^a,\varepsilon^b)$ has distribution
$\operatorname{Binomial}(k,1/2)$. By Hoeffding's inequality,
\[
    \PP
    \left\{
    \left|
    \frac{d_H(\varepsilon^a,\varepsilon^b)}{k}
    -
    \frac12
    \right|
    >
    \frac14
    \right\}
    \le
    2\exp(-k/8).
\]
Therefore, by the union bound,
\[
\begin{split}
    \PP
    \left\{
    \exists\, a<b:
    \left|
    \frac{d_H(\varepsilon^a,\varepsilon^b)}{k}
    -
    \frac12
    \right|
    >
    \frac14
    \right\}
    &\le
    \binom{N}{2}2\exp(-k/8)  \\
    &\le
    N^2\exp(-k/8) \\
    &\le
    \exp(2ck-k/8).
\end{split}
\]
Choose $c<1/16$. Then $2c-1/8<0$, and hence the last probability is strictly smaller than $1$ for every $k\ge 1$. Thus, with positive probability, all pairs $a<b$ satisfy
\[
    \left|
    \frac{d_H(\varepsilon^a,\varepsilon^b)}{k}
    -
    \frac12
    \right|
    \le
    \frac14 .
\]
Fix such a realization and set
\[
    \mathcal E_k=\{\varepsilon^1,\ldots,\varepsilon^N\}.
\]
The same event rules out duplicate codewords, since duplicates would have Hamming distance zero and therefore would violate the preceding display. Hence
$|\mathcal E_k|=N$.

It remains only to lower bound $N$. Choose $k_0$ sufficiently large so that, for all $k\ge k_0$,
\[
    \lfloor \exp(c k)\rfloor
    \ge
    \exp(ck/2).
\]
Then for every $k\ge k_0$,
\[
    |\mathcal E_k|
    =
    N
    \ge
    \exp(ck/2).
\]
Thus the cardinality condition holds with $c_{\rm code}=c/2$.

Finally, for every distinct $\varepsilon,\varepsilon'\in\mathcal E_k$,
\[
    \left|
    1-\frac{2d_H(\varepsilon,\varepsilon')}{k}
    \right|
    =
    2\left|
    \frac{d_H(\varepsilon,\varepsilon')}{k}
    -
    \frac12
    \right|
    \le
    \frac12 .
\]
Therefore the lemma holds with $\rho_0=1/2$.
\end{proof}

\begin{lemma}[Block-diagonal tensor norm]
\label{lem:block-diagonal-tensor}
Let $d\ge2$.
Suppose $\mathcal A_1,\ldots,\mathcal A_L$ are order-$d$ tensors supported on disjoint coordinate blocks $I_1^d,\ldots,I_L^d$, where $I_1,\ldots,I_L\subset[p]$ are disjoint.
Then
\[
    \left\|
    \sum_{\ell=1}^L \mathcal A_\ell
    \right\|
    =
    \max_{1\le \ell\le L}
    \|\mathcal A_\ell\|.
\]
\end{lemma}

\begin{proof}
The lower bound follows by testing vectors supported on the block attaining the maximum.
For the upper bound, let $u_1,\ldots,u_d\in\mathbb S^{p-1}$ and define
\[
    a_{j\ell}
    =
    \|(u_j)_{I_\ell}\|_2,
    \qquad
    j\in[d],\quad \ell\in[L].
\]
Since the blocks are disjoint, we have
\[
    \sum_{\ell=1}^L a_{j\ell}^2
    \le
    \|u_j\|_2^2
    =
    1,
    \qquad j\in[d].
\]
Therefore,
\be\label{eq_lemma_1_1}
    \left|
    \left\langle
    \sum_{\ell=1}^L \mathcal A_\ell,
    u_1\otimes\cdots\otimes u_d
    \right\rangle
    \right|
    \le
    \max_{1\le \ell\le L}\|\mathcal A_\ell\|
    \sum_{\ell=1}^L
    \prod_{j=1}^d a_{j\ell}.
\ee
By Holder's inequality,
\[
    \sum_{\ell=1}^L
    \prod_{j=1}^d a_{j\ell}
    \le
    \prod_{j=1}^d
    \left(
    \sum_{\ell=1}^L a_{j\ell}^d
    \right)^{1/d}.
\]
Since $d\ge2$ and $a_{j\ell}\ge0$,
\[
    \left(
    \sum_{\ell=1}^L a_{j\ell}^d
    \right)^{1/d}
    \le
    \left(
    \sum_{\ell=1}^L a_{j\ell}^2
    \right)^{1/2}
    \le
    1.
\]
Thus
\[
    \sum_{\ell=1}^L
    \prod_{j=1}^d a_{j\ell}
    \le
    1. 
\]
Together with \eqref{eq_lemma_1_1}, it follows 
\[
    \left\|
    \sum_{\ell=1}^L \mathcal A_\ell
    \right\|
    \le
    \max_{1\le \ell\le L}
    \|\mathcal A_\ell\|.
\]
Combining the lower and upper bounds proves the lemma.
\end{proof}

\begin{lemma}[Corollary 1.4 in \cite{gotzeConcentrationInequalitiesPolynomials2021}]\label{lemma_subexpo_concentration}
     Let $X_1, \ldots, X_n$ be a set of independent, centered random variables with $\left\|X_i\right\|_{\Phi_\alpha} \leq M$ for some $\alpha \in(0,1]$. Let $a \in \mathbb{R}^n$. For any $t \geq 0$ it holds

$$
\mathbb{P}\left(\left|\sum_{i=1}^n a_i X_i\right| \geq t\right) \leq 2 \exp \left(-\frac{1}{C_\alpha} \min \left(\frac{t^2}{M^2\|a\|^2}, \frac{t^\alpha}{M^\alpha \max _i\left|a_i\right|^\alpha}\right)\right). 
$$

\end{lemma}

\begin{lemma}[Local empirical moment tensor bound]
\label{lem:local-empirical-moment-bound}
Let $X_1,\ldots,X_n$ be i.i.d. copies of a centered random vector
$X\in\R^p$, and let
\[
    \widehat{\bcM}_d
    =
    \frac1n\sum_{\ell=1}^n X_\ell^{\otimes d},
    \qquad
    \bcM_d=\E X^{\otimes d}.
\]
Assume that, for some $0<\gamma\le d$,
$
    \sup_{u\in\mathbb S^{p-1}}
    \|u^\top X\|_{\Phi_\gamma}
    \le K .
$
For an interval (consecutive integers) $I\subset[p]$, let $\Pi_I(\widehat{\bcM}_d-\bcM_d)$ be the tensor obtained by keeping only the entries of $\widehat{\bcM}_d-\bcM_d$ whose indices all belong to $I$.
Then for every $m\in[p]$ and every $x\ge0$, with probability at least $1-Ce^{-x}$,
\[
    \max_{I:\ |I|\le m}
    \|\Pi_I(\widehat{\bcM}_d-\bcM_d)\|
    \le
    C_{\gamma,d,K}
    \left\{
    \sqrt{\frac{m+\log p+x}{n}}
    +
    \frac{(m+\log p+x)^{d/\gamma}}{n}
    \right\},
\]
where the maximum is over all intervals $I\subset[p]$ with length at most $m$.
\end{lemma}

\begin{proof}
Denote
\[
    \mathcal E
    =
    \widehat{\bcM}_d-\bcM_d .
\]
Fix an interval $I\subset[p]$ with $|I|\le m$.
Let $\mathcal N_I$ be a $1/4$-net of $\mathbb S^{|I|-1}$ such that
\[
    |\mathcal N_I|
    \le
    9^{|I|}
    \le
    9^m.
\]

Since $\Pi_I(\mathcal E)$ is supported on $I^d$, its spectral norm can be computed over vectors supported on $I$:
\[
    \|\Pi_I(\mathcal E)\|
    =
    \sup_{u_1,\ldots,u_d\in\mathbb S^{|I|-1}}
    \left|
    \left\langle
    \Pi_I(\mathcal E),
    u_1\otimes\cdots\otimes u_d
    \right\rangle
    \right|.
\]

For fixed $u_2,\ldots,u_d$, the map
\[
    u_1\mapsto
    \langle \Pi_I(\mathcal E),
    u_1\otimes u_2\otimes\cdots\otimes u_d\rangle
\]
is linear. If $v_1\in\mathcal N_I$ satisfies $\|u_1-v_1\|_2\le\varepsilon$, then
\[
\begin{split}
    &\left|
    \left\langle
    \Pi_I(\mathcal E),
    u_1\otimes u_2\otimes\cdots\otimes u_d
    \right\rangle
    \right| \\
    &\le
    \left|
    \left\langle
    \Pi_I(\mathcal E),
    v_1\otimes u_2\otimes\cdots\otimes u_d
    \right\rangle
    \right|
    +
    \varepsilon
    \sup_{\tilde u_1\in\mathbb S^{|I|-1}}
    \left|
    \left\langle
    \Pi_I(\mathcal E),
    \tilde u_1\otimes u_2\otimes\cdots\otimes u_d
    \right\rangle
    \right|.
\end{split}
\]
Taking the supremum over $u_1$, we have
\[
    \sup_{u_1\in\mathbb S^{|I|-1}}
    \left|
    \left\langle
    \Pi_I(\mathcal E),
    u_1\otimes u_2\otimes\cdots\otimes u_d
    \right\rangle
    \right|
    \le
    \frac{1}{1-\varepsilon}
    \max_{v_1\in\mathcal N_I}
    \left|
    \left\langle
    \Pi_I(\mathcal E),
    v_1\otimes u_2\otimes\cdots\otimes u_d
    \right\rangle
    \right|.
\]
Applying this argument successively to $u_1,\ldots,u_d$ yields
\[
    \|\Pi_I(\mathcal E)\|
    \le
    (1-\varepsilon)^{-d}
    \max_{v_1,\ldots,v_d\in\mathcal N_I}
    \left|
    \left\langle
    \Pi_I(\mathcal E),
    v_1\otimes\cdots\otimes v_d
    \right\rangle
    \right|.
\]
Taking $\varepsilon=1/4$ gives
\[
    \|\Pi_I(\mathcal E)\|
    \le
    \left(\frac43\right)^d
    \max_{v_1,\ldots,v_d\in\mathcal N_I}
    \left|
    \left\langle
    \Pi_I(\mathcal E),
    v_1\otimes\cdots\otimes v_d
    \right\rangle
    \right|.
\]
Since $v_1,\ldots,v_d$ are supported on $I$,
\[
    \left\langle
    \Pi_I(\mathcal E),
    v_1\otimes\cdots\otimes v_d
    \right\rangle
    =
    \left\langle
    \mathcal E,
    v_1\otimes\cdots\otimes v_d
    \right\rangle .
\]
Hence
\be\label{eq_lemma_2_enet}
    \|\Pi_I(\mathcal E)\|
    \le
    C_d
    \max_{v_1,\ldots,v_d\in\mathcal N_I}
    \left|
    \left\langle
    \mathcal E,
    v_1\otimes\cdots\otimes v_d
    \right\rangle
    \right|,
\ee
where one may take $C_d=(4/3)^d$.

For fixed $v_1,\ldots,v_d\in\mathcal N_I$, define
\[
    Z_\ell
    =
    \prod_{a=1}^d v_a^\top X_\ell
    -
    \E
    \prod_{a=1}^d v_a^\top X,
    \qquad
    \ell=1,\ldots,n.
\]
Then
\[
    \left\langle
    \mathcal E,
    v_1\otimes\cdots\otimes v_d
    \right\rangle
    =
    \frac1n
    \sum_{\ell=1}^n Z_\ell.
\]
Since
\[
    \|v_a^\top X_\ell\|_{\Phi_\gamma}
    \le
    K,
    \qquad
    a=1,\ldots,d,
\]
the product satisfies
\[
    \left\|
    \prod_{a=1}^d v_a^\top X_\ell
    \right\|_{\Phi_{\gamma/d}}
    \le
    C_{\gamma,d}K^d.
\]
Therefore, by Lemma A.1 and A.3 in \cite{gotzeConcentrationInequalitiesPolynomials2021}, the centered variable $Z_\ell$ also satisfies
\[
    \|Z_\ell\|_{\Phi_{\gamma/d}}
    \le
    C_{\gamma,d,K}.
\]
By Lemma \ref{lemma_subexpo_concentration}, for every $t\ge0$,
\be\label{eq_lemma_2_prob_bound}
    \PP
    \left\{
    \left|
    \frac1n
    \sum_{\ell=1}^n Z_\ell
    \right|
    >
    C_{\gamma,d,K}
    \left(
    \sqrt{\frac{t}{n}}
    +
    \frac{t^{d/\gamma}}{n}
    \right)
    \right\}
    \le
    2e^{-t}.
\ee

Now we take the union bound.
For any fixed interval $I$ and any fixed tuple
$(v_1,\ldots,v_d)\in\mathcal N_I^d$, \eqref{eq_lemma_2_prob_bound} implies, for every $s\ge0$,
\[
    \PP
    \left\{
    \left|
    \left\langle
    \mathcal E,
    v_1\otimes\cdots\otimes v_d
    \right\rangle
    \right|
    >
    C_{\gamma,d,K}
    \left(
    \sqrt{\frac{s}{n}}
    +
    \frac{s^{d/\gamma}}{n}
    \right)
    \right\}
    \le
    2e^{-s}.
\]
Since $|\mathcal N_I|\le 9^m$, we have
\[
    |\mathcal N_I|^d
    \le
    9^{dm}
    =
    \exp(C_d m).
\]
Therefore, for fixed $I$,
\[
\begin{split}
    &\PP
    \left\{
    \max_{v_1,\ldots,v_d\in\mathcal N_I}
    \left|
    \left\langle
    \mathcal E,
    v_1\otimes\cdots\otimes v_d
    \right\rangle
    \right|
    >
    C_{\gamma,d,K}
    \left(
    \sqrt{\frac{s}{n}}
    +
    \frac{s^{d/\gamma}}{n}
    \right)
    \right\}  \\
    &\qquad\le
    2\exp(C_d m-s).
\end{split}
\]
By \eqref{eq_lemma_2_enet}, after enlarging the constant $C_{\gamma,d,K}$ if necessary,
\[
\begin{split}
    &\PP
    \left\{
    \|\Pi_I(\mathcal E)\|
    >
    C_{\gamma,d,K}
    \left(
    \sqrt{\frac{s}{n}}
    +
    \frac{s^{d/\gamma}}{n}
    \right)
    \right\}  \le
    2\exp(C_d m-s).
\end{split}
\]
Let
\[
    \mathcal I_m
    =
    \{I\subset[p]: I \text{ is an interval and } |I|\le m\}.
\]
Since $|\mathcal I_m|\le p^2$, another union bound gives
\[
\begin{split}
    &\PP
    \left\{
    \max_{I\in\mathcal I_m}
    \|\Pi_I(\mathcal E)\|
    >
    C_{\gamma,d,K}
    \left(
    \sqrt{\frac{s}{n}}
    +
    \frac{s^{d/\gamma}}{n}
    \right)
    \right\}  \\
    &\qquad\le
    2p^2\exp(C_d m-s).
\end{split}
\]
Now take
\[
    s
    =
    C_d' m+3\log p+x
\]
with $C_d'>C_d$ sufficiently large.
Then
\[
    2p^2\exp(C_d m-s)
    \le
    Ce^{-x}.
\]
Since $s\lesssim_d m+\log p+x$, we obtain, with probability at least $1-Ce^{-x}$,
\[
    \max_{I:\ |I|\le m}
    \|\Pi_I(\mathcal E)\|
    \le
    C_{\gamma,d,K}
    \left\{
    \sqrt{\frac{m+\log p+x}{n}}
    +
    \frac{(m+\log p+x)^{d/\gamma}}{n}
    \right\}.
\]
This proves the lemma.

\end{proof}

\begin{lemma}[Spectral norm of a locally supported tensor]
\label{lem:local-support-tensor-norm}
Let $d\ge2$.
Let $\mathcal A\in(\R^p)^{\otimes d}$ satisfy
\[
    \mathcal A_{\bi}=0
    \qquad
    \text{whenever }
    \operatorname{diam}(\bi)>m,
\]
and
\[
    \max_{\bi\in[p]^d}|\mathcal A_{\bi}|
    \le
    a.
\]
Then
\[
    \|\mathcal A\|
    \le
    C_d a m^{d/2}.
\]
\end{lemma}

\begin{proof}
We use two shifted partitions of $[p]$ into intervals of length at most $2m$.
Let
\[
    \mathcal I_0
    =
    \{[2jm+1,2(j+1)m]\cap[p]:j\in\mathbb Z\},
\]
and
\[
    \mathcal I_1
    =
    \{[(2j-1)m+1,(2j+1)m]\cap[p]:j\in\mathbb Z\}.
\]
Every set of indices with diameter at most $m$ is contained in one interval from either $\mathcal I_0$ or $\mathcal I_1$.
Indeed, if all indices lie in an integer interval $[r,r+m]$, then this interval is contained in one interval of one of the two shifted grids.

Decompose
\[
    \mathcal A
    =
    \mathcal A^{(0)}
    +
    \mathcal A^{(1)}
\]
as follows.
Assign an entry of $\mathcal A$ to $\mathcal A^{(0)}$ if all its coordinates are contained in some interval of $\mathcal I_0$.
Assign all remaining nonzero entries to $\mathcal A^{(1)}$.
By the covering property above, every nonzero entry assigned to $\mathcal A^{(1)}$ has all its coordinates contained in some interval of $\mathcal I_1$.

For each $s\in\{0,1\}$, the tensor $\mathcal A^{(s)}$ is block diagonal with respect to the disjoint intervals in $\mathcal I_s$.
By Lemma~\ref{lem:block-diagonal-tensor},
\[
    \|\mathcal A^{(s)}\|
    =
    \max_{I\in\mathcal I_s}
    \|\Pi_I(\mathcal A^{(s)})\|,
\]
where $\Pi_I$ denotes the projection that keeps only entries whose indices all lie in $I$.
For every $I\in\mathcal I_s$, we have $|I|\le 2m$.
Therefore,
\[
    \|\Pi_I(\mathcal A^{(s)})\|
    \le
    \|\Pi_I(\mathcal A^{(s)})\|_F
    \le
    a |I|^{d/2}
    \le
    C_d a m^{d/2}.
\]
Hence
\[
    \|\mathcal A^{(s)}\|
    \le
    C_d a m^{d/2},
    \qquad
    s=0,1.
\]
The triangle inequality gives
\[
    \|\mathcal A\|
    \le
    \|\mathcal A^{(0)}\|
    +
    \|\mathcal A^{(1)}\|
    \le
    C_d a m^{d/2}.
\]
This proves the lemma.
\end{proof}

\begin{lemma}[Uniform local empirical moment bound over intervals]
\label{lem:uniform-local-empirical-moment-all-orders}
Let $d\ge2$ be fixed, and let $X_1,\ldots,X_n$ be i.i.d. copies of a random vector $X\in\R^p$.
Assume that, for some $0<\gamma\le d$,
\[
    \sup_{u\in\mathbb S^{p-1}}
    \|u^\top X\|_{\Phi_\gamma}
    \le K.
\]
For $1\le r\le d$, let
\[
    \bcM_r=\E X^{\otimes r},
    \qquad
    \widehat{\bcM}_r=\frac1n\sum_{\ell=1}^n X_\ell^{\otimes r},
    \qquad
    \mathcal E_r=\widehat{\bcM}_r-\bcM_r.
\]
For an interval $I\subset[p]$, let $\Pi_I^{(r)}$ denote the order-$r$ tensor projection that keeps only entries whose indices all belong to $I$.
Define
\[
    \mathcal I_m
    =
    \{I\subset[p]: I \text{ is an interval and } |I|\le m\}.
\]
Then for every $1\le m\le p$ and every $x\ge0$, with probability at least $1-Ce^{-x}$,
\begin{equation}
\label{eq:uniform_local_moment_errors_all_orders_lemma}
    \max_{1\le r\le d}
    \max_{I\in\mathcal I_m}
    \|\Pi_I^{(r)}(\mathcal E_r)\|
    \le
    C_{\gamma,d,K}
    \left\{
    \sqrt{\frac{m+\log p+x}{n}}
    +
    \frac{(m+\log p+x)^{d/\gamma}}{n}
    \right\}.
\end{equation}
\end{lemma}

\begin{proof}
Fix $1\le r\le d$ and $I\in\mathcal I_m$.
For $u_1,\ldots,u_r\in\mathbb S^{p-1}$ supported on $I$, we have
\[
\begin{split}
    \left\langle
    \mathcal E_r,
    u_1\otimes\cdots\otimes u_r
    \right\rangle
    &=
    \frac1n\sum_{\ell=1}^n
    \left\langle
    X_\ell^{\otimes r}-\bcM_r,
    u_1\otimes\cdots\otimes u_r
    \right\rangle  \\
    &=
    \frac1n\sum_{\ell=1}^n
    \left\{
    \prod_{a=1}^r u_a^\top X_\ell
    -
    \E\prod_{a=1}^r u_a^\top X
    \right\}.
\end{split}
\]
By Lemma A.1 and Lemma A.3 in \cite{gotzeConcentrationInequalitiesPolynomials2021}, together with the assumption
\[
    \sup_{u\in\mathbb S^{p-1}}
    \|u^\top X\|_{\Phi_\gamma}
    \le K,
\]
we have
\[
    \left\|
    \prod_{a=1}^r u_a^\top X
    \right\|_{\Phi_{\gamma/r}}
    \le
    C_{\gamma,d,K}.
\]
Since Lemma~\ref{lemma_subexpo_concentration} is stated for Orlicz exponents in $(0,1]$, set
\[
    \eta_r=\min\{\gamma/r,1\}.
\]
Then
\[
    \left\|
    \prod_{a=1}^r u_a^\top X
    \right\|_{\Phi_{\eta_r}}
    \le
    C_{\gamma,d,K}.
\]
After centering, again by Lemma A.3 in \cite{gotzeConcentrationInequalitiesPolynomials2021},
\[
    \left\|
    \left\langle
    X^{\otimes r}-\bcM_r,
    u_1\otimes\cdots\otimes u_r
    \right\rangle
    \right\|_{\Phi_{\eta_r}}
    \le
    C_{\gamma,d,K}.
\]

Applying Lemma~\ref{lemma_subexpo_concentration} to the independent centered random variables
\[
    \left\langle
    X_\ell^{\otimes r}-\bcM_r,
    u_1\otimes\cdots\otimes u_r
    \right\rangle,
    \qquad \ell=1,\ldots,n,
\]
with $a_1=\cdots=a_n=1/n$ gives, for every $t\ge0$,
\[
\begin{split}
    \PP
    \left\{
    \left|
    \left\langle
    \mathcal E_r,
    u_1\otimes\cdots\otimes u_r
    \right\rangle
    \right|
    >t
    \right\}
    &\le
    2\exp\left[
    -\frac1{C_{\gamma,d,K}}
    \min\left\{
    nt^2,
    (nt)^{\eta_r}
    \right\}
    \right].
\end{split}
\]
Equivalently, after enlarging constants, for every $s\ge0$,
\begin{equation}
\label{eq:fixed_direction_local_moment_tail}
    \PP
    \left\{
    \left|
    \left\langle
    \mathcal E_r,
    u_1\otimes\cdots\otimes u_r
    \right\rangle
    \right|
    >
    C_{\gamma,d,K}
    \left(
    \sqrt{\frac{s}{n}}
    +
    \frac{s^{1/\eta_r}}{n}
    \right)
    \right\}
    \le
    2e^{-s}.
\end{equation}

For each interval $I$, let $\mathcal N_I$ be a fixed $1/4$-net of the unit sphere in $\R^I$.
We may choose it so that
\[
    |\mathcal N_I|
    \le
    9^{|I|}
    \le
    9^m.
\]
For an order-$r$ tensor $\mathcal A$ supported on $I^r$, the standard net argument gives
\[
    \|\mathcal A\|
    \le
    \left(\frac43\right)^r
    \max_{u_1,\ldots,u_r\in\mathcal N_I}
    \left|
    \left\langle
    \mathcal A,
    u_1\otimes\cdots\otimes u_r
    \right\rangle
    \right|.
\]
Since $r\le d$, this implies
\[
    \|\mathcal A\|
    \le
    C_d
    \max_{u_1,\ldots,u_r\in\mathcal N_I}
    \left|
    \left\langle
    \mathcal A,
    u_1\otimes\cdots\otimes u_r
    \right\rangle
    \right|.
\]
Applying this to $\Pi_I^{(r)}(\mathcal E_r)$ and using \eqref{eq:fixed_direction_local_moment_tail}, we obtain, for fixed $r$ and fixed $I$,
\[
    \PP
    \left\{
    \|\Pi_I^{(r)}(\mathcal E_r)\|
    >
    C_{\gamma,d,K}
    \left(
    \sqrt{\frac{s}{n}}
    +
    \frac{s^{1/\eta_r}}{n}
    \right)
    \right\}
    \le
    2|\mathcal N_I|^r e^{-s}.
\]
Since $|\mathcal N_I|^r\le 9^{rm}\le \exp(C_d m)$, it follows that
\[
    \PP
    \left\{
    \|\Pi_I^{(r)}(\mathcal E_r)\|
    >
    C_{\gamma,d,K}
    \left(
    \sqrt{\frac{s}{n}}
    +
    \frac{s^{1/\eta_r}}{n}
    \right)
    \right\}
    \le
    2\exp(C_d m-s).
\]

There are at most $p^2$ intervals in $\mathcal I_m$ and $d$ possible values of $r$.
Taking a union bound gives
\[
\begin{split}
    &\PP
    \left\{
    \exists\,1\le r\le d,\ I\in\mathcal I_m:
    \|\Pi_I^{(r)}(\mathcal E_r)\|
    >
    C_{\gamma,d,K}
    \left(
    \sqrt{\frac{s}{n}}
    +
    \frac{s^{1/\eta_r}}{n}
    \right)
    \right\}  \\
    &\qquad\le
    2d p^2\exp(C_d m-s).
\end{split}
\]
Choose
\[
    s=C_d'(m+\log p+x)
\]
with $C_d'$ sufficiently large.
Since $d$ is fixed, after enlarging constants if necessary,
\[
    2d p^2\exp(C_d m-s)
    \le
    Ce^{-x}.
\]
Thus, with probability at least $1-Ce^{-x}$, simultaneously for all $1\le r\le d$ and all $I\in\mathcal I_m$,
\[
    \|\Pi_I^{(r)}(\mathcal E_r)\|
    \le
    C_{\gamma,d,K}
    \left\{
    \sqrt{\frac{m+\log p+x}{n}}
    +
    \frac{(m+\log p+x)^{1/\eta_r}}{n}
    \right\}.
\]
Finally,
\[
    \frac1{\eta_r}
    =
    \max\{r/\gamma,1\}
    \le
    \frac d\gamma,
\]
because $r\le d$ and $\gamma\le d$.
Since $m+\log p+x\ge1$,
\[
    (m+\log p+x)^{1/\eta_r}
    \le
    (m+\log p+x)^{d/\gamma}.
\]
Therefore,
\[
    \max_{1\le r\le d}
    \max_{I\in\mathcal I_m}
    \|\Pi_I^{(r)}(\mathcal E_r)\|
    \le
    C_{\gamma,d,K}
    \left\{
    \sqrt{\frac{m+\log p+x}{n}}
    +
    \frac{(m+\log p+x)^{d/\gamma}}{n}
    \right\}.
\]
This proves the lemma.
\end{proof}

\end{document}